\documentclass[final, a4paper]{amsart}

\usepackage[utf8]{inputenc}

\usepackage{amsthm, amssymb}

\usepackage[foot]{amsaddr}

\usepackage{thmtools, thm-restate}

\allowdisplaybreaks
    
\usepackage{hyperref}

\hypersetup{
    pdftitle   = {Matrix Product Operator Algebras II: Phases of Matter for 1D Mixed States},
    pdfauthor  = {Alberto Ruiz de Alarcon, Jose Garre-Rubio, Andras Molnar, David Perez-Garcia},
    colorlinks = true,
    linkcolor  = violet,
    citecolor  = violet,
    urlcolor   = violet,
}

\usepackage[noabbrev, nameinlink, capitalise]{cleveref}

\creflabelformat{equation}{#2#1#3}

\crefname{Definitions}{Definitions}{Definitions}

\usepackage{tikz}

\usetikzlibrary{%
    arrows,%
    decorations.text,%
    arrows.meta,%
    decorations.markings,%
    perspective,%
    3d,%
    positioning,
    shapes.geometric,
    calc
}

\newcommand{\myTrLength}{0.25}

\tikzset{ baseline = { ([yshift=-.5ex]current bounding box.center) }, %
          every picture/.style = { line width= 0.45pt } }

\tikzstyle{myMiddleArrow} = [ decoration={markings,                   %
        mark = at position 0.5*\pgfdecoratedpathlength+0.4*2pt        %
        with \arrow{>[line width=0.4pt,length=1.75pt,width=2pt]} },   %
        postaction={decorate} ]

\tikzstyle{myMiddleArrowOUT} = [ decoration={markings,                   %
        mark = at position 0.5*\pgfdecoratedpathlength+0.4*2pt        %
        with \arrow{>[line width=0.4pt,length=1.75pt,width=2pt]} },   %
        postaction={decorate} ]

\tikzstyle{myMiddleArrowIN} = [ decoration={markings,                   %
        mark = at position 0.5*\pgfdecoratedpathlength+0.4*2pt        %
        with \arrow{>[line width=0.4pt,length=1.75pt,width=2pt]} },   %
        postaction={decorate} ]
               
\tikzstyle{myMiddleArrowR} = [ decoration={markings,                  %
        mark = at position 0.5*\pgfdecoratedpathlength+0.4*2pt        %
        with \arrow{<[line width=0.4pt,length=1.75pt,width=2pt]} },   %
        postaction={decorate} ]

\tikzstyle{myGateStyle} = [ draw=black, minimum height=0.5cm,
    minimum width=0.9cm, fill=white ]
    
\tikzstyle{myVirtual} = [ red, line width = 0.35pt ]

\tikzstyle{myWhiteLine} = [ white, line width = 1.75pt ]

\tikzstyle{myDot} = [ circle, draw = black, line width = 0.2pt,
    fill = black!65, inner sep = 1.1pt ]

\tikzstyle{myWhiteDot} = [ circle, draw = black,
    line width = 0.2pt, fill = white, inner sep = 1.1pt ]

\tikzstyle{myRedDot} = [ rectangle, draw = black,
    line width = 0.2pt, fill = red!60!white, inner sep = 1pt ]

\tikzstyle{myVioletDot} = [
    star, star points = 8, draw = black, line width = 0.05pt,
    fill = violet!30, inner sep = 0.7pt, ]

\newtheorem{theorem}{Theorem}[section]
\newtheorem{lemma}[theorem]{Lemma}
\newtheorem{proposition}[theorem]{Proposition}
\newtheorem{conjecture}[theorem]{Conjecture}

\theoremstyle{definition}
\newtheorem{definition}[theorem]{Definition}
\newtheorem{definitions}[theorem]{Definitions}
\newtheorem{example}[theorem]{Example}
\newtheorem{remark}[theorem]{Remark}

\title[Matrix Product Operator Algebras II]%
{Matrix Product Operator Algebras II:\\ Phases of Matter for 1D Mixed States}

\author[A. Ruiz-de-Alarcón]{Alberto Ruiz-de-Alarcón}
\address[A. Ruiz-de-Alarcón]{Departamento de Análisis Matemático y Matemática 
Aplicada, Universidad Complutense de Madrid, 28040 Madrid, Spain \& 
Instituto de Ciencias Ma\-te\-máti\-cas (CSIC-UAM-UC3M-UCM), 
C/ Nicolás Cabrera 13-15, Campus de Cantoblanco, 28049 Madrid, Spain}
\email{alberto.ruiz.alarcon@icmat.es}

\author[J. Garre-Rubio]{José Garre-Rubio}
\address[J. Garre-Rubio]{University of Vienna, Faculty of Mathematics, 
Oskar-Morgenstern-Platz 1, 1090 Vienna, Austria}
\email{jose.garre-rubio@univie.ac.at}

\author[A. Molnár]{András Molnár}
\address[A. Molnár]{University of Vienna, Faculty of Mathematics, 
Oskar-Morgenstern-Platz 1, 1090 Vienna, Austria}
\email{andras.molnar@univie.ac.at}

\author[D. Pérez-García]{David Pérez-García}
\address[D. Pérez-García]{Departamento de Análisis Matemático y Matemática 
Aplicada, Universidad Complutense de Madrid, 28040 Madrid, Spain \& 
Instituto de Ciencias Ma\-te\-máti\-cas (CSIC-UAM-UC3M-UCM), 
C/ Nicolás Cabrera 13-15, Campus de Cantoblanco, 28049 Madrid, Spain}
\email{dperezga@ucm.es}

\date{\today}

\begin{document}

\begin{abstract}
The classification of topological phases of matter is fundamental to understand
and characterize the properties of quantum materials. In this paper we study
phases of matter in one-dimensional open quantum systems. We define two mixed
states to be in the same phase if both states can be transformed into the other
by a shallow circuit of local quantum channels. We aim to understand the emerging phase
diagram of matrix product density operators that are renormalization fixed
points. These states arise, for example, as boundaries of two-dimensional
topologically ordered states. We first construct families of such states based
on C*-weak Hopf algebras, the algebras whose representations form a fusion
category. More concretely, we provide explicit local fine-graining and local
coarse-graining quantum channels for the renormalization procedure of these
states. Finally, we prove that a subset of these states, those arising from
C*-Hopf algebras, are in the trivial phase.
\end{abstract}

\maketitle

\section{Introduction}

One of the main projects that quantum science is undertaking in the last decades
is the understanding and classification of exotic topological phases of quantum
matter. The approach to tackle this project is intrinsically connected to
quantum information theory. On the one hand,  topological phases of matter have
been identified as valuable resources in quantum computing 
\cite{freedman_2002_topological}. On the other, quantum information tools and
ideas are playing a key role in the classification program.

Before going any further, it is important to define what it means that two
systems belong to the same topological phase. Since topological properties have
an inherent global nature, the key idea is that their ground states display
similar global properties independently of their (possibly) different local
features. For instance, a ferromagnetic state
$\mid\uparrow\uparrow \cdots \uparrow\uparrow \rangle$
is topologically equivalent to an antiferromagnetic one
$\mid\uparrow \downarrow \cdots  \uparrow \downarrow \rangle$
since one can map locally one into the other, despite the fact that they have a
very different magnetization behaviour. 

A definition, motivated by quantum information, which tries to capture the
global properties, is the existence of a short-depth geometrically local quantum
circuit mapping one ground state into the other 
\cite{chen_2010_local,bravyi_2010_topological}. Using Hastings-Wen's
quasi-adiabatic evolution \cite{hastings_2005_quasiadiabatic} and Lieb-Robinson
bounds \cite{lieb_1972_finite} one can prove that this property is implied by
the more standard definition of phase based on the existence of a gapped path of
Hamiltonians connecting both systems \cite{bachmann_2012_automorphic}.  

The main advantage of the definition based on quantum circuits is that it 
focuses on states rather than on Hamiltonians, which is crucial to extend it to 
more general setups, like the one we are addressing here: open quantum systems.
However, this approach poses an additional problem: one has to identify the 
relevant class of states to classify. For closed quantum systems this relevant 
class is precisely the set of ground states of gapped short-range Hamiltonians.
Again quantum information theory provided us with a characterization of this 
set: ground states of short-ranged gapped Hamiltonians fulfill an area law for 
the entanglement between neighbouring regions, which implies that they are well 
approximated by ``tensor network states'', in particular by matrix product 
states (MPS) and projected-entagled pair states (PEPS) 
\cite{hastings_2007_area_law, anshu_2021_area, cirac_2021_matrix}. 

A natural approach to classify phases is to first restrict the classification to
``simple'' states that nevertheless are representatives for each phase. Since
topological properties are global, these representatives are taken to be
insensitive to real space renormalization steps (being those a finite depth
circuit), that is, they are renormalization fixed points (RFP). In 2D, for
instance, the string-net models of Levin and Wen \cite{levin_2005_stringnet} are
believed to provide a complete set of renormalization fixed points for
non-chiral 2D topological phases.

The restriction to RFPs has two important benefits. On the one hand, RFPs in
gapped phases have zero correlation length and thus they are \emph{exactly} MPS
and PEPS \cite{cirac_2021_matrix}; no approximation is needed. On the other
hand, it is easier to identify the key global invariants and thus identify the
different phases of RFP states.

These two points have been the crucial insights to successfully complete the
classification of 1D phases with symmetries, the so-called symmetry protected
topological (SPT) phases. Let us illustrate that this is the case by recalling
the steps that led to the classification of 1D SPT phases.
The first step was to prove that any MPS can be transformed into an RFP MPS in
the same phase \cite{schuch_2011_classif}. This restricts the classification
problem to just RFP MPS.
The second step was to identify the invariants of the phases using the
set of RFP MPS. These invariants are a set of quantities which, on the one hand,
are robust against short depth circuits and, on the other, are sufficient to
identify each phase uniquely. For SPT phases with unique ground state, the
invariants are the different equivalence classes of the second cohomology group
of the symmetry group \cite{chen_2011_classification, schuch_2011_classif}. For
SPT phases with symmetry breaking and therefore degenerate ground states, the
invariants are the different induced representations of the non-symmetry broken
subgroup together with its second cohomology group \cite{schuch_2011_classif}.
The third step was to prove that any two RFP MPS  that share the same invariants
can be mapped into each other with a short depth quantum circuit.
On top of that, a final and important step has been recently made: the
breakthrough results of Ogata \cite{ogata_2021_classif} show that one can even
extend these arguments beyond the framework of MPS to cover all  gapped
short-range Hamiltonians.

All the previous results stand for closed quantum systems, where the object of
interest is the ground state of a Hamiltonian. However, the question of
classifying phases is far from being answered for open quantum systems, even in
one dimension. Since isolation is never practically achieved, the
characterization of those systems play a fundamental role in real applications.

In this manuscript, we take the first steps towards the classification of open
quantum systems in 1D. A main difference between open and closed quantum systems
is that evolutions in closed quantum systems (either Hamiltonian evolution or
quantum circuits) are \emph{reversible}, whereas this is no longer true in open
quantum systems evolved under a Linbland master equation. For instance, if one
starts in a topologically ordered state, like the toric
code \cite{kitaev_2003_anyons}, one cannot find a short depth quantum circuit
mapping it into a product state. Short depth quantum circuits cannot create
\emph{or destroy} global correlations. However, local depolarizing noise can
convert the toric code (and indeed any topologically ordered state, no matter
how complex) into a product state in a short amount of
time \cite{coser_2019_classification}. Destroying global correlations is
therefore \emph{easy} in the open quantum systems regime. Constructing global
correlations is, on the other hand, still hard. In fact, local fast dissipative
evolutions cannot create global correlations \cite{konig_2014_generating}. This
shows that in the open quantum setting, phases should not be thought of as
classes of an equivalence relation, but rather as a partial order given by the
existence of a local fast dissipative evolution mapping one state into another
one. This partial order can also be understood as the complexity present in the
different topological phases. This proposal, due
to \cite{coser_2019_classification}, is the one we are taking here. Concretely,
we will say that a mixed state $\rho_1$ is more complex than another
one $\rho_2$ if there is a short-depth (geometrically local) circuit of quantum
channels, i.e. completely positive trace-preserving linear maps, mapping
$\rho_1$ into $\rho_2$ . 

There are several subtleties to make this definition formal. First of all, 
$\rho_1$ and $\rho_2$ should be well defined for all system size $n$. Second,
one should ask only for getting sufficiently close to $\rho_2$, allowing for
 both $\mathrm{polylog}(n)$ depth and $\mathrm{polylog}(n)$ locality in the
gates of the circuit.  Finally, one could take either a discrete point of view,
as here, or a continuous one, asking for a rapid mixing quasi-local Linbladian
evolution that approximates $\rho_2$ starting from $\rho_1$. Since in this paper
we are working only with RFP states, we will not need any of those subtleties
here and we refer to \cite{coser_2019_classification} for a detailed analysis 
of those.
 
We notice that there are other definitions of phases in the open quantum system
setting, like the works of Diehl et al. for Gaussian mixed states 
\cite{diehl_2011_topology,bardyn_2012_modes,bardyn_2013_topology} and for quasi
thermal states \cite{grusdt_2017_topological}, where the authors generalize the
notion of phases via gapped paths of Hamiltonians or via local unitary
transformations respectively. We refer also to \cite{coser_2019_classification}
for a detailed discussion about why the definition we are taking here seems more
appropriate. 

Encouraged by the successful  classification of pure states sketched
above, we will focus on RFP that are \emph{gapped} mixed states, that is, mixed
states which fulfill an area law for the mutual information. This is motivated
by two facts. On the one hand, it is known that Gibbs states of short-range
Hamiltonians fulfill an area law for the mutual information
\cite{wolf_2008_area}. On the other hand, it is known that fixed points of rapidly
mixing dissipative evolutions also fulfill an area law for the mutual
information \cite{brandao_2015_area}.

This naturally leads us to the set of RFP mixed states with a matrix product
density operator (MPDO) representation. The structure of RFP MPDOs has been
studied in detail in \cite{cirac_2017_mpdo} where, up to minor technical
conditions, the following is shown: (i) An MPDO is an RFP if there exist two
quantum channels $\mathfrak{T}$ and $\mathfrak{S}$ that implement the local
coarse graining and the local fine graining respectively, for which the given
MPDO is a fixed point. (ii) The RFP condition for MPDOs is characterized
operationally by the absence of length scales in the system; in particular by
having zero correlation length and saturation of the area law. (iii) The
existence of such $\mathfrak{T}$ and $\mathfrak{S}$ maps is equivalent to the
fact that from the MPDO an MPO algebra can be constructed.

This result brings the classification of 1D mixed states into the understanding
and classification of MPO algebras. Notably, MPO algebras are precisely the
mathematical objects behind the classification of RFP 2D topologically ordered
pure states in terms of PEPS \cite{sahinoglu_2021_char,bultinck_2017_anyons}.
This is not a lucky coincidence, but a consequence of the remarkable
bulk-boundary correspondence originated in the work of Li and
Haldane \cite{li_2008_entanglement}. In PEPS the bulk-boundary mapping is very
explicit \cite{cirac_2011_entanglement} and allows one to establish a dictionary
between bulk and boundary
properties \cite{schuch_2013_topological, kast_2019_local, perez_2020_estimates}.
Indeed, RFP MPDOs are expected to contain the set of boundary states associated
to RFP 2D non-chiral topologically ordered systems \cite{cirac_2017_mpdo}.

A throughout study of MPO algebras is done in the first paper of this
series \cite{molnar_2022_mpo}. There, it is shown that MPO algebras are closely
related to representations of semisimple finite-dimensional weak Hopf algebras,
which are, in turn, the algebraic description of fusion categories.


The paper is structured as follows.
In \cref{sec:preliminaries} we recall the basic notions and results on
weak Hopf algebras with a compatible C*-structure, called C*-weak Hopf algebras.
In this setting, we introduce the canonical regular element, which is
fundamental for our constructions. We also introduce the notion of biconnected
C*-WHA, whose representation categories are fusion categories. Moreover, we
recall their characterization as matrix product operators with a boundary, as
introduced in \cite{molnar_2022_mpo}.
In \cref{sec:rfp} we recall the definition of RFP MPDOs given in \cite{cirac_2017_mpdo}
and provide the construction of a family of RFP MPDOs arising from any given biconnected
C*-weak Hopf algebra. In particular, we provide explicit constructions of the local coarse-graining
and local fine-graining quantum channels $\mathfrak{T}$ and $\mathfrak{S}$ commented before.
In \cref{sec:boundaries} we describe the previous RFP MPDOs as the boundary
states of topological 2D PEPS. 
In \cref{sec:classif} we prove that the previous families of RFP MPDOs 
are in the trivial phase in the C*-Hopf algebra case, in the sense that they can be
obtained via a finite-depth and bounded-range circuit of quantum channels acting on
the maximally mixed state. Moreover, we show that this result can be
extended to the trivial sector of any biconnected C*-weak Hopf algebra.

\section{Preliminaries}
\label{sec:preliminaries}

In this section we collect elementary notions on algebras, coalgebras and
C*-weak Hopf algebras, as well as their representation theory in terms of matrix
product operators, recently developed in \cite{molnar_2022_mpo}. From now on,
we assume that all vector spaces are finite dimensional and their ground field
is the field of complex numbers $\mathbb{C}$. For any two vector spaces, we
denote by $\mathfrak{L}(V,W)$ the set of $\mathbb{C}$-linear maps from $V$ to
$W$ and let $\mathfrak{L}(V) :=\mathfrak{L}(V,V)$. We denote by
$V^*:=\mathfrak{L}(V,\mathbb{C})$ the dual vector space and by
$\langle\cdot,\cdot\rangle:V^*\times V\to\mathbb{C}$,
$(f,x)\mapsto \langle f,x\rangle = f(x)$ the canonical pairing.
An associative unital \emph{algebra} is a vector space $A$ endowed with an
associative linear map $A\otimes A\to A$, called multiplication, denoted
by juxtaposition, and an element $1 \in A$, called unit, satisfying
$1 x = x 1 = x$ for all elements $x \in A$. A unital \emph{C*-algebra} is an
algebra $A$ with an anti-linear involutive algebra anti-homomorphism
$(\,\cdot\,)^*:A \to A$, $x\mapsto x^*$, called \emph{$*$-operation}, and a
compatible Banach space structure. In this context, positive elements of $A$ are
elements of the form $x = y^*y$ for some element $y\in A$. As usually, the
multiplication, the unit element and the $*$-operation of two C*-algebras $A$
and $B$ are implicitly extended to their tensor product space ${A}\otimes {B}$
componentwise. Dually to the notion of algebra, a coassociative counital
\emph{coalgebra} is a vector space ${C}$ endowed with a linear map
$\Delta\in\mathfrak{L}(C,C\otimes C)$, called \emph{comultiplication}, such that
\begin{equation}
    (\mathrm{Id}\otimes \Delta)\circ \Delta =
    (\Delta\otimes\mathrm{Id}) \circ\Delta,
\end{equation}
and a linear functional $\varepsilon \in C^*$, known as \emph{counit},
compatible with the comultiplication in the sense that
$
    (\mathrm{Id}\otimes\varepsilon)\circ\Delta =
    (\varepsilon\otimes\mathrm{Id})\circ\Delta =
    \mathrm{Id}
$,
where we have identified $C\otimes\mathbb{C}\cong \mathbb{C}\otimes C \cong C$.
Henceforth, we drop the words associative, unital, coassociative and counital.
As usually done in the literature of coalgebras, we denote
\begin{equation}
    \Delta^{1}     := \Delta
    \;\,\text{ and }\;\,
    \Delta^{(n+1)} := ( \Delta \otimes \mathrm{Id}^{\otimes n} )
                      \circ \Delta^{(n)}
\end{equation}
for all $n\in\mathbb{N}$ and, complementarily, make use of Sweedler's notation
\begin{equation}
    x_{(1)}\otimes x_{(2)} \otimes\cdots \otimes x_{(n+1)} := \Delta^{(n)}(x)
\end{equation}
for all $x\in C$ and all $n\in\mathbb{N}$, ommiting the summation
symbol and cumbersome indices. In this context, an element $x\in C$ is
\emph{cocentral} if
\begin{equation}
    \label{eq:cocentral}
    x_{(1)} \otimes x_{(2)} = x_{(2)} \otimes x_{(1)}.
\end{equation}
i.e., its coproduct $\Delta(x)\in C\otimes C$ is invariant under the flip
operator. In addition, an element $x\in C$ is \emph{non-degenerate} if for all
$y\in C$ there exist $\phi,\psi\in C^*$ such that
\begin{equation}
    \label{eq:nondegenerate}
    \langle \phi, x_{(1)} \rangle x_{(2)}
    =
    y
    =
    x_{(1)} \langle \psi, x_{(2)} \rangle,
\end{equation}
roughly speaking, any element can be recovered from the coproduct
$\Delta(x)\in C\otimes C$ by applying an appropiate linear functional on any
of the cofactors.


In order to describe a sufficiently large family of renormalization fixed point
mixed states, e.g. boundary states of 2D string-net models, we will
introduce an algebraic construction that combines both structures of a
C*-algebra and a coalgebra. This is well-motivated from a representation
theoretical point of view \cite{nill_1998_axioms,bohm_1999_weak} and may look
odd from a purely algebraic one; thus, the subsequent axioms can be skipped in
a first reading. The following definition is due to G.~B{\"o}hm and
K.~Szlach{\'a}nyi.


\begin{definition}[see~\cite{bohm_1996_coassociative,nill_1998_axioms}]
\label{def:Cwha}
A \emph{C*-weak Hopf algebra} (C*-WHA) is a vector space ${A}$
simultaneously endowed with the structures of a C*-algebra and a coalgebra for
which the comultiplication is multiplicative, i.e.
\begin{equation}
    (xy)_{(1)} \otimes (xy)_{(2)} = x_{(1)} y_{(1)}\otimes x_{(2)}y_{(2)}
\end{equation}
for all elements $x,y\in {A}$;
the $*$-operation $*:{A}\to {A}$ is comultiplicative, i.e. 
\begin{equation}
    (x^*)_{(1)} \otimes (x^*)_{(2)} = (x_{(1)})^* \otimes (x_{(2)})^*
\end{equation}
for all elements $x\in {A}$; the counit $\varepsilon\in {A}^*$ is weakly
comultiplicative, i.e.
\begin{equation}
    \langle \varepsilon , xyz \rangle
    =
    \langle \varepsilon , xy_{(1)} \rangle \langle \varepsilon , y_{(2)}z \rangle
    =
    \langle \varepsilon , xy_{(2)} \rangle \langle \varepsilon , y_{(1)}z \rangle
\end{equation}
for all elements $x,y,z\in {A}$; the unit $1\in {A}$ is weakly comultiplicative,
i.e.
\begin{equation}
    1_{(1)}\otimes 1_{(2)} \otimes 1_{(3)}
    =
    1_{(1)} \otimes 1_{(2)} 1_{(1')} \otimes 1_{(2')}
    =
    1_{(1)} \otimes 1_{(1')} 1_{(2)} \otimes 1_{(2')},
\end{equation}
where the prime symbol distinguishes different coproducts of $1\in{A}$, and
there exists an anti-multiplicative linear map $S\in\mathfrak{L}({A})$ satisfying
\begin{equation}
    S(x_{(1)}) x_{(2)}
    =
    \langle\varepsilon,1_{(1)}x\rangle 1_{(2)}
    \;\,\text{ and }\;\,
    x_{(1)} S(x_{(2)})
    =
    1_{(1)} \langle\varepsilon,x 1_{(2)}\rangle
\end{equation}
for all elements $x\in {A}$, called antipode.
\end{definition}


\begin{remark}[see e.g.~Subsection 2.1~in~\cite{bohm_1999_weak}]
\label{remark:dualcwha}
The previous axioms are self-dual in the sense that for any C*-WHA $A$ its dual
vector space $A^*$ can be canonically endowed with the structure of a C*-WHA.
For simplicity, let us denote all structure maps in the same way. First, the
product of any two $\phi,\psi\in {A}^*$ is defined by the expression
$\phi\psi := (\phi\otimes\psi)\circ \Delta$, the unit element of $A^*$ is the
counit $\varepsilon\in {A}^*$ of ${A}$ and the $*$-operation is given by
$\langle \phi^* , x \rangle := \overline{ \langle \phi, S(x)^* \rangle }$ for
all $\phi\in A^*$ and all elements $x\in A$, where the bar denotes the complex
conjugate. The coalgebra structure is given via the comultiplication
$\langle\Delta(\phi),x\otimes y\rangle := \langle\phi,xy\rangle$
for all $\phi\in {A}^*$ and all elements $x,y\in {A}$, and the counit is the map
${A}^*\to \mathbb{C},\phi\mapsto \langle\phi,1\rangle$. Finally, the antipode is
defined by $S(\phi) := \phi\circ S$ for all $\phi\in {A}^*$.
\end{remark}


\begin{remark}[see~Lemma 2.8 and Theorem 2.10 in \cite{bohm_1999_weak}]
\label{remark:propsS}
In any C*-WHA the antipode is unique, invertible, both an algebra and a
coalgebra anti-homomorphism, and satisfies $S(x^*)^* = S^{-1}(x)$ for all
elements $x\in {A}$.
\end{remark}


\begin{definition}
\label{def:CHA}
A \emph{C*-Hopf algebra} (C*-HA) $A$ is a C*-WHA for which the comultiplication
$\Delta\in\mathfrak{L}(A,A\otimes A)$ is unit-preserving, i.e. $\Delta(1) = 1\otimes 1$.
Equivalently, the counit $\varepsilon\in A^*$ is an algebra homomorphism, i.e.
$\langle\varepsilon,xy\rangle = \langle\varepsilon,x\rangle\langle\varepsilon,y\rangle$ for all
 elements $x,y\in A$ and $\langle\varepsilon,1\rangle = 1$.
\end{definition}


\begin{example}
\label{ex:groupCHA}
The group C*-algebra $\mathbb{C}G$ of a finite group $G$ is endowed with the
structure of a C*-HA by the linear extensions of the maps given by the expressions
$\Delta(g) := g\otimes g$, $\langle\varepsilon,g\rangle := 1$ and
$S(g) := g^*:=g^{-1}$ for all elements $g\in G$.
\end{example}


\begin{example}
\label{ex:funcCHA}
The dual vector space $(\mathbb{C}G)^*$ is again a C*-algebra endowed with the
multiplication
$\langle\phi\psi,g\rangle := \langle\phi,g\rangle\langle\psi,g\rangle$, the unit
element $g\mapsto 1$ and the $*$-operation given by
$\langle\phi^*,g\rangle := \overline{\langle\phi,g\rangle}$,
for all $\phi,\psi\in(\mathbb{C}G)^*$ and all elements $g\in G$. Moreover, it
becomes a C*-HA too by virtue of the comultiplication
$\langle\Delta(\phi),g\otimes h\rangle := \langle\phi,gh\rangle$, the counit
$\langle\varepsilon,\phi\rangle := \langle\phi,1\rangle$, and the antipode
$\langle S(\phi),g\rangle := \langle\phi,g^{-1}\rangle$ for all
$\phi\in(\mathbb{C}G)^*$ and all elements $g,h\in G$.
\end{example}


The following example, due to G. I. Kac and V. G. Paljutkin, describes the
smallest C*-HA which is neither cocommutative, i.e. a group algebra,
nor commutative, i.e. the dual of a group algebra.

\begin{example}[see~\cite{kac_1966_finite}]\label{ex:H8}
Let ${H}_8$ be the C*-algebra generated by three elements $x$, $y$ and $z$
subject to the relations
$x^2 = 1$, 
$y^2 = 1$, 
$z^2 = 2^{-1}(1+x+y-xy)$,  
$xy=yx$, 
$zx = yz$, 
$zy = xz$, 
$x^* = x$, 
$y^* = y$ and $z^* = z^{-1}$.
It becomes a C*-HA by means of
$\Delta(x) := x\otimes x$, 
$\Delta(y) := y\otimes y$, 
$\Delta(z) := 2^{-1}(z\otimes z + yz\otimes z + z\otimes xz - yz\otimes xz)$, 
$   \langle\varepsilon,x\rangle   = \langle\varepsilon,y\rangle
  = \langle\varepsilon,z\rangle = 1$,
$S(x) := x$, $S(y) := y$ and $S(z) := z$.
\end{example}


The following example is the smallest proper C*-WHA. It is known as the Lee-Yang
C*-WHA as it is reconstructed from the solutions of the pentagon equation
arising from the Lee-Yang fusion rules.

\begin{example}[cf.~\cite{bohm_1996_coassociative}]
\label{ex:Fib1}
Let ${A}_{\mathrm{LY}}$ be the direct sum
$\mathfrak{M}(2,\mathbb{C})\oplus \mathfrak{M}(3,\mathbb{C})$ of full-matrix
$2\times 2$ and $3\times 3$ C*-al\-ge\-bras with complex coefficients, respectively.
Let $\zeta\in\mathbb{R}$ be the unique positive solution to $z^4 + z^2 -1 = 0$
and fix matrix units $e_{1}^{ij}$, $i,j=1,2$, in $\mathfrak{M}(2,\mathbb{C})$ and
$e_2^{k\ell}$, $k,\ell=1,2,3$, in $\mathfrak{M}(3,\mathbb{C})$. Then,
the comultiplication of $A_{\mathrm{LY}}$ is defined by the expressions
\begin{align*}
    \Delta(e_{1}^{11}) := { } & { }%
                    e_{1}^{11} \otimes e_{1}^{11} +
                    e_{2}^{11} \otimes e_{2}^{22}, \\
    \Delta(e_{1}^{12}) := { } & { }%
                    e_{1}^{12} \otimes e_{1}^{12} + 
            \zeta^2 e_{2}^{12} \otimes e_{2}^{21} +
            \zeta   e_{2}^{13} \otimes e_{2}^{23}, \\
    \Delta(e_{1}^{22}) := { } & { }%
                    e_{1}^{22} \otimes e_{1}^{22} +
            \zeta^4 e_{2}^{22} \otimes e_{2}^{11} + \\ & 
            \zeta^3 e_{2}^{23} \otimes e_{2}^{13} +
            \zeta^3 e_{2}^{32} \otimes e_{2}^{31} +
            \zeta^2 e_{2}^{33} \otimes e_{2}^{33}, \\
    \Delta(e_{2}^{11}) := { } & { }%
                    e_{1}^{11} \otimes e_{2}^{11} +
                    e_{2}^{11} \otimes e_{1}^{22} +
                    e_{2}^{11} \otimes e_{2}^{33}, \\
    \Delta(e_{2}^{12}) := { } & { }%
                    e_{1}^{12} \otimes e_{2}^{12} +
                    e_{2}^{12} \otimes e_{1}^{21} +
                    e_{2}^{13} \otimes e_{2}^{32}, \\
    \Delta(e_{2}^{13}) := { } & { }%
                    e_{1}^{12} \otimes e_{2}^{13} +
                    e_{2}^{11} \otimes e_{1}^{22} +
            \zeta   e_{2}^{12} \otimes e_{2}^{31} -
            \zeta^2 e_{2}^{13} \otimes e_{2}^{33}, \\
    \Delta(e_{2}^{22}) := { } & { }%
                    e_{0}^{22} \otimes e_{2}^{22} +
                    e_{2}^{22} \otimes e_{0}^{11} +
                    e_{2}^{33} \otimes e_{2}^{22}, \\
    \Delta(e_{2}^{23}) := { } & { }%
                    e_{1}^{22} \otimes e_{2}^{23} +
                    e_{2}^{23} \otimes e_{1}^{21} +
            \zeta   e_{2}^{32} \otimes e_{2}^{21} -
            \zeta^2 e_{2}^{33} \otimes e_{2}^{23}, \\
    \Delta(e_{2}^{33}) := { } & { }%
                    e_{1}^{22} \otimes e_{2}^{33} + 
                    e_{2}^{33} \otimes e_{1}^{22} + 
            \zeta^2 e_{2}^{22}\otimes e_{2}^{11} - \\ &
            \zeta^3 e_{2}^{23}\otimes e_{2}^{13} -
            \zeta^3 e_{2}^{32}\otimes e_{2}^{31} +
            \zeta^4 e_{2}^{33}\otimes e_{2}^{33}
\end{align*}
and the counit $\varepsilon\in (A_{\mathrm{LY}})^*$ and the antipode
$S\in\mathfrak{L}(A_{\mathrm{LY}})$ are given by
\begin{equation*}
    \langle\varepsilon,e_{1}^{ij}\rangle = 1,
    \quad
    \langle\varepsilon,e_{2}^{k\ell}\rangle = 0,
    \quad
    S(e_{1}^{ij}) = e_{1}^{ji}
    \;\,\text{ and }\;\,
    S(e_{2}^{k\ell}) = \zeta^{\ell-k} e_{2}^{\sigma(\ell)\sigma(k)}
\end{equation*}
for all $i,j\in\{1,2\}$ and $k,\ell\in\{ 1,2,3\}$, where $\sigma(1) := 2$,
$\sigma(2) := 1$, $\sigma(3) := 3$, endowing $A_{\mathrm{LY}}$ with the structure of a C*-WHA.
This specification has been slightly adapted from \cite{bohm_1996_coassociative}
as we will propose a tensor network description in \cref{ex:Fib2} consistent
with its string-net model definition.
\end{example}


A \emph{representation} of a C*-WHA ${A}$ is simply a representation of its
underlying C*-algebra, i.e. a couple $(V,\Phi)$ where $V$ is a
finite-dimensional complex vector space and $\Phi\in\mathfrak{L}(A,\mathfrak{L}(V))$ is an algebra
homomorphism. If, in addition, $V$ is a Hilbert space and
$\Phi(x^*) = \Phi(x)^\dagger$ for all $x\in {A}$, it is said to be a $*$-representation.
A representation is \emph{faithful} if the map $\Phi$ is injective.
Two representations $(V,\Phi)$ and $(W,\Psi)$ are \emph{equivalent} if
there is a bijective linear map $F\in\mathfrak{L}(V,W)$ such that
$\Psi(x) = F\circ \Phi(x) \circ F^{-1}$ for all elements $x\in A$.
Since $A$ is, in particular, a finite dimensional C*-algebra, the set
$\mathrm{Irr}({A})$ of equivalence classes of irreducible
re\-pre\-sen\-ta\-tions, also called \emph{sectors}, is necessarily finite.
In what follows, we fix a complete set $\{(V_\alpha,\Phi_\alpha):\alpha\in\mathrm{Irr}({A})\}$
of pairwise non-equivalent irreducible $*$-re\-pre\-sen\-ta\-tions of ${A}$
and let  $\mathrm{Tr}_\alpha := \mathrm{Tr}\circ\Phi_\alpha\in{A}^*$ stand for their
corresponding characters. On account of self-duality, $\mathrm{Irr}({A}^*)$
are labels for the sectors of the dual C*-WHA ${A}^*$, and we fix another complete set
$\{(W_a,\Psi_a):a\in\mathrm{Irr}({A}^*)\}$ of pairwise non-equivalent
irreducible $*$-re\-pre\-sen\-ta\-tions of ${A}^*$. Let
$\mathrm{Tr}^a := \mathrm{Tr}\circ \Psi_a \in{A}^{**}\cong {A}$
stand for their characters.


\begin{remark}[see~e.g.~\cite{etingof_2005_on_fusion,nill_1998_axioms}]
The category of $*$-representations of a
C*-WHA ${A}$ has, by construction, the structure of a rigid monoidal category. 
The comultiplication $\Delta:{A}\to {A}\otimes {A}$ provides the
monoidal product
\begin{equation*}
V \boxtimes W
:= \{z\in V\otimes W: \Delta(1)z=z\},
\quad
\Phi \boxtimes \Psi := (\Phi\otimes\Psi)\circ\Delta,
\end{equation*}
of any two $*$-representations $(V,\Phi)$ and $(W,\Psi)$ of ${A}$ while the counit
ensures the existence of a monoidal unit, called the
trivial representation; see \cite{bohm_1999_weak,bohm_2000_weak2}.
\end{remark}


The trivial representation has the unusual feature that it can be reducible.
This motivates the following definition.

\begin{definitions}\label{defs:conn}
A C*-WHA is said to be \emph{connected} if its trivial representation is
irreducible, \emph{coconnected} if the dual C*-WHA is connected, and
\emph{biconnected} if it is both connected and coconnected.
\end{definitions}


For the sake of simplicity, we assume from now on that $(V_1,\Phi_1)$
(resp. $(W_1,\Psi_1)$) corresponds to the trivial representation of ${A}$ if is
connected (resp. coconnected).


\begin{definitions}
\label{defs:ALARAmin}
Let ${A}$ be a C*-WHA. Then,
\begin{align*}
    {A}_L & := \{ x \in {A} :
                    x_{(1)}\otimes x_{(2)}
                    = x 1_{(1)}\otimes 1_{(2)}
                    = 1_{(1)}x \otimes 1_{(2)} \},\\
    {A}_R & := \{ y \in {A} :
                    y_{(1)}\otimes y_{(2)}
                    = 1_{(1)}\otimes y 1_{(2)}
                    = 1_{(1)} \otimes 1_{(2)}y \},
\end{align*}
are two commuting $*$-subalgebras of ${A}$, known as the target and source
counital subalgebras of ${A}$, respectively. Moreover,
$ {A}_{\mathrm{min}} := {A}_L {A}_R\subseteq {A} $
is the minimal C*-weak Hopf $*$-algebra contained in $A$ containing the unit
element. It is said to be \emph{minimal} if ${A} = {A}_\mathrm{min}$ and
\emph{regular} if the squared antipode restricted to
${A}_{\mathrm{min}}$ is the identity, i.e.
$S^2 \upharpoonright {A}_\mathrm{min} = \mathrm{Id}$.
\end{definitions}


For any connected C*-WHA $A$, its Gro\-then\-dieck ring $\mathfrak{K}_0(A)$,
i.e. the free $\mathbb{Z}$-module generated by the characters of representations
of $A$ with addition and multiplication defined accordingly, is a fusion ring
\cite{etingof_2005_on_fusion}. In particular, the characters
$\{\mathrm{Tr}_\alpha\in A^*:\alpha\in\mathrm{Irr}(A)\}$ form a basis satisfying
\begin{equation*}
    \mathrm{Tr}_\alpha\cdot \mathrm{Tr}_\beta
    =
    \sum_{\gamma} N_{\alpha\beta}^\gamma \mathrm{Tr}_\gamma
\end{equation*}
for some $N_{a b}^c\in\mathbb{N}\cup\{0\}$, for all sectors
$\alpha,\beta,\gamma\in\mathrm{Irr}(A)$. Hence, for any $*$-re\-pre\-sen\-ta\-tion $(V,\Phi)$
of $A$ we can expand its character in the form
$\mathrm{Tr}_V = \sum_{\gamma}\nu_\gamma\mathrm{Tr}_\gamma$, where
$\nu_a\in\mathbb{N}\cup\{0\}$ is the multiplicity of $(V_\alpha,\Phi_\alpha)$ within
$(V,\Phi)$. In this context, define the
$|\mathrm{Irr}(A)|\times|\mathrm{Irr}(A)|$ matrix $N_V$ with coefficients
$({N}_V)_{\beta\gamma} := \sum_\alpha \nu_\alpha N_{\alpha\beta}^\gamma$, for any two sectors
$\beta,\gamma\in\mathrm{Irr}(A)$. Since $\mathfrak{K}_0(A)$ is, in particular, a
transitive ring \cite{etingof_2015_tensor}, $N_V$ is a matrix with strictly
positive entries. Thus, by virtue of the Frobenius-Perron theorem, the spectral
radius of $N_V$, denoted $d_V$, is an algebraically simple positive eigenvalue,
known as the Frobenius-Perron dimension of $(V,\Phi)$. Though it is not needed,
we remark that this notion coincides with the one of quantum dimension from the
category of $*$-representations;
see \cite[Proposition~8.23]{etingof_2005_on_fusion} for a rigorous statement.
For simplicity of notation, let $d_\alpha := d_{V_\alpha}$ for all sectors $\alpha\in\mathrm{Irr}(A)$.
Also, let $\mathfrak{D}^2 := \sum_\alpha d_\alpha^2$ denote the Frobenius-Perron dimension
of the algebra. Dually, if $A$ is coconnected, let
$\{d_a:a\in\mathrm{Irr}(A^*)\}$ denote the dual Frobenius-Perron
dimensions of $A$. It turns out that then
$ \sum_{\alpha} d_\alpha^2 =\sum_{a} d_a^2$ if $A$ is a biconnected
C*-WHA \cite{etingof_2015_tensor}.

Now we recall the following result, which proves the existence of a
distinguished element satisfying a ``pulling-through equation'' in each
connected C*-WHA. These properties turn out to be enough to understand the
properties of renormalization fixed points.


\begin{theorem}[cf.~\cite{molnar_2022_mpo}]
\label{thm:Omega}
Let ${A}$ be a biconnected C*-WHA. Then,
\begin{equation*}
    \Omega := \mathfrak{D}^{-2} \sum_{a} d_a \mathrm{Tr}^a \in A^{**} \cong A
\end{equation*}
is a cocentral non-degenerate positive idempotent, known as the canonical regular element of $A$.
Moreover, there exists a unique linear map $T\in\mathfrak{L}(A)$ such that
\begin{equation}\label{eq:pt}
    T(x)\Omega_{(1)}\otimes \Omega_{(2)} = \Omega_{(1)}\otimes x\Omega_{(2)}
\end{equation}
for all elements $x\in A$, usually referred to as pulling-through identity.
In particular, $T$ is an involutive algebra anti-homomorphism.
\end{theorem}


The canonical regular element is well-known in the literature of
$\mathbb{Z}_+$-rings and it satisfies an eigenvalue equation of the form
\begin{equation}
    \label{eq:OmegaEigen}
    \Omega \cdot \mathrm{Tr}^a
    =
    \mathrm{Tr}^a \cdot \Omega
    =
    d_a \Omega
\end{equation}
for all sectors $a \in\mathrm{Irr}({A}^*)$.
\cref{eq:pt} resembles the
characterization of left integrals in C*-WHAs; see e.g.~\cite{bohm_1999_weak}.
However, these notions coincide if and only if ${A}$ is a C*-HA. In this context,
$\Omega\in A$ is a well-known element in the literature, the Haar integral
$h\in{A}$, and the linear map $T\in\mathfrak{L}(A)$ coincides with the antipode.
In any case, it is convenient to rewrite \cref{eq:pt} as
\begin{equation}
    \label{eq:pt2}
    T(\Omega_{(1)})x\otimes \Omega_{(2)} = T(\Omega_{(1)})\otimes x\Omega_{(2)},
\end{equation}
for all elements $x\in A$. Both identities will be used interchangeably.


Let us interpret now representations of a given C*-WHA in terms of tensor
networks; see \cite{molnar_2022_mpo} for an exhaustive discussion. As usually
done in the tensor network literature we employ a graphical notation, briefly
described in \cref{sec:graphical}. 
Consider any sequence $\{(V_{[i]},\Phi_{[i]}):i\in\mathbb{N}\}$ of representations
of a C*-WHA $A$. It turns out that the endomorphisms
$
    (\Phi_{[1]} \otimes\cdots \otimes \Phi_{[k]})\circ\Delta^{(k-1)}(x)
$
can be described in terms of matrix product operators, for all elements $x\in A$.
More concretely, there exist a Hilbert space $W$ and tensors 
$
A_{[i]} \in \mathfrak{L}(W)\otimes \mathfrak{L}(V_{[i]})
$,
$i\in\mathbb{N}$,  independent of $x\in A$, such that one can write
\begin{equation}
    \label{eq:repphii}
    \begin{tikzpicture}[scale=0.5]
        \draw [red] (-90:1) arc (-90:-90+25:1);
        \draw [red] (-25:1) arc (-25:0:1);
        \draw [red,myMiddleArrow] (90+90:1) arc (90+90:90+180:1);
        \draw [red,myMiddleArrow] (0:1) arc (0:90:1);
        \draw [red,myMiddleArrow] (90:1) arc (90:180:1);
		\draw [draw=none,postaction={decorate,decoration={text along path, text align=center,text={...}}}] (-45-30:1+0.02) to [bend right=30] (-45+30:1+0.02);
	        \draw [myMiddleArrow] (0:1-0.75) -- (0:1); 
	        \draw [myMiddleArrow] (0:1) -- (0:1.75); 
	        \node [myDot] at (0:1) {};
	        \draw [myMiddleArrow] (90:1-0.75) -- (90:1); 
	        \draw [myMiddleArrow] (90:1) -- (90:1.75); 
	        \node [myDot] at (90:1) {};
	        \draw [myMiddleArrow] (-90:1-0.75) -- (-90:1); 
	        \draw [myMiddleArrow] (-90:1) -- (-90:1.75); 
	        \node [myDot] at (-90:1) {};
	    \node [myRedDot,label={[label distance=-0.5mm,rotate=90]above :$\scriptscriptstyle{b(x)}$}] at (90+90:1) {};
	    \node [inner sep=0pt] at (60:1+0.45) {$\scriptstyle{A_{[1]}}$};
	    \node [inner sep=0pt] at (-20:1+0.55) {$\scriptstyle{A_{[2]}}$};
	    \node [inner sep=0pt] at (-65:1+0.5) {$\scriptstyle{A_{[k]}}$};
	    \node [rotate=45] at (135:1+0.3) {$\color{gray}\scriptscriptstyle{W}$};
	    \node [rotate=180-45] at (225:1+0.3) {$\color{gray}\scriptscriptstyle{W}$};
    \end{tikzpicture}
    \equiv (\Phi_{[1]}\otimes \cdots\otimes \Phi_{[k]}) \circ \Delta^{(k-1)}(x)
\end{equation}
for some linear map $b\in\mathfrak{L}(A,W \otimes W^*)$, for all $k\in\mathbb{N}$.
We will usually restrict to the trans\-la\-tion-invariant case, for which $\Phi_{[1]} = \Phi_{[2]} = \cdots =: \Phi$ and $A_{[1]}=A_{[2]}=\cdots$. Notice that the physical indices,
associated to Hilbert spaces $V$ and $V^*$, are depicted by
black lines, while the virtual indices, associated to Hilbert
spaces $W$ and $W^*$, are depicted by red lines.
Thus, from now on, we will drop the labels, since no misunderstanding can arise.
For instance, we can express the multiplicativity of the coproduct (see
\cref{def:Cwha}) with this simplified graphical notation as 
\begin{equation}\label{eq:TN-multxy}
    \begin{tikzpicture}[scale=0.85*0.5]
        \draw [red] (-90:1) arc (-90:-90+25:1);
        \draw [red] (-25:1) arc (-25:0:1);
        \draw [red,myMiddleArrow] (90+90:1) arc (90+90:90+180:1);
        \draw [red,myMiddleArrow] (90-90:1) arc (90-90:90:1);
        \draw [red,myMiddleArrow] (90:1) arc (90:90+90:1);
		\draw [draw=none,postaction={decorate,decoration={text along path, text align=center,text={...}}}] (-45-30:1+0.02) to [bend right=30] (-45+30:1+0.02);
        \draw [red, myMiddleArrow] (-90:1.75) arc (-90:-90+34:1.75);
        \draw [red, myMiddleArrow] (-34:1.75) arc (-34:0:1.75);
        \draw [red,myMiddleArrow] (90+90:1.75) arc (90+90:90+180:1.75);
        \draw [red,myMiddleArrow] (90-90:1.75) arc (90-90:90:1.75);
        \draw [red,myMiddleArrow] (90:1.75) arc (90:90+90:1.75);
		\draw [draw=none,postaction={decorate,decoration={text along path, text align=center,text={...}}}] (-45-30:1.75+0.02) to [bend right=30] (-45+30:1.75+0.02);
	    \foreach \i in {90, 90-90, 90-180} {
	        \draw [myMiddleArrow] (\i:1-0.75) -- (\i:1); 
	        \draw [myMiddleArrow] (\i:1) -- (\i:1.75); 
	        \draw [myMiddleArrow] (\i:1.75) -- (\i:2.5); 
	        \node [myDot] at (\i:1) {};
	        \node [myDot] at (\i:1.75) {};
	    }
	    \node [myRedDot,label={[label distance=-0.6mm,rotate=90]above :$\scriptscriptstyle{b(x)}$}] at (90+90:1.75) {};
	    \node [myRedDot,label={[label distance=-0.6mm,rotate=90]above :$\scriptscriptstyle{b(y)}$}] at (90+90:1) {};
    \end{tikzpicture}
=
    \begin{tikzpicture}[scale=0.5]
        \draw [red] (-90:1) arc (-90:-90+25:1);
        \draw [red] (-25:1) arc (-25:0:1);
        \draw [red,myMiddleArrow] (90+90:1) arc (90+90:90+180:1);
        \draw [red,myMiddleArrow] (90-90:1) arc (90-90:90:1);
        \draw [red,myMiddleArrow] (90:1) arc (90:90+90:1);
		\draw [draw=none,postaction={decorate,decoration={text along path, text align=center,text={...}}}] (-45-30:1+0.02) to [bend right=30] (-45+30:1+0.02);
	    \foreach \i in {90, 90-90, 90-180} {
	        \draw [myMiddleArrow] (\i:1-0.75) -- (\i:1); 
	        \draw [myMiddleArrow] (\i:1) -- (\i:1.75); 
	        \node [myDot] at (\i:1) {};
	    }
	    \node [myRedDot,label={[label distance=-0.5mm,rotate=90]above :$\scriptscriptstyle{b(xy)}$}] at (90+90:1) {};
    \end{tikzpicture}~.
\end{equation}
for all elements $x,y\in A$.

We finish this section by reinterpreting the different
properties of the canonical regular element in graphical notation. 
First, it is easy to check by induction on $n\in\mathbb{N}$ that the fact that
$\Omega\in A$ is a cocentral element implies
\begin{equation}
\label{eq:Omegacoce}
    \Omega_{(1)}\otimes \Omega_{(2)}\otimes \cdots\otimes \Omega_{(n)}
    = 
    \Omega_{(\sigma(1))} \otimes \Omega_{(\sigma(2))}\otimes\cdots\otimes \Omega_{(\sigma(n))}
\end{equation}
for any shift permutation $\sigma$ of $\{1,\ldots,n\}$. In turn, this can be rephrased
as the trans\-la\-tio\-n-in\-va\-rian\-ce of the associated MPOs:
\begin{equation}
    \label{eq:TN-Omegacoce}
    \begin{tikzpicture}[scale=0.5]
        \draw [red] (-90:1) arc (-90:-90+25:1);
        \draw [red] (-25:1) arc (-25:0:1);
        \draw [red,myMiddleArrow] (90+90:1) arc (90+90:90+180:1);
        \draw [red,myMiddleArrow] (90-90:1) arc (90-90:90:1);
        \draw [red,myMiddleArrow] (90:1) arc (90:90+90:1);
		\draw [draw=none,postaction={decorate,decoration={text along path, text align=center,text={...}}}] (-45-30:1+0.02) to [bend right=30] (-45+30:1+0.02);
	    \foreach \i in {180, 90, 0, -90} {
	        \draw [myMiddleArrow] (\i:1-0.75) -- (\i:1); 
	        \draw [myMiddleArrow] (\i:1) -- (\i:1.75); 
	        \node [myDot] at (\i:1) {};
	    }
	    \node [myRedDot,rotate=45,label={[label distance=-2mm]above :$\scriptscriptstyle{b(\Omega)}$}] at (90+45:1) {};
    \end{tikzpicture}
    =
    \begin{tikzpicture}[scale=0.5]
        \draw [red] (-90:1) arc (-90:-90+25:1);
        \draw [red] (-25:1) arc (-25:0:1);
        \draw [red,myMiddleArrow] (90+90:1) arc (90+90:90+180:1);
        \draw [red,myMiddleArrow] (90-90:1) arc (90-90:90:1);
        \draw [red,myMiddleArrow] (90:1) arc (90:90+90:1);
		\draw [draw=none,postaction={decorate,decoration={text along path, text align=center,text={...}}}] (-45-30:1+0.02) to [bend right=30] (-45+30:1+0.02);
	    \foreach \i in {180, 90, 0, -90} {
	        \draw [myMiddleArrow] (\i:1-0.75) -- (\i:1); 
	        \draw [myMiddleArrow] (\i:1) -- (\i:1.75); 
	        \node [myDot] at (\i:1) {};
	    }
	    \node [myRedDot,rotate=-45,label={[label distance=-2mm]above :$\scriptscriptstyle{b(\Omega)}$}] at (45:1) {};
    \end{tikzpicture}
    =\cdots =
    \begin{tikzpicture}[scale=0.5]
        \draw [red] (-90:1) arc (-90:-90+25:1);
        \draw [red] (-25:1) arc (-25:0:1);
        \draw [red,myMiddleArrow] (90+90:1) arc (90+90:90+180:1);
        \draw [red,myMiddleArrow] (90-90:1) arc (90-90:90:1);
        \draw [red,myMiddleArrow] (90:1) arc (90:90+90:1);
		\draw [draw=none,postaction={decorate,decoration={text along path, text align=center,text={...}}}] (-45-30:1+0.02) to [bend right=30] (-45+30:1+0.02);
	    \foreach \i in {180, 90, 0, -90} {
	        \draw [myMiddleArrow] (\i:1-0.75) -- (\i:1); 
	        \draw [myMiddleArrow] (\i:1) -- (\i:1.75); 
	        \node [myDot] at (\i:1) {};
	    }
	    \node [myRedDot,rotate=135,label={[label distance=-2mm]above :$\scriptscriptstyle{b(\Omega)}$}] at (180+45:1) {};
    \end{tikzpicture}~.
\end{equation}
In order to interpret the action of the linear map $T\in\mathfrak{L}(A)$, first note that the linear map $A\to\mathfrak{L}(V^*)$, $x\mapsto (\Phi\circ T(x))^{\mathrm{t}}$ defines a representation of $A$, where $(\,\cdot\,)^\mathrm{t}$ stands for the transpose operation. As discussed below, it is not necessarily a $*$-re\-pre\-sen\-ta\-tion. By \cref{eq:repphii}, one can depict, e.g.
\begin{equation}\label{eq:TN-defT}
    \begin{tikzpicture}[scale=0.5]
        \draw [red,myMiddleArrow] (90+90:1) arc (90+90:90+180:1);
        \draw [red,myMiddleArrow] (90-90:1) arc (90-90:90:1);
        \draw [red,myMiddleArrow] (90:1) arc (90:90+90:1);
        \draw [red,myMiddleArrow] (90-180:1) arc (90-180:90-90:1);
	    \foreach \i in {90} {
	        \draw [myMiddleArrow] (\i:1-0) -- (\i:1-0.75); 
	        \draw [myMiddleArrow] (\i:1.75) -- (\i:1); 
	        \node [myWhiteDot] at (\i:1) {};
	    }
	    \foreach \i in {0, -90} {
	        \draw [myMiddleArrow] (\i:1-0.75) -- (\i:1); 
	        \draw [myMiddleArrow] (\i:1) -- (\i:1.75); 
	        \node [myDot] at (\i:1) {};
	    }
	    \node [myRedDot,label={[label distance=-0.25mm,rotate=90]above :$\scriptscriptstyle{b(x)}$}] at (90+90:1) {};
    \end{tikzpicture}
    \equiv
    \Phi^{\otimes 3}\circ (T\otimes \mathrm{id}\otimes\mathrm{id})\circ \Delta^{(2)}(x) 
\end{equation}
for all $x\in A$, for some white rank-four tensor, where all physical spaces in
the picture are $V$. For the sake of clarity, we have reversed the direction of
the physical arrows corresponding to the new tensor, since $T\in\mathfrak{L}(A)$
is an anti-multiplicative map.
With this notation, \cref{eq:pt,eq:pt2} can be interpreted respectively as follows:
\begin{equation}\label{eq:TN-pt}
    \begin{tikzpicture}[scale=0.5]
        \draw [red] (-90:1) arc (-90:-90+25:1);
        \draw [red] (-25:1) arc (-25:0:1);
        \draw [red,myMiddleArrow] (90+90:1) arc (90+90:90+180:1);
        \draw [red,myMiddleArrow] (90-90:1) arc (90-90:90:1);
        \draw [red] (90:1) arc (90:90+90:1);
		\draw [draw=none,postaction={decorate,decoration={text along path, text align=center,text={...}}}] (-45-30:1+0.02) to [bend right=30] (-45+30:1+0.02);
	    \foreach \i in {90, 90-90, 90-180} {
	        \draw [myMiddleArrow] (\i:1-0.75) -- (\i:1); 
	        \draw [myMiddleArrow] (\i:1) -- (\i:1.75); 
	        \node [myDot] at (\i:1) {};
	    }
	        \draw [myMiddleArrow] (180:1-0.75) -- (180:1); 
	        \draw [myMiddleArrow] (180:1) -- (180:1.75); 
	        \node [myDot] at (180:1) {};
	    \begin{scope}[xshift=-2.75cm]
            \draw [red,myMiddleArrow] (0:1) arc (0:90:1);
            \draw [red,myMiddleArrow] (-90:1) arc (-90:0:1);
	        \draw [myMiddleArrow] (0:1) -- (0:1-0.75); 
            \node [myWhiteDot] at (0:1) {};
	    \end{scope}
	    \node [myRedDot,rotate=45,label={[label distance=-2mm]above :$\scriptscriptstyle{b(\Omega)}$}] at (90+45:1) {};
    \end{tikzpicture}
=
    \begin{tikzpicture}[scale=0.8*0.5]
        \draw [red] (-90:1) arc (-90:-90+25:1);
        \draw [red] (-25:1) arc (-25:0:1);
        \draw [red,myMiddleArrow] (90+90:1) arc (90+90:90+180:1);
        \draw [red,myMiddleArrow] (90-90:1) arc (90-90:90:1);
        \draw [red] (90:1) arc (90:90+90:1);
		\draw [draw=none,postaction={decorate,decoration={text along path, text align=center,text={...}}}] (-45-30:1+0.02) to [bend right=30] (-45+30:1+0.02);
        \draw [red, myMiddleArrow] (-90:1.75) arc (-90:-90+34:1.75);
        \draw [red, myMiddleArrow] (-34:1.75) arc (-34:0:1.75);
        \draw [red,myMiddleArrow] (90-90:1.75) arc (90-90:90:1.75);
        \begin{scope}[yshift=+1.75cm]
        \draw [red,myMiddleArrow] (0,0) -- (-1,0);
        \end{scope}
        \begin{scope}[yshift=-1.75cm]
        \draw [red,myMiddleArrow] (-1,0) -- (0,0);
        \end{scope}
		\draw [draw=none,postaction={decorate,decoration={text along path, text align=center,text={...}}}] (-45-30:1.75+0.02) to [bend right=30] (-45+30:1.75+0.02);
	    \foreach \i in {90, 90-90, 90-180} {
	        \draw [myMiddleArrow] (\i:1-0.75) -- (\i:1); 
	        \draw [myMiddleArrow] (\i:1) -- (\i:1.75); 
	        \draw [myMiddleArrow] (\i:1.75) -- (\i:2.5); 
	        \node [myDot] at (\i:1) {};
	        \node [myDot] at (\i:1.75) {};
	    }
	        \draw [myMiddleArrow] (180:1-0.75) -- (180:1); 
	        \draw [myMiddleArrow] (180:1) -- (180:1.75); 
	        \node [myDot] at (180:1) {};
	    \node [myRedDot,rotate=45,label={[label distance=-2mm]above :$\scriptscriptstyle{b(\Omega)}$}] at (90+45:1) {};
    \end{tikzpicture}
    \;\,\text{ and }\;\,
    \begin{tikzpicture}[scale=0.5]
        \draw [red] (-90:1) arc (-90:-90+25:1);
        \draw [red] (-25:1) arc (-25:0:1);
        \draw [red,myMiddleArrow] (90+90:1) arc (90+90:90+180:1);
        \draw [red,myMiddleArrow] (90-90:1) arc (90-90:90:1);
        \draw [red] (90:1) arc (90:90+90:1);
		\draw [draw=none,postaction={decorate,decoration={text along path, text align=center,text={...}}}] (-45-30:1+0.02) to [bend right=30] (-45+30:1+0.02);
	    \foreach \i in {90, 90-90, 90-180} {
	        \draw [myMiddleArrow] (\i:1-0.75) -- (\i:1); 
	        \draw [myMiddleArrow] (\i:1) -- (\i:1.75); 
	        \node [myDot] at (\i:1) {};
	    }
	        \draw [myMiddleArrow] (180:1) -- (180:1-0.75); 
	        \draw [myMiddleArrow] (180:1+0.75) -- (180:1); 
	        \node [myWhiteDot] at (180:1) {};
	    \begin{scope}[xshift=-2.75cm]
            \draw [red,myMiddleArrow] (0:1) arc (0:90:1);
            \draw [red,myMiddleArrow] (-90:1) arc (-90:0:1);
	        \draw [myMiddleArrow] (0:1-0.75) -- (0:1); 
            \node [myDot] at (0:1) {};
	    \end{scope}
	    \node [myRedDot,rotate=45,label={[label distance=-2mm]above :$\scriptscriptstyle{b(\Omega)}$}] at (90+45:1) {};
    \end{tikzpicture}
=
    \begin{tikzpicture}[scale=0.8*0.5]
        \draw [red] (-90:1) arc (-90:-90+25:1);
        \draw [red] (-25:1) arc (-25:0:1);
        \draw [red,myMiddleArrow] (90+90:1) arc (90+90:90+180:1);
        \draw [red,myMiddleArrow] (90-90:1) arc (90-90:90:1);
        \draw [red] (90:1) arc (90:90+90:1);
		\draw [draw=none,postaction={decorate,decoration={text along path, text align=center,text={...}}}] (-45-30:1+0.02) to [bend right=30] (-45+30:1+0.02);
        \begin{scope}[yshift=+1.75cm]
        \draw [red,myMiddleArrow] (0,0) -- (-1,0);
        \end{scope}
        \begin{scope}[yshift=-1.75cm]
        \draw [red,myMiddleArrow] (-1,0) -- (0,0);
        \end{scope}
        \draw [red, myMiddleArrow] (-90:1.75) arc (-90:-90+34:1.75);
        \draw [red, myMiddleArrow] (-34:1.75) arc (-34:0:1.75);
        \draw [red,myMiddleArrow] (90-90:1.75) arc (90-90:90:1.75);
		\draw [draw=none,postaction={decorate,decoration={text along path, text align=center,text={...}}}] (-45-30:1.75+0.02) to [bend right=30] (-45+30:1.75+0.02);
	    \foreach \i in {90, 90-90, 90-180} {
	        \draw [myMiddleArrow] (\i:1-0.75) -- (\i:1); 
	        \draw [myMiddleArrow] (\i:1) -- (\i:1.75); 
	        \draw [myMiddleArrow] (\i:1.75) -- (\i:2.5); 
	        \node [myDot] at (\i:1) {};
	        \node [myDot] at (\i:1.75) {};
	    }
	        \draw [myMiddleArrow] (180:1) -- (180:1-0.75); 
	        \draw [myMiddleArrow] (180:1.75) -- (180:1); 
	        \node [myWhiteDot] at (180:1) {};
	    \node [myRedDot,rotate=45,label={[label distance=-2mm]above :$\scriptscriptstyle{b(\Omega)}$}] at (90+45:1) {};
    \end{tikzpicture}~.
\end{equation}
Note that, since $x\in A$ is arbitrary, we have dropped the tensor $b(x)$ and
reexpressed the identities using open boundary MPOs instead; see~\cite{molnar_2022_mpo}.

\section{Renormalization Fixed Point MPDOs}
\label{sec:rfp}

In this section we define a distinguished family of MPOs starting from a
biconnected C*-WHA and show that they are RFP MPDOs, as defined
in \cite{cirac_2017_mpdo}. More concretely, we provide explicit expressions for
both local coarse-graining and local fine-graining quantum channels
$\mathfrak{T}$ and $\mathfrak{S}$ for which the generating rank-four tensor is 
a fixed point under the corresponding induced flows, very much in the spirit of
standard renormalization spirit. The generating tensor of the RFP MPDOs is
obtained here by appropiately weighting the tensor from the original MPO
algebra, described in the previous section, obtaining:
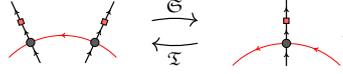
\begin{equation*}
    \begin{tikzpicture}
        \draw [myVirtual, myMiddleArrow] (45+20 : 1) arc (45+20 : 135-20 : 1);
        \draw [myVirtual] (45 : 1) arc (45 : 45+20 : 1);
        \draw [myVirtual] (135-20 : 1) arc (135-20 : 135 : 1);
        \draw [myMiddleArrow] (45+20 : 1-0.3) --(45+20 : 1);
        \draw [myMiddleArrow] (45+20 : 1)     --(45+20 : 1+0.3);
        \draw [myMiddleArrow] (45+20 : 1+0.3)     --(45+20 : 1+0.6);
        \draw [myMiddleArrow] (135-20 : 1-0.3) --(135-20 : 1);
        \draw [myMiddleArrow] (135-20 : 1)     --(135-20 : 1+0.3);
        \draw [myMiddleArrow] (135-20 : 1+0.3)     --(135-20 : 1+0.6);
        \node [myDot] at (45+20 : 1) {};
        \node [myDot] at (135-20 : 1) {};
        \node [myRedDot,rotate=-20] at (45+20 : 1.3) {};
        \node [myRedDot,rotate=20] at (135-20 : 1.3) {};
    \end{tikzpicture}
    \quad\;
    \begin{tikzpicture}
        \draw [->] (-0.3,0.15) to[bend left=10] (0.3,0.15);
        \draw [->] (0.3,-0.15) to[bend left=10] (-0.3,-0.15);
        \node [] at (0,0.35)  {$\scriptstyle{\mathfrak{S}}$};
        \node [] at (0,-0.35) {$\scriptstyle{\mathfrak{T}}$};
    \end{tikzpicture}
    \;\quad
    \begin{tikzpicture}
        \draw [myVirtual, myMiddleArrow] (90 : 1) arc (90 : 135 : 1);
        \draw [myVirtual, myMiddleArrow] (45 : 1) arc (45 : 90 : 1);
        \draw [myMiddleArrow] (90 : 1-0.3) --(90 : 1);
        \draw [myMiddleArrow] (90 : 1)     --(90 : 1+0.3);
        \draw [myMiddleArrow] (90 : 1+0.3) --(90 : 1+0.6);
        \node [myRedDot] at (90 : 1.3) {};
        \node [myDot] at (90 : 1) {};
    \end{tikzpicture}.
\end{equation*}
This weighting is done by means of the canonical regular element of the dual
C*-WHA. To this end, let us examine first the properties of this linear
functional, which formally plays the role in C*-WHAs of the character
of the usual left-regular representation.


\begin{restatable}{lemma}{restatableomega}
\label{lemma:omega}
Let $A$ be a connected C*-WHA. Then, the canonical regular element
$\omega\in A^*$ of the dual C*-WHA $A^*$ is the unique trace-like, faithful,
positive linear functional on $A$ that is idempotent, i.e.
$(\omega\otimes\omega)\circ\Delta = \omega$.
\end{restatable}


\begin{proof}
Recall \cref{thm:Omega} and \cref{remark:dualcwha}. It is easy to check that
$\omega\in A^*$ is a trace-like linear functional since it is a cocentral
element of $A^*$. Also, it is a faithful and positive linear functional
by construction. In addition, it satisfies the eigenvalue equation
$\mathrm{Tr}_\alpha \cdot \omega = \omega \cdot \mathrm{Tr}_\alpha = d_\alpha \omega$
for all sectors $\alpha\in\mathrm{Irr}(A)$; see~\cite[Section 3]{etingof_2015_tensor}
for a proof. We note that this implies, in particular, that $\omega\in A^*$ is
idempotent. Assume now that $f\in A^*$ is any linear functional satisfying the
properties above. Since it is trace-like, it can be expanded in the form
$f = \sum_\alpha f_\alpha \mathrm{Tr}_\alpha$ for some numbers $f_\alpha\in\mathbb{C}$,
$\alpha\in\mathrm{Irr}(A)$. By evaluating $f$ on the primitive central idempotents of $A$
it is easy to check that $f_\alpha > 0$ for all $\alpha\in\mathrm{Irr}(A)$, since $f$ is
assumed to be also a faithful positive linear functional. Define the
$|\mathrm{Irr}(A)|\times |\mathrm{Irr}(A)|$ matrix $N_f$ with complex
coefficients $(N_f)_{\beta\gamma} := \sum_\alpha f_\alpha N_{\alpha\beta}^{\gamma}$, which implements the
left-multiplication by $f\in A^*$ in the basis
$\{\mathrm{Tr}_\alpha:\alpha\in\mathrm{Irr}(A)\}$, i.e. it satisfies
$N_f \psi = f\psi$ for all $\psi\in A^*$. Then, $N_f f = f^2 = f$ and
$N_f \omega = f \omega = \sum_\alpha f_\alpha d_\alpha \omega \propto \omega$, where the first
equation holds since $f\in A^*$ is idempotent by hypothesis and the second
equation follows from the eigenvalue equation. Since $A$ is connected, the
Grothendieck ring $\mathfrak{K}_0(A)$ is, in particular, a transitive ring, and
hence $N_f$ has strictly positive entries; see
e.g.~\cite{nikshych_2004_semisimple} and \cite[Section 3]{etingof_2015_tensor}.
By virtue of the Frobenius-Perron theorem, it has only one eigenvector with
strictly positive entries, up to a constant. Therefore, $f = \omega$, since
both are positive idempotents.
\end{proof}


Now, given a faithful $*$-representation of the C*-WHA, we define the appropiate
weight extending the previous linear functional to the representation space.


\begin{remark}
\label{remark:phicwext}
Let $A$ be a connected C*-WHA and let $(V,\Phi)$ be a faithful
$*$-re\-pre\-sen\-ta\-tion of $A$. Let $b(f)$ denote the
boundary weight for the matrix product operators arising from the dual
C*-WHA $A^*$, for all $f\in A^*$. It turns out that 
$b(\omega) = \Phi(c_\omega)$ for some strictly positive central
element $c_\omega\in A$. It provides an extension of $\omega\in A^*$ to the
representation space $\mathfrak{L}(V)$ in the sense that
$$
    \mathrm{Tr}(  b(\omega) \Phi(x) )
    =
    \mathrm{Tr}(  \Phi(c_\omega x)  )
    =
    \langle \omega , x \rangle
$$
for all elements $x\in A$.
\end{remark}

\begin{proof}
For all sectors $\alpha\in\mathrm{Irr}(A)$, let $e_\alpha\in A$ be the corresponding primitive central
idempotent of $A$ and let $\nu_\alpha\in \mathbb{C}$ denote the multiplicity of $(V_\alpha,\Phi_\alpha)$ within
$(V,\Phi)$. Then, define the element
$ c_\omega := \mathfrak{D}^{-2} \sum_\alpha  d_\alpha \nu_\alpha^{-1}  e_\alpha \in A $.
Trivially, it is a central invertible positive element and satisfies
$
    \mathrm{Tr}(\Phi(c_\omega x)) =
    \mathfrak{D}^{-2} \sum_\alpha  \nu_\alpha^{-1} d_\alpha \mathrm{Tr}(\Phi(x e_\alpha)) 
    = 
    \mathfrak{D}^{-2}\sum_{\alpha} \nu_\alpha^{-1}  d_\alpha \nu_\alpha\mathrm{Tr}_\alpha(x)
    = \langle\omega,x\rangle
$
for all elements $x\in A$, as we wanted to prove.
\end{proof}

Let us now consider the tensor obtained by multiplying the MPO tensor in
\cref{eq:repphii} by $b(\omega) = \Phi(c_\omega)$ in the physical space:
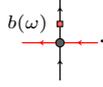
\begin{equation*}
    \begin{tikzpicture}[scale=0.5]
        \draw [] (0,0) -- (0,0.5); 
        \draw [myMiddleArrow] (0,0.5) -- (0,1.3); 
        \draw [myMiddleArrow] (0,-1) -- (0,0);
        \draw [myMiddleArrow,red] (0,0) -- (-1,0); 
        \draw [myMiddleArrow,red] (1,0) -- (0,0);
        \node [myDot] at (0,0) {};
        \node [myRedDot, label=left:$\scriptstyle{b(\omega)}$] at (0,0.5) {};
    \end{tikzpicture}.
\end{equation*}
Idempotence of $\omega\in A^*$ implies that this tensor generates an MPO with
zero correlation length; see~\cite{cirac_2017_mpdo}:
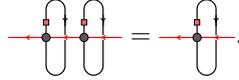
\begin{equation*}
    \begin{tikzpicture}[scale=0.5]
        \draw [rounded corners=0.125cm] (0,0) -- (0,1) -- (0.5,1) -- (0.5,0.25) -- (0.5,-1) -- (0,-1) -- (0,0);
        \draw [rounded corners=0.125cm] (1,0) -- (1,1) -- (1.5,1) -- (1.5,0.25) -- (1.5,-1) -- (1,-1) -- (1,0);
        \draw [myMiddleArrow] (0.5, 0.6) -- (0.5,0.3);
        \draw [myMiddleArrow] (1.5, 0.6) -- (1.5,0.3);
        \draw [myWhiteLine] (0.75,0)--(0.25,0);
        \draw [myWhiteLine] (1+0.75,0)--(1+0.25,0);
        \draw [myMiddleArrow,red] (0,0) -- (-1,0); 
        \draw [myMiddleArrow,red] (1,0) -- (0,0);
        \draw [myMiddleArrow,red] (2,0) -- (1,0);
        \node [myDot] at (0,0) {};
        \node [myDot] at (1,0) {};
        \node [myRedDot] at (0,0.4) {};
        \node [myRedDot] at (1,0.4) {};
    \end{tikzpicture}
    =
    \begin{tikzpicture}[scale=0.5]
        \draw [rounded corners=0.125cm] (0,0) -- (0,1) -- (0.5,1) --
            (0.5,0.25) -- (0.5,-1) -- (0,-1) -- (0,0);
        \draw [myMiddleArrow] (0.5, 0.6) -- (0.5,0.3);
        \draw [myWhiteLine] (0.75,0)--(0.25,0);
        \draw [myMiddleArrow,red] (0,0) -- (-1,0); 
        \draw [myMiddleArrow,red] (1,0) -- (0,0);
        \node [myDot] at (0,0) {};
        \node [myRedDot] at (0,0.4) {};
    \end{tikzpicture}.
\end{equation*}
It is clear that computations of correlation functions using the MPDOs generated
by the previous tensor will be length-independent. In particular, it induces the
following family of mixed states:


\begin{restatable}{theorem}{restatablethmrfp}
\label{thm:rfp}
Let $A$ be a biconnected C*-WHA and let
$(V,\Phi)$ be a faithful $*$-re\-pre\-sen\-ta\-tion of $A$. Then, the operators
$$
    \rho(x,n):= \langle\omega,x\rangle^{-1} b(\omega)^{\otimes n}
                \Phi^{\otimes n}( \Delta^{(n-1)}(x) )
    \in\mathfrak{L}(V^{\otimes n})
$$
are RFP MPDOs for all positive non-zero elements $x\in A$ and all $n\in\mathbb{N}$.
Specifically, there are quantum channels
$\mathfrak{T}:\mathfrak{L}(V)\to\mathfrak{L}(V\otimes V)$ and
$\mathfrak{S}:\mathfrak{L}(V\otimes V)\to\mathfrak{L}(V)$, known as
local fine-graining and coarse-graining maps, respectively, such that
$$
    \mathfrak{T}(\rho(x,1)) = \rho(x,2)
    \;\,\text{ and }\;\,
    \mathfrak{S}(\rho(x,2)) = \rho(x,1)
$$
for all positive non-zero elements $x\in A$ and all $n\in\mathbb{N}$.
\end{restatable}

Let us illustrate the construction with an extremely modest example.

\begin{example}\label{ex:Z2}
Let $A := \mathbb{C}\mathbb{Z}_2$ be the C*-HA arising from the
group $G:=\mathbb{Z}_2$ generated by $g\in G$; see \cref{ex:groupCHA}. 
It posseses only two sectors, namely the equivalence classes of the trivial
representation and the sign representation, each one-dimensional.
Consider that both physical and virtual spaces are $V := W:= \mathbb{C}^2$,
with basis elements $|1\rangle, |2\rangle$,
and consider the faithful $*$-representation of $A$
$\Phi\in\mathfrak{L}(A,\mathfrak{L}(\mathbb{C}^2))$
defined by
$\Phi(g):= \sigma_z$, the usual Pauli-Z matrix.
It is easy to see that both Frobenius-Perron dimensions are $1$ and hence the
canonical regular elements of $A$ and $A^*$ are given by $\Omega = 2^{-1}(e+g)$
and  $\omega(x) = (x,e)_V$, for all $x\in A$, respectively.
A tensor generating the corresponding MPOs is specified by the non-zero coefficients
\begin{equation*}
    \begin{tikzpicture}[scale=0.5]
        \draw [myMiddleArrow] (0,0) -- (0,0.51); 
        \draw [myMiddleArrow] (0,-0.51) -- (0,0);
        \draw [myMiddleArrow,red] (0,0) -- (-0.51,0); 
        \draw [myMiddleArrow,red] (0.51,0) -- (0,0);
        \node [myDot] at (0,0) {};
        \node at (-0.71,0) {$\scriptstyle{1}$};
        \node at (0.71,0) {$\scriptstyle{1}$};
        \node at (0,-0.71) {$\scriptstyle{1}$};
        \node at (0,0.71) {$\scriptstyle{1}$};
    \end{tikzpicture}
    = 
    \begin{tikzpicture}[scale=0.5]
        \draw [myMiddleArrow] (0,0) -- (0,0.51); 
        \draw [myMiddleArrow] (0,-0.51) -- (0,0);
        \draw [myMiddleArrow,red] (0,0) -- (-0.51,0); 
        \draw [myMiddleArrow,red] (0.51,0) -- (0,0);
        \node [myDot] at (0,0) {};
        \node at (-0.71,0) {$\scriptstyle{1}$};
        \node at (0.71,0) {$\scriptstyle{1}$};
        \node at (0,-0.71) {$\scriptstyle{2}$};
        \node at (0,0.71) {$\scriptstyle{2}$};
    \end{tikzpicture}
    =
    \begin{tikzpicture}[scale=0.5]
        \draw [myMiddleArrow] (0,0) -- (0,0.51); 
        \draw [myMiddleArrow] (0,-0.51) -- (0,0);
        \draw [myMiddleArrow,red] (0,0) -- (-0.51,0); 
        \draw [myMiddleArrow,red] (0.51,0) -- (0,0);
        \node [myDot] at (0,0) {};
        \node at (-0.71,0) {$\scriptstyle{2}$};
        \node at (0.71,0) {$\scriptstyle{2}$};
        \node at (0,-0.71) {$\scriptstyle{1}$};
        \node at (0,0.71) {$\scriptstyle{1}$};
    \end{tikzpicture}
    =-
    \begin{tikzpicture}[scale=0.5]
        \draw [myMiddleArrow] (0,0) -- (0,0.51); 
        \draw [myMiddleArrow] (0,-0.51) -- (0,0);
        \draw [myMiddleArrow,red] (0,0) -- (-0.51,0); 
        \draw [myMiddleArrow,red] (0.51,0) -- (0,0);
        \node [myDot] at (0,0) {};
        \node at (-0.71,0) {$\scriptstyle{2}$};
        \node at (0.71,0) {$\scriptstyle{2}$};
        \node at (0,-0.71) {$\scriptstyle{2}$};
        \node at (0,0.71) {$\scriptstyle{2}$};
    \end{tikzpicture}
    =1.
\end{equation*}
Moreover, in this case the weight is trivially given by $c_\omega = 2^{-1} e$ and thus 
$$
\rho(x,n) = \tfrac{1}{2^n}(\mathbf{1}^{\otimes n} + \tfrac{(x,g)_{V}}{(x,e)_V} \sigma_z^{\otimes n}),
$$
are the induced RFP MPDOs, for all positive non-zero $x\in A$.
In particular, $\rho(\Omega,n) =2^{-n} (\mathbf{1}^{\otimes n}+\sigma_z^{\otimes n})$ is the
boundary state of the toric code; see \cite{cirac_2017_mpdo}.
\end{example}

\begin{example}
\label{ex:Fib2}
Let $A_{\mathrm{LY}}$ be the Lee-Yang C*-WHA from
\cref{ex:Fib1}. It possesses only two sectors, denoted
$1$ and $\tau$, for which it is easy to check
that $d_1 = 1$ and $d_\tau = \zeta^{-2} = 2^{-1}(1+\sqrt{5})$, respectively. 
Consider that both physical and virtual spaces are $V:=W:=\mathbb{C}^5$ and let
$\Phi\in\mathfrak{L}(A_{\mathrm{LY}},\mathfrak{L}(\mathbb{C}^{5}))$
be the faithful $*$-representation arising from the string-net specification;
see \cite{molnar_2022_mpo,bultinck_2017_anyons} for a derivation.
A tensor generating the corresponding MPOs is then specified by the non-zero
coefficients
\begin{gather*}
\begin{tikzpicture}[scale=0.5]
\draw [myMiddleArrow] (0,0) -- (0,0.51); 
\draw [myMiddleArrow] (0,-0.51) -- (0,0);
\draw [myMiddleArrow,red] (0,0) -- (-0.51,0); 
\draw [myMiddleArrow,red] (0.51,0) -- (0,0);
\node [myDot] at (0,0) {};
\node at (-0.71,0) {$\scriptstyle{1}$};
\node at (0.71,0) {$\scriptstyle{1}$};
\node at (0,0.71) {$\scriptstyle{1}$};
\node at (0,-0.71) {$\scriptstyle{1}$};
\end{tikzpicture}
=
\begin{tikzpicture}[scale=0.5]
\draw [myMiddleArrow] (0,0) -- (0,0.51); 
\draw [myMiddleArrow] (0,-0.51) -- (0,0);
\draw [myMiddleArrow,red] (0,0) -- (-0.51,0); 
\draw [myMiddleArrow,red] (0.51,0) -- (0,0);
\node [myDot] at (0,0) {};
\node at (-0.71,0) {$\scriptstyle{1}$};
\node at (0.71,0) {$\scriptstyle{2}$};
\node at (0,0.71) {$\scriptstyle{3}$};
\node at (0,-0.71) {$\scriptstyle{3}$};
\end{tikzpicture}
=
\begin{tikzpicture}[scale=0.5]
\draw [myMiddleArrow] (0,0) -- (0,0.51); 
\draw [myMiddleArrow] (0,-0.51) -- (0,0);
\draw [myMiddleArrow,red] (0,0) -- (-0.51,0); 
\draw [myMiddleArrow,red] (0.51,0) -- (0,0);
\node [myDot] at (0,0) {};
\node at (-0.71,0) {$\scriptstyle{2}$};
\node at (0.71,0) {$\scriptstyle{1}$};
\node at (0,0.71) {$\scriptstyle{4}$};
\node at (0,-0.71) {$\scriptstyle{4}$};
\end{tikzpicture}
=
\begin{tikzpicture}[scale=0.5]
\draw [myMiddleArrow] (0,0) -- (0,0.51); 
\draw [myMiddleArrow] (0,-0.51) -- (0,0);
\draw [myMiddleArrow,red] (0,0) -- (-0.51,0); 
\draw [myMiddleArrow,red] (0.51,0) -- (0,0);
\node [myDot] at (0,0) {};
\node at (-0.71,0) {$\scriptstyle{2}$};
\node at (0.71,0) {$\scriptstyle{2}$};
\node at (0,0.71) {$\scriptstyle{2}$};
\node at (0,-0.71) {$\scriptstyle{2}$};
\end{tikzpicture}
=
\begin{tikzpicture}[scale=0.5]
\draw [myMiddleArrow] (0,0) -- (0,0.51); 
\draw [myMiddleArrow] (0,-0.51) -- (0,0);
\draw [myMiddleArrow,red] (0,0) -- (-0.51,0); 
\draw [myMiddleArrow,red] (0.51,0) -- (0,0);
\node [myDot] at (0,0) {};
\node at (-0.71,0) {$\scriptstyle{2}$};
\node at (0.71,0) {$\scriptstyle{2}$};
\node at (0,0.71) {$\scriptstyle{5}$};
\node at (0,-0.71) {$\scriptstyle{5}$};
\end{tikzpicture}
=
\begin{tikzpicture}[scale=0.5]
\draw [myMiddleArrow] (0,0) -- (0,0.51); 
\draw [myMiddleArrow] (0,-0.51) -- (0,0);
\draw [myMiddleArrow,red] (0,0) -- (-0.51,0); 
\draw [myMiddleArrow,red] (0.51,0) -- (0,0);
\node [myDot] at (0,0) {};
\node at (-0.71,0) {$\scriptstyle{3}$};
\node at (0.71,0) {$\scriptstyle{3}$};
\node at (0,0.71) {$\scriptstyle{2}$};
\node at (0,-0.71) {$\scriptstyle{1}$};
\end{tikzpicture}
=
\begin{tikzpicture}[scale=0.5]
\draw [myMiddleArrow] (0,0) -- (0,0.51); 
\draw [myMiddleArrow] (0,-0.51) -- (0,0);
\draw [myMiddleArrow,red] (0,0) -- (-0.51,0); 
\draw [myMiddleArrow,red] (0.51,0) -- (0,0);
\node [myDot] at (0,0) {};
\node at (-0.71,0) {$\scriptstyle{3}$};
\node at (0.71,0) {$\scriptstyle{4}$};
\node at (0,0.71) {$\scriptstyle{4}$};
\node at (0,-0.71) {$\scriptstyle{3}$};
\end{tikzpicture}
=
\begin{tikzpicture}[scale=0.5]
\draw [myMiddleArrow] (0,0) -- (0,0.51); 
\draw [myMiddleArrow] (0,-0.51) -- (0,0);
\draw [myMiddleArrow,red] (0,0) -- (-0.51,0); 
\draw [myMiddleArrow,red] (0.51,0) -- (0,0);
\node [myDot] at (0,0) {};
\node at (-0.71,0) {$\scriptstyle{3}$};
\node at (0.71,0) {$\scriptstyle{5}$};
\node at (0,0.71) {$\scriptstyle{5}$};
\node at (0,-0.71) {$\scriptstyle{3}$};
\end{tikzpicture}
=
\\
\begin{tikzpicture}[scale=0.5]
\draw [myMiddleArrow] (0,0) -- (0,0.51); 
\draw [myMiddleArrow] (0,-0.51) -- (0,0);
\draw [myMiddleArrow,red] (0,0) -- (-0.51,0); 
\draw [myMiddleArrow,red] (0.51,0) -- (0,0);
\node [myDot] at (0,0) {};
\node at (-0.71,0) {$\scriptstyle{4}$};
\node at (0.71,0) {$\scriptstyle{4}$};
\node at (0,0.71) {$\scriptstyle{1}$};
\node at (0,-0.71) {$\scriptstyle{2}$};
\end{tikzpicture}
=
\begin{tikzpicture}[scale=0.5]
\draw [myMiddleArrow] (0,0) -- (0,0.51); 
\draw [myMiddleArrow] (0,-0.51) -- (0,0);
\draw [myMiddleArrow,red] (0,0) -- (-0.51,0); 
\draw [myMiddleArrow,red] (0.51,0) -- (0,0);
\node [myDot] at (0,0) {};
\node at (-0.71,0) {$\scriptstyle{5}$};
\node at (0.71,0) {$\scriptstyle{4}$};
\node at (0,0.71) {$\scriptstyle{4}$};
\node at (0,-0.71) {$\scriptstyle{5}$};
\end{tikzpicture}
=
\begin{tikzpicture}[scale=0.5]
\draw [myMiddleArrow] (0,0) -- (0,0.51); 
\draw [myMiddleArrow] (0,-0.51) -- (0,0);
\draw [myMiddleArrow,red] (0,0) -- (-0.51,0); 
\draw [myMiddleArrow,red] (0.51,0) -- (0,0);
\node [myDot] at (0,0) {};
\node at (-0.71,0) {$\scriptstyle{5}$};
\node at (0.71,0) {$\scriptstyle{5}$};
\node at (0,0.71) {$\scriptstyle{2}$};
\node at (0,-0.71) {$\scriptstyle{2}$};
\end{tikzpicture}
=1 
,\;\,
\begin{tikzpicture}[scale=0.5]
\draw [myMiddleArrow] (0,0) -- (0,0.51); 
\draw [myMiddleArrow] (0,-0.51) -- (0,0);
\draw [myMiddleArrow,red] (0,0) -- (-0.51,0); 
\draw [myMiddleArrow,red] (0.51,0) -- (0,0);
\node [myDot] at (0,0) {};
\node at (-0.71,0) {$\scriptstyle{4}$};
\node at (0.71,0) {$\scriptstyle{5}$};
\node at (0,0.71) {$\scriptstyle{3}$};
\node at (0,-0.71) {$\scriptstyle{5}$};
\end{tikzpicture}
=
\begin{tikzpicture}[scale=0.5]
\draw [myMiddleArrow] (0,0) -- (0,0.51); 
\draw [myMiddleArrow] (0,-0.51) -- (0,0);
\draw [myMiddleArrow,red] (0,0) -- (-0.51,0); 
\draw [myMiddleArrow,red] (0.51,0) -- (0,0);
\node [myDot] at (0,0) {};
\node at (-0.71,0) {$\scriptstyle{5}$};
\node at (0.71,0) {$\scriptstyle{3}$};
\node at (0,0.71) {$\scriptstyle{5}$};
\node at (0,-0.71) {$\scriptstyle{4}$};
\end{tikzpicture}
=\zeta
,\;\,
\begin{tikzpicture}[scale=0.5]
\draw [myMiddleArrow] (0,0) -- (0,0.51); 
\draw [myMiddleArrow] (0,-0.51) -- (0,0);
\draw [myMiddleArrow,red] (0,0) -- (-0.51,0); 
\draw [myMiddleArrow,red] (0.51,0) -- (0,0);
\node [myDot] at (0,0) {};
\node at (-0.71,0) {$\scriptstyle{4}$};
\node at (0.71,0) {$\scriptstyle{3}$};
\node at (0,0.71) {$\scriptstyle{3}$};
\node at (0,-0.71) {$\scriptstyle{4}$};
\end{tikzpicture}
= -
\begin{tikzpicture}[scale=0.5]
\draw [myMiddleArrow] (0,0) -- (0,0.51); 
\draw [myMiddleArrow] (0,-0.51) -- (0,0);
\draw [myMiddleArrow,red] (0,0) -- (-0.51,0); 
\draw [myMiddleArrow,red] (0.51,0) -- (0,0);
\node [myDot] at (0,0) {};
\node at (-0.71,0) {$\scriptstyle{5}$};
\node at (0.71,0) {$\scriptstyle{5}$};
\node at (0,0.71) {$\scriptstyle{5}$};
\node at (0,-0.71) {$\scriptstyle{5}$};
\end{tikzpicture}
=\zeta^2.
\end{gather*}
Finally, it is straightforward to check that
$\Phi(c_\omega) = 2(5+5^{1/2})^{-1}\mathbf{1}_2\oplus 5^{-1/2}\mathbf{1}_3$.
\end{example}


With the aim of giving explicit definitions of both quantum channels and
prove \cref{thm:rfp}, we introduce the following auxiliary result.

\begin{restatable}{lemma}{restatablelemmaxi}
\label{lemma:xi}
Let $A$ be a biconnected C*-WHA. There exists a unique element $\xi\in{A}$
such that $\langle \omega,\xi T(\Omega_{(1)})\rangle \Omega_{(2)} = 1$.
Furthermore, it satisfies the following properties:
\begin{enumerate}
    \item it is positive, invertible and
          $ \xi^{-1} = \langle \omega , \Omega_{(1)}    \rangle \Omega_{(2)}
                     = \langle \omega , T(\Omega_{(1)}) \rangle \Omega_{(2)}$;
    \item it is invariant under $T\in\mathfrak{L}(A)$, i.e. $T(\xi) = \xi$;
    \item it satisfies the relation $T(x)^* = \xi T(x^*)\xi^{-1}$
          for all elements $x\in A$;
    \item $\mathrm{Tr}_\alpha(\xi^{-1}) = d_\alpha\langle\omega,\Omega\rangle$
          for all sectors $\alpha\in\mathrm{Irr}(A)$;
    \item it can be decomposed as $\xi = \xi_L \xi_R$ for two positive
          elements $\xi_L\in{A}_L$ and $\xi_R = S(\xi_L) = S^{-1}(\xi_L)\in{A}_R$;
\end{enumerate}
Dually, if we denote $\hat{\xi} = \hat{\xi}_L \hat{\xi}_R \in A^*$, then:
\begin{enumerate}
    \setcounter{enumi}{5}
    \item $x_{(1)}\langle \hat{\xi}_L, x_{(2)}\rangle = \xi_L^{-1}x$ and
    $ x_{(1)}\langle \hat{\xi}_R^{-1},x_{(2)}\rangle = x \xi_L^{-1}$ for all $x\in A$.
\end{enumerate}
Finally, if $A$ is a C*-HA, then
$\xi_L^2 = \xi_R^2 = \xi = \mathfrak{D}^2\varepsilon(1)1 = \langle \omega,\Omega\rangle^{-1} 1$.
\end{restatable}


See \cref{appendix:lemmaxi} for a proof. The fundamental property of
the definition of $\xi\in A$ here, interpreted in terms of tensor networks,
is provided by the following result.


\begin{restatable}{lemma}{restatablelemmaopen}
\label{lemma:open}
Let $A$ be a biconnected C*-WHA. Then,
\begin{equation*}
\begin{tikzpicture}[rotate = 0]
    \draw [myWhiteLine] (-0.65,0)--(0.65,0);
    %
    %
    \draw [rounded corners = 0.5*\myTrLength cm]
        (0.5 * \myTrLength,1) --++ (0, 0.25 + \myTrLength)       --++
        (-\myTrLength,0) --++ (0, -0.25 - 2 * \myTrLength -0.5) --++
        (\myTrLength,0) --   (0.5 * \myTrLength,0.5);
    \draw [myMiddleArrow] (0.5*\myTrLength,0.5)--++(0,0.25+0.5);
    \node [inner sep = 1pt, label={[label distance=-2pt]180:$\scriptstyle{b(\omega)\Phi(\xi)}$}] at (-0.5*\myTrLength,0.125+0.5+0.25) {};
    \node [myVioletDot] at (-0.5*\myTrLength, 0.125+0.5+0.25-0.1) {};
    \node [myRedDot] at (-0.5*\myTrLength, 0.125+0.5+0.25+0.1) {};
    %
    %
    \draw [myWhiteLine]                  (0   : 0.5) arc (0   : 180 : 0.5);
    \draw [myMiddleArrow, myVirtual] (0   : 0.5) arc (0   : 38 : 0.5);
    \draw [myVirtual] (38  : 0.5) arc (38  : 90-14 : 0.5);
    \draw [myMiddleArrow, myVirtual] (90-14 : 0.5) arc (90-14 : 180 : 0.5);
    %
    %
    \begin{scope}[yshift = 1.75 cm]
        \draw [myWhiteLine]                  (180 : 0.5) arc (180 : 360 : 0.5);
        \draw [myMiddleArrow, myVirtual] (180 : 0.5) arc (180 : 270+14 : 0.5);
        \draw [myMiddleArrow, myVirtual] (270+14 : 0.5) arc (270+14 : 360 : 0.5);
    \end{scope}
    %
    %
    \node [myDot] at (90-14 : 0.5) {};
    \node [inner sep = 1pt, label={[label distance=-1pt]0 : $\scriptstyle{b(\Omega)}$}] at (45 : 0.5) {};
    \node [myRedDot, rotate=45] at (45 : 0.5) {};
    \begin{scope}[yshift = 1.75 cm]
        \node [myWhiteDot] at (270 + 14 : 0.5) {};
    \end{scope}
    %
    %
    \begin{scope}[yshift = 1.75 cm]
        \node [myVirtual] at (320 : 0.5+0.15) {$\scriptscriptstyle{b}$};
    \end{scope}
    \node [myVirtual] at (130 : 0.5-0.15) {$\scriptscriptstyle{a}$};
\end{tikzpicture}
=
\delta_{ab}
\begin{tikzpicture}[scale=1, rotate=0]
    \draw [myWhiteLine] (-0.65,0)--(0.65,0);
    %
    %
    \node [myVirtual] at (-0.35, 0.875) {$\scriptscriptstyle{a}$};
    \node [myVirtual] at ( 0.35, 0.875) {$\scriptscriptstyle{a}$};
    %
    %
    \draw [myVirtual, myMiddleArrow, rounded corners] (-0.5, 1.75) %
        to ($ (0,1.75) + (180+20:0.5) $) %
        to [bend left=40] (180-20:0.5) %
        to (-0.5, 0);
    %
    %
    \draw [myVirtual, myMiddleArrowR, rounded corners] (+0.5, 1.75) %
        to ($ (0,1.75) + (-20:0.5) $) %
        to [bend right=40] (20:0.5) %
        to (0.5, 0);
\end{tikzpicture}
\end{equation*}
for all sectors $a,b\in\mathrm{Irr}({A})$, where $\delta_{ab}$ stands for the Kronecker delta.
\end{restatable}


\begin{proof}
Note that
\begin{align*}
\langle f,\Omega_{(1)}\rangle \langle \omega,\xi T(x)\Omega_{(2)}\rangle
 &= \langle f,T(\xi T(x))\Omega_{(1)}\rangle \langle\omega,\Omega_{(2)}\rangle
 \vphantom{\xi^{-1}}
 &&\text{by~\cref{eq:pt}}\\
 &= \langle f,x\xi\Omega_{(1)}\rangle \langle\omega,\Omega_{(2)}\rangle
 \vphantom{\xi^{-1}}
 &&\text{by~\cref{thm:Omega}}\\ 
 &= \langle f, x \xi \xi^{-1}\rangle = \langle f, x\rangle
 \vphantom{\xi^{-1}}
 && \text{by~\cref{lemma:xi}}
\end{align*}
for any two elements $x\in A$ and $f\in A^*$. Pictorially:
\begin{equation*}
\begin{tikzpicture}[rotate=90]
    \draw [white] (2.4,0)--(-1.2,0);
    \draw [white] (0,-0.965)--(0,0.965);
    \draw [rounded corners = 0.5*\myTrLength cm] (1,-0.5*\myTrLength) --++ (\myTrLength,0) --++ (0,\myTrLength) --++ (-0.5-2*\myTrLength,0) --++ (0,-\myTrLength) -- (0.5,-0.5*\myTrLength);
    \draw [myMiddleArrow] (0.5,-0.5*\myTrLength)--++(0.5,0);
    \begin{scope}[overlay]
        \node [inner sep = 1pt, label={[label distance=1pt]180:$\scriptstyle{b(\omega)\Phi(\xi)}$}] () at (0.5+\myTrLength,0.5*\myTrLength) {};
    \end{scope}
    \node [myVioletDot] () at (0.1+0.5+\myTrLength,0.5*\myTrLength) {};
    \node [myRedDot] () at (-0.1+0.5+\myTrLength,0.5*\myTrLength) {};
    \draw [rounded corners = 0.5*\myTrLength cm] (-0.5,-0.5*\myTrLength) --++ (-\myTrLength,0)
        --++ (0,\myTrLength) --++ (2*\myTrLength,0) --++ (0,-\myTrLength) -- cycle;
    \begin{scope}[]
        \node [myRedDot, label={[label distance=-2pt]270:$\scriptstyle{b(f)}$}] (nodeXi) at (-0.5-\myTrLength,0) {};
    \end{scope}
    \draw [myWhiteLine] (0:0.5) arc (0:90:0.5);
    \draw [myMiddleArrow,myVirtual] (0:0.5) arc (0:180:0.5);
    \draw [myMiddleArrow,myVirtual] (180:0.5) arc (180:270:0.5);
    \draw [myMiddleArrow,myVirtual] (-90:0.5) arc (-90:0:0.5);
    \begin{scope}[xshift=1.5cm]
        \draw [myWhiteLine] (90:0.5) arc (90:225:0.5);
        \draw [myMiddleArrow,myVirtual] (0:0.5) arc (0:180:0.5);
        \draw [myMiddleArrow,myVirtual] (180:0.5) arc (180:360:0.5);
    \end{scope}
    \node [myDot] at (-14:0.5) {};
    \node [myDot] at (180+14:0.5) {};
    \begin{scope}[overlay]
    \node [myRedDot, label={[label distance=-1pt]0:$\scriptstyle{b(\Omega)}$}] at (-90:0.5) {};
    \end{scope}
    \begin{scope}[xshift=1.5cm]
        \node [myWhiteDot] (nodeWhite) at (180+14:0.5) {};
        \node [myRedDot, label={[label distance=-2pt]90:$\scriptstyle{b(x)}$}] at (0:0.5) {};
    \end{scope}
\end{tikzpicture}
=
\begin{tikzpicture}[rotate=90]
    \draw [white] (2.4,0)--(-1.2,0);
    \draw [white] (0,-0.965)--(0,0.965);
    \draw [rounded corners = 0.5*\myTrLength cm] (0.5-0.075,-0.5*\myTrLength) --++ (\myTrLength,0)
            --++ (0,\myTrLength) --++ (-2*\myTrLength,0) --++ (0,-\myTrLength) -- cycle;
    \node [inner sep = 1pt, label={[label distance=0pt,rotate=90]0:$\scriptscriptstyle{b(\omega)\Phi(\xi)}$}] () at (0.5-0.075+\myTrLength,0) {};
    \node [myVioletDot,rotate=20] () at (0.5-0.075+\myTrLength-0.1,-0.12) {};
    \node [myRedDot,rotate=-20] () at (0.5-0.075+\myTrLength-0.1,0.12) {};
    \draw [rounded corners = 0.5*\myTrLength cm] (-0.5,-0.5*\myTrLength) --++ (-\myTrLength,0)
        --++ (0,\myTrLength) --++ (2*\myTrLength,0) --++ (0,-\myTrLength) -- cycle;
    \begin{scope}
        \node [myRedDot, label={[label distance=-2pt]270:$\scriptstyle{b(f)}$}] at (-0.5-\myTrLength,0) {};
    \end{scope}
    \draw [myWhiteLine] (10:0.5-0.075) arc (10:170:0.5-0.075);
    \draw [myMiddleArrow,myVirtual] (0:0.5-0.075) arc (0:180:0.5-0.075);
    \draw [myMiddleArrow,myVirtual] (180:0.5-0.075) arc (180:270:0.5-0.075);
    \draw [myMiddleArrow,myVirtual] (270:0.5-0.075) arc (270:360:0.5-0.075);
    \draw [myVirtual] (90-20:0.5+0.075) arc (90-20:270+20:0.5+0.075);
    \node [myDot] (nodeBlack) at (-15:0.5-0.075) {};
    \node [myDot] (nodeBlack) at (180+15:0.5-0.075) {};
    \node [myDot] (nodeBlack) at (180+12:0.5+0.075) {};
    \begin{scope}[overlay]
    \node [myRedDot,label={[label distance=2pt]0:$\scriptstyle{b(\Omega)}$}] (nodeRed) at (-90:0.5-0.075) {};
    \end{scope}
    \begin{scope}[xshift=1.5cm]
        \draw [myVirtual] (-90-20:0.5) arc (-90-20:110:0.5);
        \node [myRedDot, label={[label distance=-2pt]90:$\scriptstyle{b(x)}$}] at (0.5,0) {};
    \end{scope}
    \draw [myVirtual,myMiddleArrow] (1.5,0)++(110:0.5) to [bend left=20] (70:0.5+0.075);
    \draw [myVirtual,myMiddleArrowR] (1.5,0)++(-110:0.5) to [bend right=20] (-70:0.5+0.075);
\end{tikzpicture}
=
\begin{tikzpicture}[rotate=90]
    \draw [white] (2.4,0)--(-1.2,0);
    \draw [white] (0,-0.965)--(0,0.965);
    \draw [rounded corners = 0.5*\myTrLength cm] (-0.5,-0.5*\myTrLength) --++ (-\myTrLength,0)
        --++ (0,\myTrLength) --++ (2*\myTrLength,0) --++ (0,-\myTrLength) -- cycle;
    \begin{scope}[]
    \node [myRedDot, label={[label distance=-2pt]270:$\scriptstyle{b(f)}$}] at (-0.5-\myTrLength,0) {};
    \end{scope}
    \draw [myWhiteLine] (10:0.5-0.075) arc (10:170:0.5-0.075);
    \draw [myMiddleArrow,myVirtual] (0:0.5-0.075) arc (0:180:0.5-0.075);
    \draw [myMiddleArrow,myVirtual] (180:0.5-0.075) arc (180:270:0.5-0.075);
    \draw [myMiddleArrow,myVirtual] (270:0.5-0.075) arc (270:360:0.5-0.075);
    \draw [myVirtual] (90-20:0.5+0.075) arc (90-20:270+20:0.5+0.075);
    \node [myDot] (nodeBlack) at (180+15:0.5-0.075) {};
    \node [myDot] (nodeBlack) at (180+12:0.5+0.075) {};
    \begin{scope}[overlay]
    \node [myRedDot,label={[label distance=2pt]0:$\scriptstyle{b(1)}$}] (nodeRed) at (-90:0.5-0.075) {};
    \end{scope}
    \begin{scope}[xshift=1.5cm]
        \draw [myVirtual] (-90-20:0.5) arc (-90-20:110:0.5);
        \node [myRedDot, label={[label distance=-2pt]90:$\scriptstyle{b(x)}$}] at (0.5,0) {};
    \end{scope}
    \draw [myVirtual,myMiddleArrow] (1.5,0)++(110:0.5) to [bend left=20] (70:0.5+0.075);
    \draw [myVirtual,myMiddleArrowR] (1.5,0)++(-110:0.5) to [bend right=20] (-70:0.5+0.075);
\end{tikzpicture}
=
\begin{tikzpicture}[rotate=90]
    \draw [white] (2.4,0)--(-1.2,0);
    \draw [white] (0,-0.965)--(0,0.965);
    %
    \draw [rounded corners = 0.5*\myTrLength cm] (-0.5,-0.5*\myTrLength) --++ (-\myTrLength,0)
        --++ (0,\myTrLength) --++ (2*\myTrLength,0) --++ (0,-\myTrLength) -- cycle;
    \begin{scope}[]
    \node [myRedDot, label={[label distance=-2pt]270:$\scriptstyle{b(f)}$}] (nodeXi) at (-0.5-\myTrLength,0) {};
    \end{scope}
    \draw [myWhiteLine] (90-20:0.5) arc (90-20:180:0.5);
    \draw [myWhiteLine] (180:0.5) arc (180:270+20:0.5);
    \draw [myVirtual] (90-20:0.5) arc (90-20:180:0.5);
    \draw [myVirtual] (180:0.5) arc (180:270+20:0.5);
    \begin{scope}[xshift=1.5 cm]
        \draw [myWhiteLine] (-90:0.5) arc (-90:90:0.5);
        \draw [myVirtual] (0:0.5) arc (0:90+20:0.5);
        \draw [myVirtual] (270-20:0.5) arc (270-20:360:0.5);
        \node [myRedDot, label={[label distance=-2pt]90:$\scriptstyle{b(x)}$}] at (0.5,0) {};
    \end{scope}
    \draw [myVirtual,myMiddleArrow] (1.5,0)++(110:0.5) to [bend left=20] (70:0.5);
    \draw [myVirtual,myMiddleArrowR] (1.5,0)++(-110:0.5) to [bend right=20] (-70:0.5);
    \node [myDot] (nodeBlack) at (180+14:0.5) {};
\end{tikzpicture}
=
\begin{tikzpicture}[rotate=90]
    \draw [white] (2.4,0)--(-1.2,0);
    \draw [white] (0,-0.965)--(0,0.965);
    \draw [myWhiteLine] (90-20:0.5) arc (90-20:180:0.5);
    \draw [myWhiteLine] (180:0.5) arc (180:270+20:0.5);
    \draw [myVirtual] (90-20:0.5) arc (90-20:180:0.5);
    \draw [myVirtual] (180:0.5) arc (180:270+20:0.5);
    \begin{scope}[xshift=1.5 cm]
        \draw [myWhiteLine] (-90:0.5) arc (-90:90:0.5);
        \draw [myVirtual] (0:0.5) arc (0:90+20:0.5);
        \draw [myVirtual] (270-20:0.5) arc (270-20:360:0.5);
        \node [myRedDot, label={[label distance=-2pt]90:$\scriptstyle{b(x)}$}] at (0.5,0) {};
    \end{scope}
    \draw [myVirtual,myMiddleArrow] (1.5,0)++(110:0.5) to [bend left=20] (70:0.5);
    \draw [myVirtual,myMiddleArrowR] (1.5,0)++(-110:0.5) to [bend right=20] (-70:0.5);
    \node [myRedDot, label={[label distance=-2pt]270:$\scriptstyle{\Psi(f)}$}] (nodeXi) at (-0.5,0) {};
\end{tikzpicture}
\end{equation*}
The results follows from the surjectivity of $b$ and $\Psi$ in each block.
\end{proof}


We are now in the position to partially prove that the MPOs generated by the
MPO tensor presented above are RFP.

\begin{proof}[Proof of \cref{thm:rfp}]
Define the map $\mathfrak{T}\in\mathfrak{L}(\mathfrak{L}(V),\mathfrak{L}(V\otimes V))$ by
$$ X\mapsto 
    \begin{tikzpicture}[scale=0.5, baseline={(0,-.5ex)} ]
        \draw [myMiddleArrow] (180:1+0.75) -- (180:1); 
        \draw [rounded corners=0.125cm] (180:1) --++ (0.75+0.1,0) --++ (0,-0.5) --++ (-2.55-0.25-0.1+0.3,0) --++ (0,0.5) -- (180:1+0.75);
        \node [myVioletDot,label={[label distance=-1.211mm]above  :$\scriptscriptstyle{\;\;\Phi(\xi)}$}] at (180:1-0.5) {};
        \draw [myWhiteLine] (-160:1) arc (-160:-100:1); 
        \draw [red,myMiddleArrow] (-45:1) arc (-45:45:1);
        \draw [red,myMiddleArrow] (45:1) arc (45:180:1);
        \draw [red,myMiddleArrow] (-180:1) arc (-180:-90:1);
        \draw [red] (-90:1) arc (-90:-45:1);
	    \foreach \i in {45, -45} {
	        \draw [myMiddleArrow] (\i:1-0.75) -- (\i:1);  
	        \draw [myMiddleArrowOUT] (\i:1) -- (\i:1.75); 
	        \node [myDot] at (\i:1) {};
            \node [myRedDot,rotate=45] at (\i:1+0.3) {};
	    }
	        \node [myWhiteDot] at (180:1) {};
	    \begin{scope}[xshift=-2.75cm-0.6cm+0.3cm]
            \node [fill=white, circle, inner sep=0.5pt] at (1,0) {$\scriptscriptstyle{X}$};
	    \end{scope}
	    \node [myRedDot,label={[label distance=-0.5mm]below :$\scriptscriptstyle{b(\Omega)}$}] at (-90:1) {};
    \end{tikzpicture}.
$$ 
Trivially, it has the property of duplicating the tensor defining the MPDO:
\begin{align*}
    \begin{tikzpicture}[scale=0.5, baseline={([yshift=-.5ex]0,0)}]
            \draw [myWhiteLine] (-90:1) arc (-90:-20:1); 
            \draw [red,myMiddleArrow] (0:1) arc (0:90:1);
            \draw [red,myMiddleArrow] (-90:1) arc (-90:0:1);
            \node [myDot] at (0:1) {};
            \draw [myMiddleArrow] (0:1-0.75)--(0:1); 
            \draw [myMiddleArrowOUT] (0:1)--(0:1+0.75);
            \node [myRedDot] at (0:1+0.3) {};
    \end{tikzpicture}
    \overset{\mathcal{T}}{\longmapsto}
    &\begin{tikzpicture}[scale=0.5, baseline={(0,-.5ex)} ]
        \draw [myMiddleArrow] (180:1+0.75) -- (180:1); 
        \draw [rounded corners=0.125cm] (180:1) --++ (0.75+0.1,0) --++ (0,-0.5) --++ (-2.55-0.25-0.1+0.3,0) --++ (0,0.5) -- (180:1+0.75);
        \node [myVioletDot,label={[label distance=-1.211mm]above  :$\scriptscriptstyle{\;\;\Phi(\xi)}$}] at (180:1-0.5) {};
        \draw [myWhiteLine] (-160:1) arc (-160:-100:1); 
        \draw [red,myMiddleArrow] (-45:1) arc (-45:45:1);
        \draw [red,myMiddleArrow] (45:1) arc (45:180:1);
        \draw [red,myMiddleArrow] (-180:1) arc (-180:-90:1);
        \draw [red] (-90:1) arc (-90:-45:1);
	    \foreach \i in {45, -45} {
	        \draw [myMiddleArrow] (\i:1-0.75) -- (\i:1);  
	        \draw [myMiddleArrowOUT] (\i:1) -- (\i:1.75); 
	        \node [myDot] at (\i:1) {};
            \node [myRedDot,rotate=45] at (\i:1+0.3) {};
	    }
	        \node [myWhiteDot] at (180:1) {};
	    \begin{scope}[xshift=-2.75cm-0.6cm+0.3cm]
            \draw [myWhiteLine] (-90:1) arc (-90:-20:1); 
            \draw [red,myMiddleArrow] (0:1) arc (0:90:1);
            \draw [red,myMiddleArrow] (-90:1) arc (-90:0:1);
            \node [myDot] at (0:1) {};
            \node [myRedDot] at (0:1+0.3) {};
	    \end{scope}
	    \node [myRedDot,label={[label distance=-0.5mm]below :$\scriptscriptstyle{b(\Omega)}$}] at (-90:1) {};
    \end{tikzpicture}
    =
    \begin{tikzpicture}[scale=0.5,baseline={([yshift=-.5ex]0,0)}]
        \draw [myMiddleArrow] (180:1+0.75) -- (180:1); 
        \draw [rounded corners=0.125cm] (180:1) --++ (0.75+0.1,0) --++ (0,-0.5) --++ (-2.55-0.25-0.1+0.6,0) --++ (0,0.5) -- (180:1+0.75);
        \node [myVioletDot] at (180:1-0.6) {};
        \draw [myWhiteLine] (-160:1) arc (-160:-100:1); 
        \draw [red,myMiddleArrow] (-45:1) arc (-45:45:1);
        \draw [red,myMiddleArrow] (45:1) arc (45:180:1);
        \draw [red,myMiddleArrow] (-180:1) arc (-180:-90:1);
        \draw [red] (-90:1) arc (-90:-45:1);
	    \foreach \i in {45, -45} {
	        \draw [myMiddleArrow] (\i:1-0.75) -- (\i:1); 
	        \draw [myMiddleArrowOUT] (\i:1) -- (\i:1.75);  
	        \node [myDot] at (\i:1) {};
            \node [myRedDot,rotate=45] at (\i:1+0.3) {};
	    }
	        \node [myWhiteDot] at (180:1) {};
	        \node [myRedDot] at (180:1-0.3) {};
	    \begin{scope}[xshift=-2.75cm]
            \draw [myWhiteLine] (-90:1) arc (-90:-20:1); 
            \draw [red,myMiddleArrow] (0:1) arc (0:90:1);
            \draw [red,myMiddleArrow] (-90:1) arc (-90:0:1);
            \node [myDot] at (0:1) {};
	    \end{scope}
	    \node [myRedDot,label={[label distance=-0.5mm]below :$\scriptscriptstyle{b(\Omega)}$}] at (-90:1) {};
    \end{tikzpicture}
    =
    \begin{tikzpicture}[scale=0.5,baseline={([yshift=-.5ex]0,0)}]
            \draw [red] (-90:1) arc (-90:-45:1);
            \draw [red,myMiddleArrow] (-45:1) arc (-45:45:1);
            \draw [red] (45:1) arc (45:90:1);
    	    \foreach \i in {45, -45} {
    	        \draw [myMiddleArrow] (\i:1-0.75) -- (\i:1); 
    	        \draw [myMiddleArrowOUT] (\i:1) -- (\i:1.75); 
    	        \node [myDot] at (\i:1) {};
                \node [myRedDot,rotate=45] at (\i:1+0.3) {};
    	    }
    \end{tikzpicture}.
\end{align*}
In the first equality we have used that the weight $\Phi(c_\omega)\in \mathfrak{L}(V)$
can be freely moved along the physical indices since $c_\omega\in A$ is a central element.
The second equality follows from \cref{lemma:open}.
We postpone the proof of the fact that it is a quantum channel and the
definition of the quantum channel $\mathfrak{S}$ to \cref{appendix:proofrfp}.
\end{proof}

\section{RFP MPDOs are boundary states of topological 2D PEPS}
\label{sec:boundaries}

In this section we show that RFP MPDOs $\rho(\Omega,n)$ defined
in \cref{thm:rfp} arise as boundary states of topological 2D PEPS with
certain properties. As commented above, PEPS are tensor networks built using 2D arrays
of tensors for the particular case of a rectangular lattice. To construct a PEPS,
one associates a tensor describing a map from some virtual vector space to the
physical Hilbert space, to each site of a lattice and performs tensor
contractions on the virtual space according to the graph of the lattice.
PEPS exhibiting topological order has been constructed from unitary fusion categories.
The same approach can be reformulated using biconnected C*-WHAs.
See \cite[Section 7]{molnar_2022_mpo} for a detailed discussion.
As in the 1D case of MPS, global properties of PEPS can be characterized locally
using the virtual level of the individual tensors. In particular, PEPS
exhibiting topological order are characterized by tensors with MPO symmetries
acting purely at the virtual level. That corresponds to the pulling-through
condition of the MPOs on the PEPS tensors.
Finally, in this setting, the boundary state associated to the 2D PEPS is obtained
by contracting the physical indices of the 2D PEPS with open boundaries and its
conjugate transpose. Pictorially:
\begin{equation*}
    \begin{tikzpicture}[3d view = {65}{40},scale=0.45]
        \begin{scope}[canvas is xy plane at z=0]
            \foreach \k in {2,4,-1} {
                \draw [] (-0.75,\k)--(3.5+0.22,\k);
                \draw [] (4.5-0.22,\k)--(5.75,\k);
            }
            \foreach \k in {0,2,5} {
                \draw [] (\k,0.5-0.22)--(\k,4.75);
                \draw [] (\k,-1.75)--(\k,-0.5+0.22);
            }
        \end{scope}
        \begin{scope}[canvas is xy plane at z=0]
            \foreach \k in {-1,2,4} { \node [transform shape] at (4.05,\k) {$\cdots$}; }
            \foreach \k in {0,2,5} { \node [transform shape, rotate=90] at (\k,0.05) {$\cdots$}; }
        \end{scope}
        \begin{scope}[canvas is xy plane at z=1]
            \foreach \k in {2,4,-1} {
                \draw [myWhiteLine] (-0.75,\k)--(3.5,\k);
                \draw [myWhiteLine] (4.5,\k)--(5.75,\k);
                \draw [myWhiteLine] (3.5,\k)--(4.5,\k);
                \draw (-0.75,\k)--(3.5+0.22,\k);
                \draw (4.5-0.22,\k)--(5.75,\k);
            }
            \foreach \k in {0,2,5} {
                \draw [myWhiteLine] (\k,0.5)--(\k,1.5);
                \draw [myWhiteLine] (\k,2.5)--(\k,3.5);
                \draw (\k,0.5-0.22)--(\k,4.75);
                \draw (\k,-1.75)--(\k,-0.5+0.22);
            }
        \end{scope}
        \begin{scope}[canvas is xy plane at z=1]
            \foreach \k in {-1,2,4} { \node [transform shape] at (4.05,\k) {$\cdots$}; }
            \foreach \k in {0,2,5} { \node [transform shape, rotate=90] at (\k,0.05) {$\cdots$}; }
        \end{scope}
        \foreach \i in {0,2,5} {
            \foreach \j in {-1,2,4} {
                \draw [line width=1] (\i,\j,0)--++(0,0,1);
                \begin{scope}[canvas is xy plane at z=0]
                \node [rectangle, draw = black, fill = violet!5, inner sep = 4pt, transform shape] at (\i,\j) {};
                \end{scope}
                \begin{scope}[canvas is xy plane at z=1]
                \node [rectangle, draw = black, fill = violet!5, inner sep = 4pt, transform shape] at (\i,\j) {};
                \end{scope}
            }
        }
    \end{tikzpicture}
    \;\,=\;\,
    \begin{tikzpicture}[3d view = {65}{40}, scale=0.5]
    \begin{scope}[canvas is xy plane at z=0.5]
        \foreach \k in {-1,4} { \node [transform shape] at (3,\k) {$\cdots$}; }
        \foreach \k in {0,5} { \node [transform shape, rotate=90] at (\k,+0.45) {$\cdots$}; }
        \draw [red, rounded corners = 0.125cm] (0,-0.5) -- (0,-1) -- (0.5,-1);
        \draw [red, rounded corners = 0.125cm] (0,3.5) -- (0,4) -- (0.5,4);
        \draw [red, rounded corners = 0.125cm] (5,-0.5) -- (5,-1) -- (4.5,-1);
        \draw [red, rounded corners = 0.125cm] (4.5,4) -- (5,4) -- (5,3.5);
    \end{scope}
    \foreach \k in {0.75,2,4.25} {
        \draw [rounded corners = 0.05cm] (\k,4,0.5) --++ (0,0,0.5) --++ (0,0.75,0);
        \draw [rounded corners = 0.05cm] (\k,4,0.5) --++ (0,0,-0.5) --++ (0,0.75,0);
        \node [myRedDot] at (\k, 4, 0.5+0.35) {};
    }
    \foreach \k in {4,-1} {
        \draw [myWhiteLine] (1.3,\k,0.5)--(1.75,\k,0.5);
        \draw [myWhiteLine] (3.5,\k,0.5)--(3.95,\k,0.5);
        \draw [red] (0.5,\k,0.5)--(3-0.25,\k,0.5);
        \draw [red] (3+0.25,\k,0.5)--(4.5,\k,0.5);
    }
    \foreach \k in {-0.5,2,3.5} {
        \draw [rounded corners = 0.05cm] (5,\k,0.5) --++ (0,0,0.5) --++ (0.75,0,0);
        \draw [rounded corners = 0.05cm] (5,\k,0.5) --++ (0,0,-0.5) --++ (0.75,0,0);
        \node [myRedDot] at (5,\k, 0.5+0.35) {};
    }
    \foreach \k in {0,5} {
        \draw [red] (\k,-0.5,0.5)--(\k,0.45-0.25,0.5);
        \draw [myWhiteLine] (\k,1,0.5)--(\k,1.75,0.5);
        \draw [myWhiteLine] (\k,2.5,0.5)--(\k,3.25,0.5);
        \draw [red] (\k,0.45+0.25,0.5)--(\k,3.5,0.5);
    }
    \foreach \k in {-0.5,2,3.5} {
        \draw [myWhiteLine,rounded corners = 0.05cm] (0,\k,0.5) --++ (0,0,0.5) --++ (-0.75,0,0);
        \draw [myWhiteLine,rounded corners = 0.05cm] (0,\k,0.5) --++ (0,0,-0.5) --++ (-0.75,0,0);
        \draw [rounded corners = 0.05cm] (0,\k,0.5) --++ (0,0,0.5) --++ (-0.75,0,0);
        \draw [rounded corners = 0.05cm] (0,\k,0.5) --++ (0,0,-0.5) --++ (-0.75,0,0);
        \node [myRedDot] at (0,\k, 0.5+0.35) {};
    }
    \foreach \k in {0.75,2,4.25} {
        \draw [myWhiteLine,rounded corners = 0.05cm] (\k,-1,0.5) --++ (0,0,0.5) --++ (0,-0.75,0);
        \draw [myWhiteLine,rounded corners = 0.05cm] (\k,-1,0.5) --++ (0,0,-0.5) --++ (0,-0.75,0);
        \draw [rounded corners = 0.05cm] (\k,-1,0.5) --++ (0,0,0.5) --++ (0,-0.75,0);
        \draw [rounded corners = 0.05cm] (\k,-1,0.5) --++ (0,0,-0.5) --++ (0,-0.75,0);
        \node [myRedDot] at (\k, -1, 0.5+0.35) {};
    }
    \node [myDot] at (0.75,-1,0.5) {};
    \node [myDot] at (2,-1,0.5) {};
    \node [myDot] at (4.25,-1,0.5) {};
    \node [myDot] at (0.75,4,0.5) {};
    \node [myDot] at (2,4,0.5) {};
    \node [myDot] at (4.25,4,0.5) {};
    \node [myDot] at (0,-0.5,0.5) {};
    \node [myDot] at (0,2,0.5) {};
    \node [myDot] at (0,3.5,0.5) {};
    \node [myDot] at (5,-0.5,0.5) {};
    \node [myDot] at (5,2,0.5) {};
    \node [myDot] at (5,3.5,0.5) {};
    \end{tikzpicture},
\end{equation*}
Let us prove the following theorem.

\begin{theorem}
For any regular biconnected C*-WHA, RFP MPDOs defined in \cref{thm:rfp}
are boundary states of topological 2D PEPS fulfilling a renormalization fixed point property.
\end{theorem}

\begin{proof}
Fix a regular biconnected C*-WHA ${A}$ and a faithful
$*$-representation $(V,\Phi)$. As commented in \cref{sec:preliminaries}, the
associate MPO tensors are described in terms of another $*$-representation
$(W,\Psi)$ of ${A}^*$ in the virtual level.
Let us first construct  the ansatz tensor
for the 2D PEPS whose boundary state is the given matrix product density
operator $\rho(\Omega,n)$. For the sake of simplicity, we will restrict to
underlying geometries described by square lattices, although the proof
works for any 2D PEPS defined on any directed pseudo-graph. In this case,
we will consider the 2D PEPS tensor depicted as follows:
\begin{equation}
    \label{eq:2DPEPStensor}
    \begin{tikzpicture}[3d view = {65}{40}, scale = 0.75]
        \draw [white] (65:2.5)--(180+65:2.5);
        \begin{scope}[canvas is xy plane at z = 0]
            \draw [myMiddleArrow]  (0   : 1+1) -- (0   : 0);
            \draw [myMiddleArrowR] (90  : 1+1) -- (90  : 0);
            \draw [myMiddleArrowR] (180 : 1+1) -- (180 : 0);
            \draw [myMiddleArrow]  (270 : 1+1) -- (270 : 0);
        \end{scope}
        \begin{scope}[canvas is xy plane at z = 0]
        \node [rectangle, draw = black, fill = violet!5, inner sep = 10pt, transform shape] at (0,0) {};
        \end{scope}
        \draw [line width = 0.65pt] ( 0, 0, 0) --++ (0,0,2.5);
    \end{tikzpicture}
    :=
    \begin{tikzpicture}[3d view = {65}{40}, scale = 0.75]
        \draw [white] (65:3)--(180+65:3);
        \begin{scope}[canvas is xy plane at z = 0]
            \draw [myVirtual,myMiddleArrowR] (90  : 1) arc (90  : 0  : 1);
            \draw [myVirtual]                (180 : 1) arc (180 : 90 : 1);
            \draw [myVirtual]                (270 : 1) arc (270 : 180 : 1);
            \draw [myVirtual,myMiddleArrowR] (360 : 1) arc (360 : 270 : 1);
            \draw [myMiddleArrow] (0   : 1+1) -- (0   : 1);
            \draw []              (180 : 1) -- (180 : 1+1);
            \draw [myMiddleArrow] (270 : 1+1) -- (270 : 1);
            \draw [myMiddleArrow] (90 : 1) -- (90 : 1+1);
        \end{scope}
        \draw [myWhiteLine, rounded corners] ( 1,  0, 0) --++ (-0.4,0,0) --++ (0,0,2) --++ (0.4,0,0);
        \draw [myWhiteLine, rounded corners] (-1,  0, 0) --++ (0.4,0,0)  --++ (0,0,2) --++ (-0.4,0,0);
        \draw [myWhiteLine, rounded corners] ( 0,  1, 0) --++ (0,-0.4,0) --++ (0,0,2) --++ (0,0.4,0);
        \draw [myWhiteLine, rounded corners] ( 0, -1, 0) --++ (0,0.4,0)  --++ (0,0,2) --++ (0,-0.4,0);
        \draw [rounded corners, myMiddleArrowR] ( 1-0.4,  0, 2.5) --++  (0,0,-2.5) --++ (0.4,0,0); 
        \draw [rounded corners, myMiddleArrowR] (-1,  0, 0) --++ (0.4,0,0)  --++ (0,0,2.5);
        \draw [rounded corners, myMiddleArrow]  ( 0,  1-0.4, 2.5) --++ (0,0,-2.5) --++ (0,0.4,0);
        \draw [rounded corners, myMiddleArrow]  ( 0, -1, 0) --++ (0,0.4,0)  --++ (0,0,2.5);
        \begin{scope}[canvas is xy plane at z = 0]
            \node [myWhiteDot]          at (0   : 1) {};
            \node [myDot]               at (90  : 1) {};
            \node [myDot]               at (180 : 1) {};
            \node [myRedDot,rotate=-20] at (225 : 1) {};
            \node [myWhiteDot]          at (270 : 1) {};
            \node [rotate = -20, inner sep = 1pt, label={[label distance=-3pt]200:$\scriptstyle{b(\Omega)}$}] at (225 : 1) {};
        \end{scope}
        \begin{scope}[canvas is xy plane at z = 1.5]
            \node [myRedDot] at (0 : 1-0.4) {};
            \node [myRedDot, label={[label distance=-1pt]0:$\scriptstyle{b(\omega)^{\frac{1}{4}}}$}] at (90 : 1-0.4) {};
            \node [myRedDot] at (180 : 1-0.4) {};
            \node [myRedDot] at (270 : 1-0.4) {};
        \end{scope}
        \begin{scope}[canvas is xy plane at z = 1.7, overlay]
            \node [myVioletDot] at (0 : 1-0.4) {};
            \node [myVioletDot, label={[label distance=1pt]180:$\scriptstyle{   b(\omega)^{\frac{1}{4}} \Phi(\xi)^{\frac{1}{2}}   }$}] at (270 : 1-0.4) {};
        \end{scope}
        \begin{scope}[canvas is xy plane at z = 0]
            \node [myVioletDot, label={[label distance=-3pt]0:$\scriptstyle{\Psi(\hat \xi_L)}$}] at (43 : 1) {};
        \end{scope}
    \end{tikzpicture}
    ~.
\end{equation}
Here, the physical space is given by the tensor product
$V\otimes V\otimes V^*\otimes V^*$ and there are four virtual indices, each of them corresponding to $V$ or $V^*$. If read in clockwise direction starting
from the virtual weight $b(\Omega)\in\mathfrak{L}(W)$, the tensor is
algebraically described by the expresion
\begin{equation*} %
    (b(\omega)^{\frac{1}{4}})^{\otimes 4} 
        \big( \Phi(\Omega_{(1)})
              \otimes \Phi(\Omega_{(2)})
              \otimes \langle\hat\xi_L,\Omega_{(3)}\rangle
                      \Phi(\xi^{\frac{1}{2}}T(\Omega_{(4)}))
              \otimes \Phi(\xi^{\frac{1}{2}}T(\Omega_{(5)})) 
        \big)
\end{equation*}
as an operator from physical to virtual spaces.
Recall that $b(\omega)\in\mathfrak{L}(V)$ is an invertible positive central
operator and hence it can be freely moved along the physical indices. 
Let us now show that the boundary operator is the desired operator.

\vspace{0.5em}\noindent\emph{Step 1.} Let us first simplify the transfer
operator associated to the previous 2D PEPS tensor,
$\mathbb{E} = \mathbb{E}(A,V,\Phi) \in \mathfrak{L}(V\otimes V\otimes V^*\otimes V^*)$
obtained by contracting the physical indices of the 2D PEPS tensor and its
corresponding conjugate transpose if regarded as an operator. Algebraically, it
is given by the expression
\begin{multline*}
    \mathbb{E} = (b(\omega)^{\frac{1}{2}})^{\otimes 4}
    \big(
        \Phi(\Omega_{(1)}) \Phi(\Omega_{(1')})^\dagger 
        \otimes
        \Phi(\Omega_{(2)}) \Phi(\Omega_{(2')})^\dagger \otimes
        \overline{\langle\hat{\xi}_L^{-1},\Omega_{(3')}\rangle}
        \langle\hat{\xi}_L^{-1},\Omega_{(3')}\rangle
    \\
        \cdot 
        \Phi (T(\Omega_{(4')}))^\dagger \Phi(\xi)  \Phi(T(\Omega_{(4)}))
        \otimes
        \Phi (T(\Omega_{(5')}))^\dagger \Phi(\xi) \Phi(T(\Omega_{(5)}))
    \big),
\end{multline*}
where we have employed that $b(\omega)\in\mathfrak{L}(V)$ is positive and
central and $\Phi(\xi)\in\mathfrak{L}(V)$ is positive, since
$\Phi\in\mathfrak{L}(A,\mathfrak{L}(V))$ is a $*$-representation
and $\xi\in A$ is positive. Note that the order of composition is reversed for
the terms associated to white tensors. In order to fully describe $\mathbb{E}$ in terms
of tensor networks, note that
\begin{equation*}
    \overline{   \langle\hat{\xi}_L^{-1} , x   \rangle}
    =
    \langle ( \hat\xi_L^{-1} )^* , S(x)^* \rangle
    =
    \langle \hat\xi_L^{-1} , S(x)^* \rangle
    =
    \langle \hat\xi_L^{-1} , S^{-1}(x^*) \rangle
    =
    \langle \hat\xi_R^{-1} , x^* \rangle
\end{equation*}
for all $x\in A$, where the first equality is due to \cref{remark:dualcwha},
the second equality follows from the positivity of
$\hat\xi_L\in {A}^*$, the third equality is due to \cref{remark:propsS} and the fourth equality follows from the
definition of
$\hat \xi_R\in A^*$, see \cref{lemma:xi}. In addition, recall that
$\Phi\in\mathfrak{L}(A,\mathfrak{L}(V))$ is a $*$-re\-pre\-sen\-tation and
$T(x)^*\xi = \xi T(x^*)$ for all $x\in A$, see \cref{lemma:xi}. Therefore:
\begin{multline*}
    \mathbb{E} =
    (b(\omega)^{\frac{1}{2}})^{\otimes 4}
    \big(
        \Phi(\Omega_{(1)}\Omega_{(1')}^*) \otimes \Phi(\Omega_{(2)}\Omega_{(2')}^*) 
    \\
        \otimes
        \langle\hat{\xi}_R^{-1},\Omega_{(3')}^*\rangle\langle\hat{\xi}_L^{-1},\Omega_{(3)}\rangle
        \Phi (\xi T(\Omega_{(4')}^*) T(\Omega_{(4)}))
        \otimes \Phi (\xi T(\Omega_{(5')}^*) T(\Omega_{(5)})).
    \big)
\end{multline*}
Hence, the transfer operator can be represented graphically as follows:
\begin{equation*}
    \mathbb{E}
    =
    \begin{tikzpicture}[3d view = {65}{40}, scale = 0.7]
        \begin{scope}[canvas is xy plane at z = 0]
            \draw [myVirtual,myMiddleArrowR] (90  : 1) arc (90  : 0  : 1);
            \draw [myVirtual]                (180 : 1) arc (180 : 90 : 1);
            \draw [myVirtual]                (270 : 1) arc (270 : 180 : 1);
            \draw [myVirtual,myMiddleArrowR] (360 : 1) arc (360 : 270 : 1);
            \draw [myMiddleArrow] (0   : 1+1) -- (0   : 1);
            \draw [myMiddleArrow] (90   : 1) -- (90   : 1+1);
            \draw []              (180 : 1) -- (180 : 1+1);
            \draw [myMiddleArrow] (270 : 1+1) -- (270 : 1);
        \end{scope}
        \draw [myWhiteLine, rounded corners] ( 1,  0, 0) --++ (-0.4,0,0) --++ (0,0,3) --++ (0.4,0,0);
        \draw [myWhiteLine, rounded corners] (-1,  0, 0) --++ (0.4,0,0)  --++ (0,0,3) --++ (-0.4,0,0);
        \draw [myWhiteLine, rounded corners] ( 0,  1, 0) --++ (0,-0.4,0) --++ (0,0,3) --++ (0,0.4,0);
        \draw [myWhiteLine, rounded corners] ( 0, -1, 0) --++ (0,0.4,0)  --++ (0,0,3) --++ (0,-0.4,0);
        \draw [rounded corners,myMiddleArrowR] ( 1,  0, 3) --++ (-0.4,0,0) --++ (0,0,-1.5);
        \draw [rounded corners,myMiddleArrowR] ( 1-0.4,  0, 1.5) --++  (0,0,-1.5) --++ (0.4,0,0);
        \draw [rounded corners,myMiddleArrowR] (-1,  0, 0) --++ (0.4,0,0)  --++ (0,0,3) --++ (-0.4,0,0);
        \draw [rounded corners]                ( 0,  1, 3) --++ (0,-0.4,0) --++ (0,0,-1.5);
        \draw [rounded corners,myMiddleArrow]  ( 0,  1-0.4, 1.5) --++  (0,0,-1.5) --++ (0,0.4,0);
        \draw [rounded corners,myMiddleArrow]  ( 0, -1, 0) --++ (0,0.4,0)  --++ (0,0,3) --++ (0,-0.4,0);
        \begin{scope}[canvas is xy plane at z = 3]
            \draw [myWhiteLine] (90  : 1) arc (90  : 0   : 1);
            \draw [myWhiteLine] (180 : 1) arc (180 : 90  : 1);
            \draw [myWhiteLine] (270 : 1) arc (270 : 180 : 1);
            \draw [myWhiteLine] (360 : 1) arc (360 : 270 : 1);
            \draw [myVirtual,myMiddleArrowR] (90  : 1) arc (90  : 0   : 1);
            \draw [myVirtual] (180 : 1) arc (180 : 90 : 1);
            \draw [myVirtual]                (270 : 1) arc (270 : 180 : 1);
            \draw [myVirtual,myMiddleArrowR] (360 : 1) arc (360 : 270 : 1);
            \draw [myWhiteLine] (0   : 1+1) -- (0   : 1);
            \draw [myWhiteLine] (90  : 1+1) -- (90  : 1);
            \draw [myWhiteLine] (180 : 1+1) -- (180 : 1);
            \draw [myWhiteLine] (270 : 1+1) -- (270 : 1);
            \draw [myMiddleArrow] (0   : 1+0.5) -- (0   : 1+1);
            \draw (0   : 1) -- (0   : 1+0.5);
            \draw (90  : 1+0.5) -- (90  : 1);
            \draw [myMiddleArrow] (90  : 1+1) -- (90  : 1+0.5);
            \draw (180 : 1+0.5) -- (180 : 1);
            \draw [myMiddleArrow] (180 : 1+1) -- (180 : 1+0.5);
            \draw (270 : 1) -- (270 : 1+0.5);
            \draw [myMiddleArrow] (270 : 1+0.5) -- (270 : 1+1);
        \end{scope}
        \begin{scope}[canvas is xy plane at z = 0]
            \node [myWhiteDot] at (0   : 1) {};
            \node [myDot] at (90  : 1) {};
            \node [myDot] at (180 : 1) {};
            \node [myRedDot,rotate=-20] at (225 : 1) {};
            \node [myWhiteDot] at (270 : 1) {};
            \node [rotate=-20, inner sep = 1pt,label ={[label distance=-3pt]200:$\scriptstyle{b(\Omega)}$}] at (225 : 1) {};
        \end{scope}
        \begin{scope}[canvas is xy plane at z = 3]
            \node [myRedDot,rotate=30] at (0 : 1.5) {};
            \node [myRedDot,rotate=20] at (90 : 1.5) {};
            \node [inner sep = 1pt,label ={[label distance=-2pt]0:$\scriptstyle{b(\omega)^{\frac{1}{2}}}$}] at (175 : 1.8) {};
            \node [myRedDot,rotate=30] at (180 : 1.5) {};
            \node [myRedDot,rotate=20] at (270 : 1.5) {};
            \node [rotate=-20, inner sep = 1pt,label ={[label distance=-5pt]180:$\scriptstyle{b(\Omega^*)}$}] at (225 : 1) {};
        \end{scope}
        \begin{scope}[canvas is xy plane at z = 3, label distance = -2mm]
            \node [myWhiteDot] at (0   : 1) {};
            \node [myDot] at (90  : 1) {};
            \node [myDot] at (180 : 1) {};
            \node [myRedDot,rotate=-12.5] at (225 : 1) {};
            \node [myWhiteDot] at (270 : 1) {};
        \end{scope}
        \begin{scope}[canvas is xy plane at z = 1.5]
            \node [myVioletDot] at (0:1-0.4) {};
            \node [myVioletDot, label={[label distance=2pt]180:$\scriptstyle{\Phi(\xi)}$}] at (-90:1-0.4) {};
        \end{scope}
        \begin{scope}[canvas is xy plane at z = 0]
            \node [myVioletDot, label={[label distance=-2pt]0:$\scriptstyle{\Psi(\hat \xi_L^{-1})}$}] at (43:1) {};
        \end{scope}
        \begin{scope}[canvas is xy plane at z = 3]
            \node [myVioletDot, label={[label distance=-2pt]0:$\scriptstyle{\Psi(\hat \xi_R^{-1})}$}] at (43:1) {};
        \end{scope}
    \end{tikzpicture}
    =
    \begin{tikzpicture}[3d view = {65}{40}, scale = 0.7]
        \begin{scope}[canvas is xy plane at z = 0]
            \draw [myVirtual,myMiddleArrowR] (90  : 1) arc (90  : 0  : 1);
            \draw [myVirtual]                (180 : 1) arc (180 : 90 : 1);
            \draw [myVirtual]                (270 : 1) arc (270 : 180 : 1);
            \draw [myVirtual,myMiddleArrowR] (360 : 1) arc (360 : 270 : 1);
            \draw [myMiddleArrow] (0   : 1+1) -- (0   : 1);
            \draw [myMiddleArrow] (90   : 1) -- (90   : 1+1);
            \draw []              (180 : 1) -- (180 : 1+1);
            \draw [myMiddleArrow] (270 : 1+1) -- (270 : 1);
        \end{scope}
        \draw [myWhiteLine, rounded corners] ( 1,  0, 0) --++ (-0.4,0,0) --++ (0,0,3) --++ (0.4,0,0);
        \draw [myWhiteLine, rounded corners] (-1,  0, 0) --++ (0.4,0,0)  --++ (0,0,3) --++ (-0.4,0,0);
        \draw [myWhiteLine, rounded corners] ( 0,  1, 0) --++ (0,-0.4,0) --++ (0,0,3) --++ (0,0.4,0);
        \draw [myWhiteLine, rounded corners] ( 0, -1, 0) --++ (0,0.4,0)  --++ (0,0,3) --++ (0,-0.4,0);
        \draw [rounded corners,myMiddleArrowR] ( 1,  0, 3) --++ (-0.4,0,0) --++ (0,0,-1.5);
        \draw [rounded corners,myMiddleArrowR] ( 1-0.4,  0, 1.5) --++  (0,0,-1.5) --++ (0.4,0,0);
        \draw [rounded corners,myMiddleArrowR] (-1,  0, 0) --++ (0.4,0,0)  --++ (0,0,3) --++ (-0.4,0,0);
        \draw [rounded corners]                ( 0,  1, 3) --++ (0,-0.4,0) --++ (0,0,-1.5);
        \draw [rounded corners,myMiddleArrow]  ( 0,  1-0.4, 1.5) --++  (0,0,-1.5) --++ (0,0.4,0);
        \draw [rounded corners,myMiddleArrow]  ( 0, -1, 0) --++ (0,0.4,0)  --++ (0,0,3) --++ (0,-0.4,0);
        \begin{scope}[canvas is xy plane at z = 3]
            \draw [myWhiteLine] (90  : 1) arc (90  : 0   : 1);
            \draw [myWhiteLine] (180 : 1) arc (180 : 90  : 1);
            \draw [myWhiteLine] (270 : 1) arc (270 : 180 : 1);
            \draw [myWhiteLine] (360 : 1) arc (360 : 270 : 1);
            \draw [myVirtual,myMiddleArrowR] (90  : 1) arc (90  : 0   : 1);
            \draw [myVirtual] (180 : 1) arc (180 : 90 : 1);
            \draw [myVirtual]                (270 : 1) arc (270 : 180 : 1);
            \draw [myVirtual,myMiddleArrowR] (360 : 1) arc (360 : 270 : 1);
            \draw [myWhiteLine] (0   : 1+1) -- (0   : 1);
            \draw [myWhiteLine] (90  : 1+1) -- (90  : 1);
            \draw [myWhiteLine] (180 : 1+1) -- (180 : 1);
            \draw [myWhiteLine] (270 : 1+1) -- (270 : 1);
            \draw [myMiddleArrow] (0   : 1+0.5) -- (0   : 1+1);
            \draw (0   : 1) -- (0   : 1+0.5);
            \draw (90  : 1+0.5) -- (90  : 1);
            \draw [myMiddleArrow] (90  : 1+1) -- (90  : 1+0.5);
            \draw (180 : 1+0.5) -- (180 : 1);
            \draw [myMiddleArrow] (180 : 1+1) -- (180 : 1+0.5);
            \draw (270 : 1) -- (270 : 1+0.5);
            \draw [myMiddleArrow] (270 : 1+0.5) -- (270 : 1+1);
        \end{scope}
        \begin{scope}[canvas is xy plane at z = 0]
            \node [myWhiteDot] at (0   : 1) {};
            \node [myDot] at (90  : 1) {};
            \node [myDot] at (180 : 1) {};
            \node [myRedDot,rotate=-20] at (225 : 1) {};
            \node [myWhiteDot] at (270 : 1) {};
            \node [rotate=-20, inner sep = 1pt,label ={[label distance=-3pt]200:$\scriptstyle{b(\Omega)}$}] at (225 : 1) {};
        \end{scope}
        \begin{scope}[canvas is xy plane at z = 3]
            \node [myRedDot,rotate=30] at (0 : 1.5) {};
            \node [myRedDot,rotate=20] at (90 : 1.5) {};
            \node [myRedDot,rotate=30] at (180 : 1.5) {};
            \node [myRedDot,rotate=20] at (270 : 1.5) {};
            \node [rotate=-20, inner sep=1pt,label ={[label distance=-5pt]180:$\scriptstyle{b(\Omega^*)}$}] at (225 : 1) {};
        \end{scope}
        \begin{scope}[canvas is xy plane at z = 3]
            \node [myWhiteDot] at (0   : 1) {};
            \node [myDot] at (90  : 1) {};
            \node [myDot] at (180 : 1) {};
            \node [myRedDot, rotate=-12.5] at (225 : 1) {};
            \node [myWhiteDot] at (270 : 1) {};
        \end{scope}
        \begin{scope}[canvas is xy plane at z = 3]
            \node [myVioletDot] at (0:1+0.3) {};
            \node [myVioletDot, label={[label distance=-4pt]135:$\scriptstyle{b(\omega)^{\frac{1}{2}}\Phi(\xi)}$}] at (-90:1+0.3) {};
        \end{scope}
        \begin{scope}[canvas is xy plane at z = 0]
            \node [myVioletDot, label={[label distance=-2pt]0:$\scriptstyle{\Psi(\hat \xi_L^{-1})}$}] at (43:1) {};
        \end{scope}
        \begin{scope}[canvas is xy plane at z = 3]
            \node [myVioletDot, label={[label distance=-2pt]0:$\scriptstyle{\Psi(\hat \xi_R^{-1})}$}] at (43:1) {};
        \end{scope}
    \end{tikzpicture}
\end{equation*}
On the other hand, $\Psi(\hat\xi_L^{-1})$ and $\Psi(\hat\xi_R^{-1})$ can be ``moved''
from the virtual to the physical spaces using the following identities:
\begin{equation*}
    x_{(1)} \langle\hat\xi_L^{-1},x_{(2)}\rangle   = \xi_L^{-1} x 
    \;\,\text{ and }\;\,
    x_{(1)} \langle\hat\xi_R^{-1},x_{(2)}\rangle = x \xi_L^{-1} 
\end{equation*}
for all elements $x\in {A}$; see \cref{lemma:xi}. In graphical
notation, the previous formulas are rephrased in the following form:
\begin{equation}
    \label{eq:xiPhysVirt}
.
\end{equation}
By virtue of these identities, the fact that $T\in\mathfrak{L}({A})$ is an algebra
anti-homomorphism and $\Omega^* = \Omega = \Omega^2$ by \cref{thm:Omega}, it follows that
\begin{align*}
    %
    .
\end{align*}
Applying again \cref{eq:xiPhysVirt}, we conclude that the transfer operator takes the form:
\begin{align*}
    \mathbb{E}=
    %
.
\end{align*}
since $\Psi\in\mathfrak{L}(A^*,\mathfrak{L}(W))$ is a $*$-representation and $\xi^{-1} = \xi_L^{-1}\xi_{R}^{-1}$.


\vspace{0.5em}\noindent\emph{Step 2.} Let us consider the concatenation
of two transfer operators. 
By virtue of the pulling-through identity \cref{eq:pt},
\begin{equation}
\label{eq:twoTO}
%
\end{equation}

%
\vspace{0.5em}\noindent\emph{Step 3.} Let us consider the concatenation of four
transfer operators around the vertices of a plaquette. In particular, we prove
that the ansatz 2D PEPS tensors gives
rise to a \emph{normalized} PEPS and its boundary state is the RFP MPDO defined
in \cref{thm:rfp}. Consider, from a top view, the procedure of
simplifying the concatenation of transfer operators that form a whole plaquette:
\begin{equation*}
    %
\end{equation*}
For the purpose of closing the plaquette, recall \cref{lemma:open} and \cref{lemma:xi}:
\begin{equation*}
    %
    .
\end{equation*}
Note that in the previous equations the inner circle representing
$\mathrm{Tr}^a(\hat\xi^{-1})$ is not independent of the outer shape and
hence it gives rise to possibly different constant in each sector, as it is a
sum over all sectors $a \in\mathrm{Irr}(A)$. As showed in \cref{lemma:xi}, these
are precisely the Frobenius-Perron dimensions which define, in each sector, the
canonical regular element $\Omega\in {A}$. Therefore we can rewrite it in terms
of the weight $b(\Omega)\in\mathfrak{L}(W)$, as done in the last equality.
Iterating this procedure for each plaquette of the lattice proves that matrix
product density operators defined in the previous section arise naturally as
boundary states of topological 2D PEPS.

Note also that \cref{eq:twoTO} is
nothing but a natural 2D generalization of the renormalization fixed point condition for MPS
defined in \cite{cirac_2017_mpdo}. In that sense, we can conclude that the RFP
MPDOs considered in this paper are boundary states of PEPS fulfilling
this renormalization fixed point property.
\end{proof}

\section{Classification via shallow circuits of quantum channels}
\label{sec:classif}

In this section we prove that RFP MPDOs arising from C*-HAs belong
to the trivial phase. Namely, we provide explicit definitions of depth-two
circuits of finite-range quantum channels that map the maximally mixed state
to these RFP MPDOs. Finally, we show that our construction cannot be extended
to arbitrary biconnected C*-WHAs, which lead us to the conjecture
that there are non-trivial phases in this context.

In order to deepen the intuition towards the general case of C*-HAs,
let first examine the simplest non-trivial example.

\begin{example}\label{example:circuitZ2}
RFP MPDOs arising from the group C*-HA $A:=\mathbb{CZ}_2$, introduced
in \cref{ex:groupCHA} and \cref{ex:Z2}, are in the trivial phase.
Specifically, we build
$$ \rho(\Omega,n) = \tfrac{1}{2^{n}}(\mathbf{1}^{\otimes n} + \sigma_z^{\otimes n}) $$
via a depth-two circuit of range-two quantum channels from the maximally mixed
state $\mathrm{Tr}(\mathbf{1})^{-n} \mathbf{1}^{\otimes n}$. We assume without
loss of generality that $n\in\mathbb{N}$ is even and propose the following
procedure:

\vspace{0.5em}\noindent\emph{Step 1 (``initialization'').}
We first construct $n/2$ copies $(\rho_{2})^{\otimes n/2}$ of the mixed state
$\rho_2$ between pairs of nearest neighbors by replacing the product states
separately. This is easily done by means of the quantum channel
$\mathfrak{N}:X\otimes Y\mapsto \mathrm{Tr}(X\otimes Y)\rho_{2}$. In the
Choi-Jamio{\l}kowski picture, this process can be depicted as follows:
\begin{equation*}
    \begin{tikzpicture}[baseline={(0,0.45cm-0.5ex)}, scale=0.5]
        \foreach \i in {0,...,8} {
            \ifnum \i=6 {}
            \else {
                \draw [myMiddleArrow,xshift = \i cm] (-0.25,2) -- (-0.25,0.95);
                \draw [myMiddleArrow,xshift = \i cm] (0.25,0.95) -- (0.25,2);
                \draw [rounded corners=0.0625cm, xshift = \i cm] (-0.25,0.1) -- (-0.25,-0.3) -- (0.25,-0.3) -- (0.25,0.1);
                \draw [myMiddleArrow,xshift = \i cm] (-0.15,-0.3)--(0.15,-0.3);
            }
            \fi
        }
        \foreach \i in {0,1,2,3.5} {
            \node [myGateStyle] at (0.5+2*\i,0.5) {$\scriptstyle\mathfrak{N}$};
        }
        \node at (6.1,0.5) {$\cdots$};
        \node at (6.1,1.375) {$\cdots$};
    \end{tikzpicture}
    =
    \begin{tikzpicture}[baseline={(0,0.45cm-0.5ex)}, scale=0.5]
        \foreach \i in {0,...,8} {
            \ifnum \i=6 {}
            \else {
                \draw [myMiddleArrow,xshift = \i cm] (-0.25,2) -- (-0.25,0.99);
                \draw [myMiddleArrow,xshift = \i cm] (0.25,0.99) -- (0.25,2);
            }
            \fi
        }
        \foreach \i in {0,1,2,3.5} {
            \node [myGateStyle] at (0.5+2*\i,0.5) {$\scriptstyle{\rho_{2}\vphantom{A^a}}$};
        }
        \node at (6.1,0.5) {$\cdots$};
        \node at (6.1,1.375) {$\cdots$};
    \end{tikzpicture}
\end{equation*}
When the system size is an odd natural number simply replace three of them with
the mixed state $\rho_3$, for example.

\vspace{0.5em}\noindent\emph{Step 2 (``gluing'').} Now, we ``glue'' together all
these copies of $\rho_{2} $ in order to obtain the target mixed state $\rho_{n}$.
This is done inductively by means of the following quantum channel, called from
now on \emph{gluing} map:
\begin{equation}
\mathfrak{G}: X \otimes Y \mapsto
    \tfrac{1}{2^{2}}(\mathrm{Tr}(X\otimes Y)\mathbf{1}\otimes\mathbf{1}+
    \mathrm{Tr}(X \sigma_z\otimes Y\sigma_z)\sigma_z\otimes \sigma_z)
\end{equation}
for all $X,Y\in\mathfrak{L}(\mathbb{C}^2)$. It is easy to check that it is a
quantum channel and that
\begin{equation*}
\mathrm{Id}\otimes\mathfrak{G}\otimes\mathrm{Id}(\tfrac{1}{2^{2}}(\mathbf{1}^{\otimes 2}+ \sigma_z^{\otimes 2})\otimes \tfrac{1}{2^{2}}(\mathbf{1}^{\otimes 2} + \sigma_z^{\otimes 2})) = \tfrac{1}{2^{4}}(\mathbf{1}^{\otimes 4} + \sigma_z^{\otimes 4}). 
\end{equation*}
By induction, it is clear that simultaneous applications of these quantum
channels lead to the mixed state $\rho_n$. Again, in the Choi-Jamio{\l}kowski
picture this procedure can be depicted as follows:
\begin{equation*}
    \begin{tikzpicture}[baseline={(0,0.75cm-0.5ex)}, scale=0.5]
        \foreach \i in {0,...,8} {
            \ifnum \i=6 {}
            \else {
                \draw [xshift = \i cm] (-0.25,2.5) -- (-0.25,0.1);
                \draw [xshift = \i cm] (0.25,0.1) -- (0.25,2.5);
                \draw [myMiddleArrow,xshift = \i cm] (-0.25,2.5) -- (-0.25,3.5);
                \draw [myMiddleArrow,xshift = \i cm] (0.25,3.5) -- (0.25,2.5);
            } \fi
        }
        \foreach \i in {0,1,2,3.5} {
            \node [myGateStyle] at (0.5+2*\i,0.5) {$\scriptstyle{\rho_{2}\vphantom{A^a}}$};
        }
        \foreach \i in {0,1} {
            \node [myGateStyle] at (1.5+2*\i,1.875) {$\scriptstyle\mathfrak{G}$};
        }
        \foreach \i in {2,2.5} {
            \node [myGateStyle] at (1.5+2*\i,1.875) {$\phantom{\scriptstyle\mathfrak{G}}$};
        }
        \filldraw [draw=none,fill=white] (5.5,1) rectangle ++ (1,1.5);
        \node at (6.1,0.5) {$\cdots$};
        \node at (6.1,1.875) {$\cdots$};
    \end{tikzpicture}
    =
    \begin{tikzpicture}[baseline={(0,0.75cm-0.5ex)}, scale=0.5]
        \foreach \i in {0,...,8} {
            \ifnum \i=6 {}
            \else {
                \draw [xshift = \i cm] (-0.25,2.5) -- (-0.25,0.1);
                \draw [xshift = \i cm] (0.25,0.1) -- (0.25,2.5);
                \draw [myMiddleArrow,xshift = \i cm] (-0.25,2.5) -- (-0.25,3.5);
                \draw [myMiddleArrow,xshift = \i cm] (0.25,3.5) -- (0.25,2.5);
            } \fi
        }
        \node [myGateStyle, minimum width=4.4cm] at (4,0.5) {$\scriptstyle{\rho_{n}\phantom{A^a}}$};
        \filldraw [draw=none,fill=white] (5.5,-0.1) rectangle ++ (1,1.5);
        \node at (6.1,0.5) {$\cdots$};
        \node at (6.1,1.875) {$\cdots$};
    \end{tikzpicture}
\end{equation*}
\end{example}

The previous construction can be generalized to arbitrary C*-HAs as
follows. In the first place, the role of the previous element is replaced by the
RFP MPDO associated to the canonical regular element. In addition, we introduce
a family of quantum channels that ``glue'' together two RFP MPDOs associated to
the canonical regular element $\Omega\in A$ into a larger one, associated to any
arbitrary positive non-zero element of $A$.

\begin{restatable}{lemma}{restatableglueHA}
\label{lemma:glueHA}
Let $A$ be a C*-HA and let $(V,\Phi)$ be a faithful $*$-re\-pre\-sen\-ta\-tion
of $A$. Then, for all positive non-zero elements $x\in A$ there exists a quantum channel
$\mathfrak{G}_x\in\mathfrak{L}(\mathfrak{L}(V\otimes V))$,
called ``gluing'' map, such that
\begin{equation}
\label{eq:glueHA}
    (\mathrm{Id}^{\otimes m-1}\otimes \mathfrak{G}_x \otimes\mathrm{Id}^{\otimes n-1})
    (\rho(\Omega,m)\otimes\rho(\Omega,n))
    =
    \rho(x,m+n)
\end{equation}
for all $m,n\in\mathbb{N}$.
\end{restatable}
See \cref{sec:proofs-classif-hopf} for a proof, but let us propose now an
explicit expression for the gluing map and check using graphical notation
that \cref{eq:glueHA} holds. To this end, fix any positive non-zero
element $x\in A$ and assume without loss of generality that $m = n = 2$. Define
the map $\mathfrak{G}_x\in\mathfrak{L}(\mathfrak{L}(V\otimes V))$ by the
expression
$$
    X\otimes Y \mapsto
    \tfrac{1}{\langle\omega,x\rangle}
    \begin{tikzpicture}[scale=0.5]
    \begin{scope}[shift={(-2.1,0.9)},rotate=-45] 
        \draw [rounded corners=0.125cm] (90:1) -- (90:2.5) --++ (-0.5,0) --++ (0,-2) --++ (0.5,0) -- (90:1);
        \draw [myMiddleArrow] (-0.5,1.5-0.15) --++ (0,-0.5);
        \node [fill=white, circle, inner sep=0.5pt] at (0,1.25) {$\scriptscriptstyle{X}$};
    \end{scope}
    \begin{scope}[shift={(2.1,0.9)},rotate=45] 
        \draw [rounded corners=0.125cm] (90:1) -- (90:2.5) --++ (0.5,0) --++ (0,-2) --++ (-0.5,0) -- (90:1);
        \draw [myMiddleArrow] (0.5,1.5-0.15) --++ (0,-0.5);
        \node [fill=white, circle, inner sep=0.5pt] at (0,1.25) {$\scriptscriptstyle{Y}$};
    \end{scope}
    \begin{scope}[shift={(0,3)}] 
        \draw [myWhiteLine] (-30:1) arc (-30:0:1);
        \draw [myWhiteLine] (180:1) arc (180:210:1);
        \draw [myMiddleArrow,red] (45:1) arc (45:90+45:1);
        \draw [myMiddleArrow,red] (135:1) arc (135:90+135:1);
        \draw [red] (225:1) arc (225:90+225:1);
        \draw [myMiddleArrow,red] (-45:1) arc (-45:90-45:1);
        \draw [myMiddleArrow] (45:1-0.75) -- (45:1); 
        \draw [myMiddleArrowOUT] (45:1) -- (45:1+0.75);
        \draw [myMiddleArrow] (135:1-0.75) -- (135:1); 
        \draw [myMiddleArrowOUT] (135:1) -- (135:1+0.75);
        \node [myDot] at (90+45:1) {};
        \node [myRedDot, rotate=45] at (90+45:1.3) {};
        \node [myDot] at (90-45:1) {};
        \node [myRedDot, rotate=45] at (90-45:1.3) {};
        \node [myRedDot,label={[label distance=-0.7mm]below :$\scriptscriptstyle{b(x)}$}] at (-90:1) {};
        \node [myWhiteDot] at (180+45:1) {};
        \node [myWhiteDot] at (-45:1) {};
    \end{scope}
    \end{tikzpicture}.
$$
for all $X,Y\in\mathfrak{L}(V)$.
To prove that \cref{eq:glueHA} holds, recall first that $\Phi(c_\omega)\in\mathfrak{L}(V)$ can be moved freely
along the physical vector spaces. By virtue of \cref{lemma:open}:
\begin{equation*}
\tfrac{1}{\langle\omega,x\rangle}\tfrac{1}{\langle\omega,\Omega\rangle^2}
\begin{tikzpicture}[baseline={([yshift=-.5ex]0,1)},scale=0.5]
\begin{scope}[shift={(-2.1,0.9)},rotate=-45] 
    \draw [rounded corners=0.125cm] (90:1) -- (90:2+0.5) --++ (-0.5,0) --++ (0,-2) --++ (0.5,0)  -- (90:1);
\end{scope}
\begin{scope}[shift={(2.1,0.9)},rotate=45] 
    \draw [rounded corners=0.125cm] (90:1) -- (90:2+0.5) --++ (0.5,0) --++ (0,-2) --++ (-0.5,0) -- (90:1);
\end{scope}
\begin{scope}[shift={(0,3)}] 
    \draw [myWhiteLine] (-30:1) arc (-30:0:1);
    \draw [myWhiteLine] (180:1) arc (180:210:1);
    \draw [myMiddleArrow,red] (45:1) arc (45:90+45:1);
    \draw [myMiddleArrow,red] (135:1) arc (135:90+135:1);
    \draw [red] (225:1) arc (225:90+225:1);
    \draw [myMiddleArrow,red] (-45:1) arc (-45:90-45:1);
    \draw [myMiddleArrow] (45:1-0.75) -- (45:1); 
    \draw [] (45:1) -- (45:1+0.3);
    \draw [myMiddleArrow] (45:1+0.3) -- (45:1+0.3+0.75);
    \draw [myMiddleArrow] (135:1-0.75) -- (135:1); 
    \draw [] (135:1) -- (135:1+0.3);
    \draw [myMiddleArrow] (135:1+0.3) -- (135:1+0.3+0.75);
    \node [myDot] at (90+45:1) {};
    \node [myRedDot, rotate=45] at (90+45:1.3) {};
    \node [myDot] at (90-45:1) {};
    \node [myRedDot, rotate=45] at (90-45:1.3) {};
    \node [myRedDot,label={[label distance=-0.7mm]below :$\scriptscriptstyle{b(x)}$}] at (-90:1) {};
\end{scope}
\begin{scope}[shift={(-2.1,0.9)},rotate=-45] 
    \draw [myWhiteLine] (100:1) arc (100:130:1);
    \draw [red,myMiddleArrow] (0:1) arc (0:90:1);
    \draw [red,myMiddleArrow] (90:1) arc (90:270:1);
    \draw [red,myMiddleArrow] (-90:1) arc (-90:0:1);
    \draw [] (-90:1)--(-90:1+0.3);
    \draw [myMiddleArrow] (-90:1+0.3)--(-90:1+0.3+0.75);
    \draw [myMiddleArrow] (-90:1-0.75)--(-90:1); 
    \node [myDot] at (90:1) {};
    \node [myRedDot, rotate=45] at (90:1+0.3) {};
    \node [myWhiteDot] at (90:2) {};
    \node [myRedDot, rotate=45] at (-90:1+0.3) {};
    \node [myRedDot, rotate=45,label={[label distance=-1.5mm]below :$\scriptscriptstyle{b(\Omega)}$}] at (0:1) {};
    \node [myDot] at (-90:1) {};
\end{scope}
\begin{scope}[shift={(2.1,0.9)},rotate=45] 
    \draw [myWhiteLine] (30:1) arc (30:75:1);
    \draw [red,myMiddleArrow] (-90:1) arc (-90:90:1);
    \draw [red,myMiddleArrow] (90:1) arc (90:180:1);
    \draw [red,myMiddleArrow] (180:1) arc (180:270:1);
    \draw [] (-90:1)--(-90:1+0.3);
    \draw [myMiddleArrow] (-90:1+0.3)--(-90:1+0.3+0.75);
    \draw [myMiddleArrow] (-90:1-0.75)--(-90:1); 
    \node [myDot] at (90:1) {};
    \node [myRedDot, rotate=45] at (90:1+0.3) {};
    \node [myWhiteDot] at (90:2) {};
    \node [myRedDot, rotate=45] at (-90:1+0.3) {};
    \node [myRedDot, rotate=-45, label={[label distance=-1.5mm]below :$\scriptscriptstyle{b(\Omega)}$}] at (180:1) {};
    \node [myDot] at (-90:1) {};
\end{scope}
\end{tikzpicture}
=
\tfrac{1}{\langle\omega,x\rangle}\tfrac{1}{\langle\omega,\Omega\rangle}
\begin{tikzpicture}[scale=0.5]
\begin{scope}[shift={(2.1,0.9)},rotate=45] 
    \draw [rounded corners=0.125cm] (90:1) -- (90:2.5) --++ (0.5,0) --++ (0,-2) --++ (-0.5,0) -- (90:1);
\end{scope}
\begin{scope}[shift={(0,3)}] 
    \draw [myWhiteLine] (-30:1) arc (-30:0:1);
    \draw [myMiddleArrow,red] (45:1) arc (45:90+45:1);
    \draw [myMiddleArrow,red] (135:1) arc (135:90+135:1);
    \draw [red] (225:1) arc (225:90+225:1);
    \draw [myMiddleArrow,red] (-45:1) arc (-45:90-45:1);
    \draw [myMiddleArrow] (45:1-0.75) -- (45:1); 
    \draw [] (45:1) -- (45:1+0.3);
    \draw [myMiddleArrow] (45:1.3) -- (45:1+0.3+0.75);
    \draw [myMiddleArrow] (135:1-0.75) -- (135:1); 
    \draw [] (135:1) -- (135:1+0.3); 
    \draw [myMiddleArrow] (135:1+0.3) -- (135:1+0.3+0.75); 
    \draw [myMiddleArrow] (225:1-0.75) -- (225:1); 
    \draw [] (225:1) -- (225:1+0.3); 
    \draw [myMiddleArrow] (225:1+0.3) -- (225:1+0.3+0.75); 
    \node [myDot] at (135:1) {};
    \node [myRedDot,rotate=45] at (135:1.3) {};
    \node [myDot] at (45:1) {};
    \node [myRedDot,rotate=45] at (45:1.3) {};
    \node [myDot] at (225:1) {};
    \node [myRedDot,rotate=45] at (225:1.3) {};
    \node [myRedDot,label={[label distance=-0.5mm]below :$\scriptscriptstyle{b(x)}$}] at (-90:1) {};
\end{scope}
\begin{scope}[shift={(2.1,0.9)},rotate=45] 
    \draw [myWhiteLine] (30:1) arc (30:75:1);
    \draw [red,myMiddleArrow] (-90:1) arc (-90:90:1);
    \draw [red,myMiddleArrow] (90:1) arc (90:180:1);
    \draw [red,myMiddleArrow] (180:1) arc (180:270:1);
    \draw [myMiddleArrowOUT] (-90:1)--(-90:1.75);
    \draw [myMiddleArrow] (-90:1-0.75)--(-90:1); 
    \node [myDot] at (90:1) {};
    \node [myRedDot,rotate=45] at (90:1+0.3) {};
    \node [myWhiteDot] at (90:2) {};
    \node [myRedDot,rotate=45] at (-90:1+0.3) {};
    \node [myRedDot,rotate=-45,label={[label distance=-1.5mm]below :$\scriptscriptstyle{b(\Omega)}$}] at (180:1) {};
    \node [myDot] at (-90:1) {};
\end{scope}
\end{tikzpicture}
=
\tfrac{1}{\langle\omega,x\rangle}
\begin{tikzpicture}[scale=0.5]
\begin{scope}
    \draw [myMiddleArrow,red] (45:1) arc (45:90+45:1);
    \draw [myMiddleArrow,red] (135:1) arc (135:90+135:1);
    \draw [red] (225:1) arc (225:90+225:1);
    \draw [myMiddleArrow,red] (-45:1) arc (-45:90-45:1);
    \draw [myMiddleArrow] (45:1-0.75) -- (45:1); 
    \draw [myMiddleArrowOUT] (45:1) -- (45:1+0.75);
    \draw [myMiddleArrow] (135:1-0.75) -- (135:1); 
    \draw [myMiddleArrowOUT] (135:1) -- (135:1+0.75);
    \draw [myMiddleArrow] (225:1-0.75) -- (225:1); 
    \draw [myMiddleArrowOUT] (225:1) -- (225:1+0.75);
    \draw [myMiddleArrow] (-45:1-0.75) -- (-45:1); 
    \draw [myMiddleArrowOUT] (-45:1) -- (-45:1+0.75);
    \node [myDot] at (135:1) {};
    \node [myRedDot,rotate=45] at (135:1.3) {};
    \node [myDot] at (45:1) {};
    \node [myRedDot,rotate=45] at (45:1.3) {};
    \node [myDot] at (-45:1) {};
    \node [myRedDot,rotate=45] at (-45:1.3) {};
    \node [myDot] at (225:1) {};
    \node [myRedDot,rotate=45] at (225:1.3) {};
    \node [myRedDot,label={[label distance=-0.5mm]below :$\scriptscriptstyle{b(x)}$}] at (-90:1) {};
\end{scope}
\end{tikzpicture},
\end{equation*}
since $\Phi(\xi) = \langle \omega,\Omega\rangle^{-1} 1$ by \cref{lemma:xi}.

Similar to the construction described for the boundary state of the
toric code, the existence of such a
quantum channel immediately induces a finite-depth circuit of quantum channels
manifesting the triviality of these states.

\begin{theorem}\label{thm:Hopf}
Let $A$ be a C*-HA and let $(V,\Phi)$ be a faith\-ful
$*$-re\-pre\-sen\-ta\-tion of $A$. Then, for all positive non-zero elements $x\in A$ and
all $n\in\mathbb{N}$ there exists a depth-two circuit
of bounded-range quantum channels that maps $\mathrm{Tr}(\mathbf{1})^{-n} \mathbf{1}^{\otimes n}$
into $\rho(x,n)$.
That is, the sequence $(\rho(x,n))_{n=1}^\infty$ is in the
trivial phase.
\end{theorem}

\begin{proof}
Assume without loss of generality that $n\in\mathbb{N}$ is even.
The circuit consists of two layers, as presented above in \cref{example:circuitZ2}.
In the first layer, we replace the maximally mixed state
$\mathrm{Tr}(\mathbf{1})^{-n}  \mathbf{1}^{\otimes n}$ with the sequence of $n/2$ tensor products $\rho(\Omega,2)\otimes\cdots\otimes\rho(\Omega,2)$ as previously done.
Now, by virtue of \cref{lemma:glueHA}, let
$\mathrm{Id}\otimes \mathfrak{G}_\Omega\otimes\cdots\otimes \mathfrak{G}_\Omega\otimes \mathfrak{G}_x\otimes\mathrm{Id}$ be the second layer of quantum channels, where all
subindices are $\Omega\in A$ except for one, which is $x\in A$.
This second layer of channels then glues together all local MPDOs into the
single MPDO $\rho(x,n)$.
\end{proof}

For general RFP MPDOs constructed from biconnected C*-WHAs a
straightforward generalization of the previous procedure is
not possible anymore, since the comultiplication is no longer
unit-preserving.

\begin{restatable}{remark}{restatablenogogluing}
\label{rem:nogluing}
There are no trace-preserving gluing maps for general biconnected C*-WHAs
such that \cref{eq:glueHA} holds for all elements $x\in A$.
\end{restatable}

See \cref{sec:proofs-classif-weakhopf} for a proof. Unfortunately, the
description of the phases in this general case is still an open problem.
Nevertheless, some evidence indicate the existence of non-trivial phases,
as we conjecture here.

\begin{conjecture}
RFP MPDOs arising from the Lee-Yang C*-WHA of \cref{ex:Fib1} do not belong to the trivial phase.
\end{conjecture}

However, these obstructions can be circumvented if one restricts to the trivial
sector. The following result
establishes the existence of an special gluing map,
motivated by the characterization of simple RFP MPDO tensors in \cite{cirac_2017_mpdo}.
\begin{restatable}{lemma}{restatableGlueWHA}
\label{lemma:glueVacuum}
Let $A$ be a biconnected C*-WHA and let $(V,\Phi)$ be a faithful
$*$-re\-pre\-sen\-ta\-tion of $A$. There is a quantum channel
$\mathfrak{G}_1\in\mathfrak{L}(\mathfrak{L}(V\otimes V))$,
called ``gluing'' map, such that
\begin{equation}
    \label{eq:glueWHA}
    (\mathrm{Id}^{\otimes m-1}\otimes \mathfrak{G}_1 \otimes\mathrm{Id}^{\otimes n-1})
    (\rho(1,m)\otimes\rho(1,n)) = \rho(1,m+n)
\end{equation}
for all $m,n\in\mathbb{N}$.
\end{restatable}

A proof is given in \cref{sec:proofs-classif-weakhopf}.
As an immediate corollary, similar to the case of C*-HAs, we obtain
the following result.
\begin{theorem}\label{thm:wHopf}
Let $A$ be a biconnected C*-WHA and let $(V,\Phi)$ be a faithful
$*$-re\-pre\-sen\-ta\-tion of $A$. Then, for all $n\in\mathbb{N}$ there exist
two depth-two circuits of bounded-range
quantum channels that map
$\mathrm{Tr}(\mathbf{1})^{-n} \mathbf{1}^{\otimes n}$
into $\rho(1,n)$ and $\rho(\mathrm{Tr}^1,n)$.
That is, the sequences $(\rho(1,n))_{n=1}^\infty$ and $(\rho(\mathrm{Tr}^1,n))_{n=1}^\infty$ are in the trivial
phase.
\end{theorem}

\section*{Acknowledgements}

This work has received support from the European Union’s Horizon 2020 program
through the ERC CoG SEQUAM (No. 863476) and the ERC CoG GAPS (No. 648913),
from the Spanish Ministry of Science and Innovation 
through the Agencia Estatal de Investigación MCIN/AEI/10.13039/501100011033
(PID2020-113523GB-I00 and grant BES-2017-081301 under the ``Severo Ochoa Programme
for Centres of Excellence in R\&D'' CEX2019-000904-S and ICMAT Severo Ochoa
project SEV-2015-0554), from CSIC Quantum Technologies Platform PTI-001, from
Comunidad Autónoma de Madrid through the grant QUITEMAD-CM (P2018/TCS-4342).


\appendix

\section{Graphical notation for tensor networks}
\label{sec:graphical}

In this appendix we introduce a slightly adapted version of the usual graphical notation from the
literature of tensor networks. From now on, let $V$ be a finite dimensional complex vector space.
First, a vector $v\in V$ is depicted by a shape (e.g. a circle) and a line sticking out of it,
associated to the vector space and labeled accordingly; we convey to draw the arrow outgoing from
the shape. Second, any element $f\in V^*$ is a vector from $V^*$ or a linear
functional on $V$, and we alternatively draw it with an arrow ingoing to the shape. Pictorially,
\begin{equation*}
    v = 
    \begin{tikzpicture}[baseline = {(0,-0.5ex)}]
        \node [circle, draw = white, fill = white, inner sep = 1pt, label = above:\color{black!50}$\scriptstyle{V\vphantom{f}}$] at (0.5,0) {};
        \draw [myMiddleArrow] (0,0) -- (1,0);
        \node [myDot, label = above:$\scriptstyle{v\vphantom{f}}$] at (0,0) {};
    \end{tikzpicture}
    \quad\text{ and }\quad
    f =
    \begin{tikzpicture}[baseline = {(0,-0.5ex)}]
        \node [circle, draw = white, fill = white, inner sep = 1pt, label = above:\color{black!50}$\scriptstyle{V^*\vphantom{f}}$] at (0.5,0) {};
        \draw [myMiddleArrow] (1,0) -- (0,0);
        \node [myDot, label = above:$\scriptstyle{f}$] at (1,0) {};
    \end{tikzpicture}
    =
    \begin{tikzpicture}[baseline = {(0,-0.5ex)}]
        \node [circle, draw = white, fill = white, inner sep = 1pt, label = above:\color{black!50}$\scriptstyle{V\vphantom{}}$] at (0.5,0) {};
        \draw [myMiddleArrow] (0,0) -- (1,0);
        \node [myDot, label = above:$\scriptstyle{f}$] at (1,0) {};
    \end{tikzpicture}.
\end{equation*}
Note that, although we have represented the previous elements using straight horizontal lines, no
reading order has been prescribed. Here, we identify each (finite dimensional complex) vector space
$V$ and its double dual $V^{**}$, which are canonically isomorphic. Therefore, one can regard any
vector $v\in V$ as a linear functional on $V^*$ and thus, exchange the orientation of any line by
labeling the dual:
\begin{equation*}
    \begin{tikzpicture}[baseline = {(0,-0.5ex)}]
        \node [circle, draw = white, fill = white, inner sep = 1pt, label = above:\color{black!50}$\scriptstyle{V\vphantom{f}}$] at (0.5,0) {};
        \draw [myMiddleArrow] (0,0) -- (1,0);
        \node [myDot, label = above:$\scriptstyle{v\vphantom{f}}$] at (0,0) {};
    \end{tikzpicture}
    =
    \begin{tikzpicture}[baseline = {(0,-0.5ex)}]
        \node [circle, draw = white, fill = white, inner sep = 1pt, label = above:\color{black!50}$\scriptstyle{V^{**}\vphantom{f}}$] at (0.5,0) {};
        \draw [myMiddleArrow] (0,0) -- (1,0);
        \node [myDot, label = above:$\scriptstyle{v\vphantom{f}}$] at (0,0) {};
    \end{tikzpicture}
    =
    \begin{tikzpicture}[baseline = {(0,-0.5ex)}]
        \node [circle, draw = white, fill = white, inner sep = 1pt, label = above:\color{black!50}$\scriptstyle{V^{*}\vphantom{f}}$] at (0.5,0) {};
        \draw [myMiddleArrow] (1,0) -- (0,0);
        \node [myDot, label = above:$\scriptstyle{v\vphantom{f}}$] at (0,0) {};
    \end{tikzpicture}.
\end{equation*}
Now, let $V$ and $W$ be two (finite dimensional complex) vector spaces. Any vector $v\in V\otimes W$
in the tensor product is depicted by a shape with two lines, e.g.:
\begin{equation*}
    \begin{tikzpicture}[baseline = {(0,-0.5ex)}]
        \node [circle, draw = white, fill = white, inner sep = 1pt, label = above:\color{black!50}$\scriptstyle{V\vphantom{f}}$] at (-0.5,0) {};
        \node [circle, draw = white, fill = white, inner sep = 1pt, label = above:\color{black!50}$\scriptstyle{W\vphantom{f}}$] at (0.5,0) {};
        \draw [myMiddleArrow] (-0.8,0) -- (0,0);
        \draw [myMiddleArrow] (0.8,0)  -- (0,0);
        \node [myDot, label = above:$\scriptstyle{v\vphantom{f}}$] at (0,0) {};
    \end{tikzpicture}.
\end{equation*}
By virtue of the previous identification, one can rewrite
\begin{equation*}
    \begin{tikzpicture}[baseline = {(0,-0.5ex)}]
        \node [circle, draw = white, fill = white, inner sep = 1pt, label = above:\color{black!50}$\scriptstyle{V\vphantom{f}}$] at (-0.5,0) {};
        \node [circle, draw = white, fill = white, inner sep = 1pt, label = above:\color{black!50}$\scriptstyle{W\vphantom{f}}$] at (0.5,0) {};
        \draw [myMiddleArrow] (-0.8,0) -- (0,0);
        \draw [myMiddleArrow] (0.8,0)  -- (0,0);
        \node [myDot, label = above:$\scriptstyle{v\vphantom{f}}$] at (0,0) {};
    \end{tikzpicture}
=
    \begin{tikzpicture}[baseline = {(0,-0.5ex)}]
        \node [circle, draw = white, fill = white, inner sep = 1pt, label = above:\color{black!50}$\scriptstyle{V^*\vphantom{f}}$] at (-0.5,0) {};
        \node [circle, draw = white, fill = white, inner sep = 1pt, label = above:\color{black!50}$\scriptstyle{W\vphantom{f}}$] at (0.5,0) {};
        \draw [myMiddleArrow] (0,0) -- (-0.8,0);
        \draw [myMiddleArrow] (0.8,0)  -- (0,0);
        \node [myDot, label = above:$\scriptstyle{v\vphantom{f}}$] at (0,0) {};
    \end{tikzpicture}
=
    \begin{tikzpicture}[baseline = {(0,-0.5ex)}]
        \node [circle, draw = white, fill = white, inner sep = 1pt, label = above:\color{black!50}$\scriptstyle{V^*\vphantom{f}}$] at (-0.5,0) {};
        \node [circle, draw = white, fill = white, inner sep = 1pt, label = above:\color{black!50}$\scriptstyle{W^*\vphantom{f}}$] at (0.5,0) {};
        \draw [myMiddleArrow] (0,0) -- (-0.8,0);
        \draw [myMiddleArrow] (0,0)  -- (0.8,0);
        \node [myDot, label = above:$\scriptstyle{v\vphantom{f}}$] at (0,0) {};
    \end{tikzpicture}
    =
    \begin{tikzpicture}[baseline = {(0,-0.5ex)}]
        \node [circle, draw = white, fill = white, inner sep = 1pt, label = above:\color{black!50}$\scriptstyle{V\vphantom{f}}$] at (-0.5,0) {};
        \node [circle, draw = white, fill = white, inner sep = 1pt, label = above:\color{black!50}$\scriptstyle{W^*\vphantom{f}}$] at (0.5,0) {};
        \draw [myMiddleArrow] (-0.8,0) -- (0,0);
        \draw [myMiddleArrow] (0,0)  -- (0.8,0);
        \node [myDot, label = above:$\scriptstyle{v\vphantom{f}}$] at (0,0) {};
    \end{tikzpicture}.
\end{equation*}
In addition, the tensor product of two vectors $v\in V$ and $w\in W$ is depicted by placing both
representations in the same picture, next to each other:
\begin{equation*}
    \begin{tikzpicture}[baseline = {(0,-0.5ex)}]
        \node [circle, draw = white, fill = white, inner sep = 1pt, label = above:\color{black!50}$\scriptstyle{V\vphantom{f}}$] at (-0.5,0) {};
        \node [circle, draw = white, fill = white, inner sep = 1pt, label = above:\color{black!50}$\scriptstyle{W\vphantom{f}}$] at (0.5,0) {};
        \draw [myMiddleArrow] (-0.8,0) -- (0,0);
        \draw [myMiddleArrow] (0.8,0)  -- (0,0);
        \node [circle, draw = black, fill = black, inner sep = 1pt, label = above:$\scriptstyle{v\otimes w\vphantom{f}}$] at (0,0) {};
    \end{tikzpicture}
    =
    \begin{tikzpicture}[baseline = {(0,-0.5ex)}]
        \node [circle, draw = white, fill = white, inner sep = 1pt, label = above:\color{black!50}$\scriptstyle{V\vphantom{f}}$] at (-0.7,0) {};
        \node [circle, draw = white, fill = white, inner sep = 1pt, label = above:\color{black!50}$\scriptstyle{W\vphantom{f}}$] at (0.7,0) {};
        \draw [myMiddleArrow] (-1,0) -- (-0.2,0);
        \draw [myMiddleArrow] (1,0)  -- (0.2,0);
        \node [myDot, label = above:$\scriptstyle{v\vphantom{f}}$] at (-0.2,0) {};
        \node [myDot, label = above:$\scriptstyle{w\vphantom{f}}$] at (0.2,0) {};
    \end{tikzpicture}.
\end{equation*}
Note that, although the labels $V$ and $W$ should allow to identify the
corresponding indices, even if $V = W$, it may not be enough. One solution,
used in this paper, consists on prescribing colors to each index. Additionally,
note that we have implicitly identified $V\otimes W$ and $W\otimes V$, since no
order is prescribed in the previous picture. However, these are again
canonically isomorphic and no preference is considered here. In this context,
it is natural to represent the action of the canonical pairing by joining the
corresponding lines, i.e.
\begin{equation*}
    f\otimes v =
    \begin{tikzpicture}[baseline = {(0,-0.5ex)}]
        \node [circle, draw = white, fill = white, inner sep = 1pt, label = above:\color{black!50}$\scriptstyle{V\vphantom{f}}$] at (0.4,0) {};
        \draw [myMiddleArrow] (0,0) -- (0.8,0);
        \node [circle, draw = black, fill = black, inner sep = 1pt, label = above:$\scriptstyle{v\vphantom{f}}$] at (0,0) {};
        \begin{scope}[xshift = +1cm]
            \node [circle, draw = white, fill = white, inner sep = 1pt, label = above:\color{black!50}$\scriptstyle{V\vphantom{f}}$] at (0.6,0) {};
            \draw [myMiddleArrow] (0.2,0) -- (1,0);
            \node [circle, draw = black, fill = black, inner sep = 1pt, label = above:$\scriptstyle{f}$] at (1,0) {};
        \end{scope}
    \end{tikzpicture}
    \mapsto
    \begin{tikzpicture}[baseline = {(0,-0.5ex)}]
        \node [circle, draw = white, fill = white, inner sep = 1pt, label = above:\color{black!50}$\scriptstyle{V\vphantom{f}}$] at (0.5,0) {};
        \draw [myMiddleArrow] (0,0) -- (1,0);
        \node [circle, draw = black, fill = black, inner sep = 1pt, label = above:$\scriptstyle{v\vphantom{f}}$] at (0,0) {};
        \node [circle, draw = black, fill = black, inner sep = 1pt, label = above:$\scriptstyle{f}$] at (1,0) {};
    \end{tikzpicture}
    =\langle f,v\rangle.
\end{equation*}
On the other hand, since $\mathfrak{L}(V,W)$ is canonically isomorphic to $V^*\otimes W$, we can also represent any linear map $F\in\mathfrak{L}(V,W)$ in the following form
\begin{equation*}
    \begin{tikzpicture}[baseline = {(0,-0.5ex)}]
        \node [circle, draw = white, fill = white, inner sep = 1pt, label = above:\color{black!50}$\scriptstyle{V^*\vphantom{f}}$] at (-0.5,0) {};
        \node [circle, draw = white, fill = white, inner sep = 1pt, label = above:\color{black!50}$\scriptstyle{W\vphantom{f}}$] at (0.5,0) {};
        \draw [myMiddleArrow] (-0.8,0) -- (0,0);
        \draw [myMiddleArrow] (0.8,0)  -- (0,0);
        \node [circle, draw = black, fill = black, inner sep = 1pt, label = above:$\scriptstyle{F\vphantom{f}}$] at (0,0) {};
    \end{tikzpicture}
    =
    \begin{tikzpicture}[baseline = {(0,-0.5ex)}]
        \node [circle, draw = white, fill = white, inner sep = 1pt, label = above:\color{black!50}$\scriptstyle{V\vphantom{f}}$] at (-0.5,0) {};
        \node [circle, draw = white, fill = white, inner sep = 1pt, label = above:\color{black!50}$\scriptstyle{W\vphantom{f}}$] at (0.5,0) {};
        \draw [myMiddleArrow] (0,0) -- (-0.8,0);
        \draw [myMiddleArrow] (0.8,0)  -- (0,0);
        \node [circle, draw = black, fill = black, inner sep = 1pt, label = above:$\scriptstyle{F\vphantom{f}}$] at (0,0) {};
    \end{tikzpicture}.
\end{equation*}
Remarkably, the distinction between $F\in\mathfrak{L}(V,W)$ and its transpose $F^{\mathrm{t}}\in\mathfrak{L}(W^*,V^*)$, $W^*\ni g\mapsto g\circ F\in V^*$, is only reflected in the diagram by the arrows
and their labels.

\section{The Canonical Regular Element}
\label{appendix:cre}

Here we recall additional results on the framework
of C*-WHA not introduced in the main text. As a matter of fact, we are
interested in describing the canonical regular element in terms of these.
First, it is well-known that in any C*-WHA ${A}$ there exists a unique
non-degenerate two-sided normalized integral $h\in {A}$, known as the \emph{Haar integral of ${A}$};
see Definition~3.24 and Theorem~4.5 in \cite{bohm_1999_weak}. In particular, 
\begin{equation}
h^2 = h^* = h = S(h).
\end{equation}
By self-duality, let $\hat{h}\in {A}^*$ denote the Haar integral of the dual C*-WHA.  We also recall the existence of $\Lambda\in {A}$, known as the \emph{dual left-integral of $\hat{h}$}, such that
\begin{equation}\label{eq:dualLeftIntHaar}
\langle\hat{h},\Lambda_{(1)}\rangle \Lambda_{(2)} = 1
\;\,\text{ and }\;\,
S(\Lambda_{(1)}) \otimes\Lambda_{(2)} =
\Lambda_{(2)}\otimes \Lambda_{(1)};
\end{equation}
see e.g. Theorem~3.18 and Lemma~3.20 in \cite{bohm_1999_weak}.
Second, there is a unique positive element $g\in {A}$ implementing the antipode squared as an inner automorphism, i.e.
\begin{equation}\label{eq:S2gxg1}
    S^2(x) = g x g^{-1}
\end{equation}
for all elements $x\in{A}$,
among other properties, known as the \emph{canonical group-like element of ${A}$};
see Proposition~4.9 in \cite{bohm_1999_weak}.
As its name implies, it is a \emph{group-like element}, i.e. it satisfies the following property:
\begin{equation}
    g_{(1)}\otimes g_{(2)} = g 1_{(1)}\otimes g1_{(2)} = 1_{(1)}g\otimes 1_{(2)}g.
\end{equation}
Moreover, it can be decomposed in the form $g = g_L g_R^{-1}$ for two $g_L,g_R > 0$ given by
\begin{equation}\label{eq:defgL}
    g_L:= (\langle\hat{h},h_{(1)}\rangle h_{(2)})^{\frac{1}{2}}\in {A}_L
    \;\,\text{ and }\;\,
    g_R := S(g_L) = S^{-1}(g_L) \in {A}_R.
\end{equation}
By self-duality, we denote by $\hat{g}\in {A}^*$ the canonical group-like element of the
dual C*-WHA. Finally, let us recall the following formula.
\begin{proposition}
\label{prop:arrowsghat}
For any C*-WHA,
\begin{equation*}
    x_{(1)}\langle\hat{g},x_{(2)}\rangle = g_R x g_R^{-1}
    \;\;\text{ and }\;\;
    \langle\hat{g},x_{(1)}\rangle x_{(2)} = g_L x g_L^{-1}
\end{equation*}
for all elements $x\in A$. In particular,
$$
    1_{(1)}\langle\hat{g},1_{(2)}\rangle = 1 = \langle\hat{g},1_{(1)}\rangle 1_{(2)}.
$$
\end{proposition}


\begin{proof}
See Scholium~2.7 and Lemma~4.13 in \cite{bohm_1999_weak} for a proof.
\end{proof}


\begin{proposition}[see~\cite{nikshych_2004_semisimple}]
\label{prop:chiagda}
For any connected C*-WHA ${A}$,
\begin{equation*}
    \mathrm{Tr}_\alpha(g) = \langle\varepsilon,1\rangle d_\alpha 
\end{equation*}
for all sectors $\alpha\in\mathrm{Irr}(A)$.
\end{proposition}



\begin{proposition}\label{prop:OmegaIsTau}
In any connected C*-WHA
$$
    \langle\omega,x\rangle =
    \mathfrak{D}^{-2}\varepsilon(1)^{-1} \langle\hat{h},g_L^{-1} g_R^{-1} x\rangle
    =
    \mathfrak{D}^{-2}\varepsilon(1)^{-1} \langle\hat{h},x g_L^{-1} g_R^{-1}\rangle
$$
for all elements $x\in {A}$. Equivalently,
for any coconnected C*-WHA,
$$
\Omega = \mathfrak{D}^{-2}\varepsilon(1)^{-1} \Lambda_{(1)} \langle\hat{g}^{-1},\Lambda_{(2)}\rangle
= \mathfrak{D}^{-2}\varepsilon(1)^{-1} \langle\hat{g}^{-1},\Lambda_{(1)}\rangle \Lambda_{(2)}.
$$
\end{proposition}


\begin{proof}
Assume first that ${A}$ is a connected C*-WHA. There exists a well-known
element, called the $S$-invariant trace of ${A}$, see~\cite{bohm_1999_weak}, given by the expression
$\sum_\alpha\mathrm{Tr}_\alpha(g)\mathrm{Tr}_\alpha$.
By virtue of \cref{thm:Omega} and \cref{prop:chiagda}, one easily checks that both
elements are proportional.
\end{proof}


\begin{remark}[see~\cite{molnar_2022_mpo}]
\label{prop:expDefT}
The linear map $T\in\mathfrak{L}({A})$ in \cref{thm:Omega} is given by
\begin{equation*}
T(x) = S(x_{(1)})\langle \hat{g},x_{(2)}\rangle  = \langle\hat{g},x_{(1)}\rangle S^{-1}(x_{(2)})
\end{equation*}
\end{remark}


\begin{remark}[see~\cite{molnar_2022_mpo}]
\label{remark:wSwwT}
In any coconnected C*-WHA, $\omega \circ T =\omega\circ S = \omega$.
\end{remark}


Finally, let us particularize the previous notions and results in the context of
C*-Hopf algebras. We refer the reader to \cite{montgomery_2001_rep} for more details.

\begin{proposition}\label{prop:someOnHopf}
Let ${A}$ be a C*-HA. Then:
\begin{enumerate}
    \item $S^2 = \mathrm{Id}$ and the canonical grouplike element is $g = 1$;
    \item $d_\alpha = \dim_{\mathbb{C}}(V_\alpha)$ for all sectors $\alpha\in\mathrm{Irr}(A)$;
    \item the dual left integral of the Haar measure $\hat{h}\in A^*$ is $t = \mathfrak{D}^2 \Omega$;
    \item the canonical regular element and the Haar integral coincide, i.e. $\Omega = h$;
    \item the linear map $T\in\mathfrak{L}({A})$ coincides with the antipode $S\in\mathfrak{L}({A})$;
    \item $g_L = g_R = \mathfrak{D}^{-1} 1$.
\end{enumerate}
\end{proposition}


\begin{proof}
(1) It is well-known that for any C*-HA it holds that
$S^2 = \mathrm{Id}$ \cite{larson_1988_semisimple,larson_1988_finite}. Since the unit
element $1\in{A}$ satisfies
the defining properties of the canonical group-like element too, which is unique, we can
conclude that $g = 1$. (2) Consider that $\varepsilon(1) = 1$ by \cref{def:CHA} and hence
\cref{prop:chiagda} proves that
$\dim_{\mathbb{C}}(V_\alpha) = \mathrm{Tr}_\alpha(1) = \mathrm{Tr}_\alpha(g) =
\varepsilon(1) d_\alpha = d_\alpha$
for all sectors $\alpha \in\mathrm{Irr}(A)$.
(3) Since the axioms of C*-HAs are self-dual $\hat{g} = \varepsilon$ and hence
$\Omega = \mathfrak{D}^{-2}\varepsilon(1)^{-1} t = \mathfrak{D}^{-2} t$, where the first expression
follows from \cref{prop:OmegaIsTau}.
(4) Every C*-HA is \emph{unimodular}, see~\cite{montgomery_2001_rep}, i.e. every left integral is
a two-sided integral, and the subspace of two-sided integrals is one-dimensional.
Hence $t = \eta h$ for some
$\eta\in\mathbb{C}$. Since $\Omega^2 = \Omega$ and $h^2 = h$, the only possibility left is
$\eta = \mathfrak{D}^{2}$. (5) This follows trivially as a consequence of \cref{prop:expDefT} since
$\hat{g} =\varepsilon$. (6) Recall the definition of $g_L$ and $g_R$ in \cref{eq:defgL} and consider
both steps (3) and (4).
\end{proof}

\section{Proof of \texorpdfstring{\cref{lemma:xi}}{Lemma~3.6}}
\label{appendix:lemmaxi}

Here we restate and prove the following result.

\restatablelemmaxi*

\begin{proof}
Since $\Omega\in A$ is non-degenerate there
exists a linear functional $f\in {A}^*$ such that
$\langle f,\Omega_{(1)}\rangle \Omega_{(2)} = 1$;
see \cref{eq:nondegenerate} and \cref{thm:Omega}.
On the other hand, since $\omega\in {A}^*$ is non-degenerate, there exists an
element $\xi\in A$ such that
$\langle \omega,\xi x\rangle = \langle f \circ T,x\rangle$ for all elements
$x\in {A}$. Therefore,
\begin{equation*}
    \langle \omega , \xi T( \Omega_{(1)}) \rangle    \Omega_{(2)} =
    \langle f\circ T , T(\Omega_{(1)}) \rangle \Omega_{(2)} =
    \langle f ,\Omega_{(1)} \rangle \Omega_{(2)} = 1,
\end{equation*}
where in the second equality we have used that $T\in\mathfrak{L}(A)$ is involutive,
i.e. $T\circ T = \mathrm{Id}$.
Recall that $\omega\in A^*$ is cocentral, it is a trace-like linear functional
of $A^*$, see \cref{remark:dualcwha}. It follows by the pulling-through identity in \cref{eq:pt}
that
\begin{equation*}
    1 = \langle\omega,\xi T(\Omega_{(1)})\rangle\Omega_{(2)} = \langle\omega,T(\Omega_{(1)})\xi\rangle \Omega_{(2)} = \langle\omega,T(\Omega_{(1)})\rangle \xi\Omega_{(2)},
\end{equation*}
and hence $\xi\in{A}$ is invertible. Its inverse is then trivially given by the expression
\begin{equation*}
\xi^{-1} = \langle\omega,T(\Omega_{(1)})\rangle\Omega_{(2)} = \langle\omega,\Omega_{(1)}\rangle\Omega_{(2)}.
\end{equation*}
where the last equality follows from \cref{remark:wSwwT}.
Let us prove now (4).
By virtue of \cref{prop:OmegaIsTau}, it is easy to conclude by its defining property $\langle\omega,\xi T(\Omega_{(1)})\rangle\Omega_{(2)} = 1$ that $\xi\in {A}$ is necessarily given by the expression
\begin{equation}\label{eq:defExpXi}
 \xi  = \mathfrak{D}^4\varepsilon(1)^2 g_L g_R.
\end{equation}
Consequently, a natural choice of positive elements $\xi_L\in {A}_L$ and
$\xi_R\in {A}_R$ is
\begin{equation}\label{eq:defExpXiLXiR}
    \xi_L := \mathfrak{D}^2\varepsilon(1) g_L
    \;\,\text{ and }\;\,
    \xi_R := \mathfrak{D}^2\varepsilon(1) g_R = S(\xi_L).
\end{equation}
Since $g_L,g_R > 0$, $\xi$ is strictly positive, as we wanted to prove.
We now prove (2), i.e. that $ T(\xi) = \xi $,
note by the previous expressions that it turns out to be enough to check that $T(g_L) = g_R$ and
$T(g_R) = g_L$. We refer to \cref{eq:TxL01,eq:TyR01} below for elementary proofs
of these facts. In addition, note that (4) is straightforward by the eigenvalue
equation \cref{eq:OmegaEigen}. See Scholium~2.7 and Lemma~4.13 from \cite{bohm_1999_weak}
for a proof of (6). Let us now move to the proof of (3). For simplicity, we prove
the equivalent formula $\xi T(x)\xi^{-1} = T(x^*)^*$ for all $x\in {A}$. To this end, we
 recall first that
\begin{equation}
    \label{eq:xiyxiInv}
    \xi y \xi^{-1} =
    g_L g_R y g_L^{-1} g_R^{-1} =
    \langle\hat{g},y_{(1)}\rangle y_{(2)}\langle\hat{g},y_{(3)}\rangle
\end{equation}
for all elements $y\in A$,
see \cref{prop:arrowsghat}.
On the other hand, by virtue of the the fact that $S^{-1}(\hat{g}) = \hat{g}^{-1}$
and the positivity of $\hat{g}\in {A}^*$,
\begin{equation}
    \label{eq:ghatinvis}
    \langle\hat{g}^{-1},y\rangle      =
    \langle\hat{g},S^{-1}(y)\rangle   =
    \langle\hat{g}^*,S^{-1}(y)\rangle =
    \overline{ \langle\hat{g},S(S^{-1}(y))^*\rangle } =
    \overline{ \langle\hat{g},y^*\rangle}
\end{equation}
for all elements $y\in A$.
Thus,
\begin{align*}
    \xi T(x)\xi^{-1}
     &  =
    \xi S(x_{(1)}) \xi^{-1} \langle\hat{g},x_{(2)}\rangle
    \vphantom{\overline{ (x_{(2)}^*) }}
    &&\text{by~\cref{prop:expDefT}}
    \\ &= 
    \langle\hat{g},S(x_{(1)})_{(1)}\rangle S(x_{(1)})_{(2)} \langle\hat{g}, S(x_{(1)})_{(3)}\rangle \langle\hat{g},x_{(2)}\rangle
    \vphantom{\overline{ (x_{(2)}^*) }}
    &&\text{by~\cref{eq:xiyxiInv}}
    \\   & =
    \langle\hat{g},S(x_{(3)})\rangle S(x_{(2)}) \langle\hat{g},S(x_{(1)})\rangle \langle\hat{g},x_{(4)}\rangle
    \vphantom{\overline{ (x_{(2)}^*) }}
    &&\text{by~\cref{remark:propsS}}
    \\  &=
    \langle\hat{g}^{-1},x_{(1)}\rangle S(x_{(2)}) \langle\hat{g}^{-1},x_{(3)}\rangle \langle\hat{g},x_{(4)}\rangle
    \vphantom{\overline{ (x_{(2)}^*) }}
    &&{\text{by~\cref{eq:defgL}}}
    \\  &=
    \langle\hat{g}^{-1},x_{(1)}\rangle S(x_{(2)})
    \vphantom{\overline{ (x_{(2)}^*) }}
    &&{\text{by~\cref{remark:dualcwha}}}
    \\  &=
    \langle\hat{g}^{-1},x_{(1)}\rangle S^{-1}(S^2(x_{(2)}))
    \vphantom{\overline{ (x_{(2)}^*) }}
    \\  &=
    \langle\hat{g}^{-1},x_{(1)}\rangle \langle\hat{g},x_{(2)}\rangle S^{-1}(x_{(3)}) \langle\hat{g}^{-1},x_{(4)}\rangle
    \vphantom{\overline{ (x_{(2)}^*) }}
    &&{\text{by~\cref{eq:S2gxg1}}}
    \\  &=
    S^{-1}(x_{(1)}) \langle\hat{g}^{-1},x_{(2)}\rangle
    \vphantom{\overline{ (x_{(2)}^*) }}
    &&{\text{by~\cref{remark:dualcwha}}}
    \\  &=
    S^{-1}(x_{(1)}) \overline{\langle \hat{g},x_{(2)}^*\rangle}
    \vphantom{\overline{ (x_{(2)}^*) }}
    &&\text{by~\cref{eq:ghatinvis}}
    \\  &=
    S(x_{(1)}^*)^* \overline{\langle \hat{g},x_{(2)}^*\rangle}
    \vphantom{\overline{ (x_{(2)}^*) }}
    &&\text{by~\cref{remark:propsS}}
    \\  &=
    S((x^*)_{(1)})^* \overline{ \langle\hat{g},(x^*)_{(2)} \rangle }
    \vphantom{\overline{ (x_{(2)}^*) }}
    &&\text{by~\cref{def:Cwha}}
    \\ &=
    T(x^*)^*
    \vphantom{\overline{ (x_{(2)}^*) }},
    &&\text{by \cref{prop:expDefT}}
\end{align*}
for all elements $x\in{A}$, as we wanted to prove.
Finally, if ${A}$ is a C*-HA, it is already known by \cref{prop:someOnHopf} that
$g_L = g_R = \mathfrak{D}^{-1} 1$. This, together with the definition of $\xi\in{A}$
in \cref{lemma:xi} and the fact that $\varepsilon(1) = 1$, leads to the expressions
$\xi_L = \xi_R = \mathfrak{D} 1$ and $\xi = \mathfrak{D}^2 1$, as we wanted to prove.
\end{proof}

\section{Proof of \texorpdfstring{\cref{thm:rfp}}{Theorem~3.3}}
\label{appendix:proofrfp}

We now provide algebraic explicit expressions for both local coarse-graining and
fine-graining quantum channels. We restate and prove the following theorem now.

\restatablethmrfp*

\begin{proof}
As previously done, let us define the local coarse-graining quantum channel
\begin{equation}\label{eq:rfpT}
\mathfrak{T}(X):= \mathrm{Tr}(  \Phi( \xi T( \Omega_{(1)} ) )   X )
\Phi(c_\omega \Omega_{(2)} ) \otimes\Phi(c_\omega \Omega_{(3)} )
\end{equation}
for all $X\in\mathfrak{L}(V)$.
First, let us check that $\mathfrak{T}(\rho(x,1)) = \rho(x,2)$
for all positive non-zero $x\in A$. Indeed,
\begin{align*}
    \mathfrak{T}(\rho(x,1)) 
    & = 
    \tfrac{1}{\langle \omega,x\rangle}\mathrm{Tr}(\Phi(\xi T(\Omega_{(1)}) c_\omega x )) 
    \Phi(c_\omega \Omega_{(2)}) \otimes
    \Phi(c_\omega \Omega_{(3)})
    \\ &=
    \tfrac{1}{\langle \omega,x\rangle}\langle\omega , \xi T(\Omega_{(1)})x\rangle
    \Phi(c_\omega \Omega_{(2)} ) \otimes
    \Phi(c_\omega \Omega_{(3)} )
    &&\text{by~\cref{remark:phicwext}}
    \\  &=
    \tfrac{1}{\langle \omega,x\rangle}\langle\omega, \xi T(\Omega_{(1)})\rangle
    \Phi(c_\omega x_{(1)}\Omega_{(2)} ) \otimes
    \Phi(c_\omega x_{(2)}\Omega_{(3)} )
     &&\text{by~\cref{eq:pt2}}
    \\ &=
    \tfrac{1}{\langle \omega,x\rangle}\Phi(c_\omega x_{(1)}1_{(1)} ) \otimes
    \Phi(c_\omega x_{(2)}1_{(2)} )
    &&\text{by~\cref{lemma:xi}}
    \\ &=
    \tfrac{1}{\langle \omega,x\rangle}\Phi(c_\omega x_{(1)} ) \otimes
    \Phi(c_\omega x_{(2)}) = \rho(x,2)
    &&\text{by~\cref{def:Cwha}}
    \end{align*}
    Second, this map is trace-preserving:
    \begin{align*}
    \mathrm{Tr}(\mathfrak{T}(X)) 
    &=
    \mathrm{Tr}(\Phi(\xi T(\Omega_{(1)})) X) \mathrm{Tr}(\Phi(c_\omega\Omega_{(2)}))
    \mathrm{Tr}(\Phi(c_\omega\Omega_{(3)}))
    \vphantom{\xi^{-1}}
    \\  &=
    \mathrm{Tr}(\Phi(\xi T(\Omega_{(1)})) X) \langle\omega,\Omega_{(2)}\rangle
    \langle\omega,\Omega_{(3)}\rangle
    \vphantom{\xi^{-1}}
     &&\text{by~\cref{remark:phicwext}}
    \\  &=
    \mathrm{Tr}(\Phi(\xi T(\Omega_{(1)})) X) \langle\omega,\Omega_{(2)}\rangle
    \vphantom{\xi^{-1}}
     &&\text{by~\cref{lemma:omega}}
    \\  &=
    \mathrm{Tr}(\Phi(\xi T(\xi^{-1})) X)
     &&\text{by~\cref{lemma:xi}}
    \\  &=
    \mathrm{Tr}(\Phi(\xi \xi^{-1}) X)
    \mathrm{Tr}(X)
     &&\text{by~\cref{lemma:xi}}
\end{align*}
Finally, since $\Omega = \Omega^2 = \Omega \Omega^*$ (in fact, only positivity of $\Omega$ is needed), we can rewrite the map in the following form:
\begin{align*}
    \mathfrak{T}(X)
    &=
    \mathrm{Tr}(\Phi(\xi T(\Omega_{(1)}(\Omega^*)_{(1')}) X)
    \Phi^{\otimes 2}(c_\omega^{\otimes 2} \Delta(\Omega_{(2)}(\Omega^*)_{(2')}))
    \\ &=
    \mathrm{Tr}(\Phi(\xi T(\Omega_{(1)}(\Omega^*)_{(1')}) X)
    \Phi^{\otimes 2}(c_\omega^{\otimes 2} \Delta(\Omega_{(2)})\Delta((\Omega^*)_{(2')}))
    \\  &=
    \mathrm{Tr}(\Phi(\xi T(\Omega_{(1)}\Omega_{(1')}^*) X)
    \Phi^{\otimes 2}(c_\omega^{\otimes 2} \Delta(\Omega_{(2)})\Delta(\Omega_{(2')})^*)
    \\ &=
    \mathrm{Tr}(\Phi(\xi T(\Omega_{(1')}^*) T(\Omega_{(1)})) X)
    \Phi^{\otimes 2}(c_\omega^{\otimes 2} \Delta(\Omega_{(2)})\Delta(\Omega_{(2')})^*)
    \\ &=
    \mathrm{Tr}(\Phi( T(\Omega_{(1')})^*\xi T(\Omega_{(1)})) X)
    \Phi^{\otimes 2}(c_\omega^{\otimes 2} \Delta(\Omega_{(2)})\Delta(\Omega_{(2')})^*)
    \\ &= 
    \mathrm{Tr}(\Phi( T(\Omega_{(1')})^*\xi^{\frac{1}{2}}\xi^{\frac{1}{2}} T(\Omega_{(1)})) X)
    \Phi^{\otimes 2}((c_\omega^{\frac{1}{2}})^{\otimes 2} \Delta(\Omega_{(2)})\Delta(\Omega_{(2')})^*(c_\omega^{\frac{1}{2}})^{\otimes 2})
    \\  &=
    \mathrm{Tr}(\Phi( T(\Omega_{(1')})^*\xi^{\frac{1}{2}})\Phi(\xi^{\frac{1}{2}} T(\Omega_{(1)})) X)
    \Phi^{\otimes 2}((c_\omega^{\frac{1}{2}})^{\otimes 2} \Delta(\Omega_{(2)}))
    \Phi^{\otimes 2}(\Delta(\Omega_{(2')})^*(c_\omega^{\frac{1}{2}})^{\otimes 2})
    \\ &=
    \mathrm{Tr}(\Phi(\xi^{\frac{1}{2}} T(\Omega_{(1)})) X \Phi (\xi^{\frac{1}{2}}T(\Omega_{(1')}))^\dagger)
    \Phi^{\otimes 2}((c_\omega^{\frac{1}{2}})^{\otimes 2} \Delta(\Omega_{(2)}))\Phi^{\otimes 2}((c_\omega^{\frac{1}{2}})^{\otimes 2}\Delta(\Omega_{(2')}))^\dagger
    \\ &=
    (\mathrm{Tr}\otimes\mathrm{Id}\otimes\mathrm{Id})(Q(X\otimes\mathbf{1}\otimes\mathbf{1}) Q^\dagger)
\end{align*}
where
\begin{equation*}
    Q := \Phi^{\otimes 3}(\xi^{\frac{1}{2}}T(\Omega_{(1)}) \otimes c_\omega^{\frac{1}{2}}\Omega_{(2)}\otimes c_\omega^{\frac{1}{2}} \Omega_{(3)})
\end{equation*}
Thus, $\mathfrak{T}$ is completely positive.
Now, let us define a local fine-graining quantum channel $\mathfrak{S}$.
Consider first the following hermitian projectors
\begin{equation}\label{eq:SPPperp}
    P := \Phi^{\otimes 2}(\Delta(1)),
    \quad
    P^\perp := \Phi^{\otimes 2} (1\otimes 1 - \Delta(1)),
    \quad
    P + P^\perp = \mathbf{1}\otimes\mathbf{1}
\end{equation}
and let $\rho_0\in\mathfrak{L}(V)$ be any mixed state. Define
\begin{equation}\label{eq:rfpS}
    \mathfrak{S}(X) :=
    \mathrm{Tr}(\Phi(\Delta(\xi T(\Omega_{(1)}))) X) \Phi(c_\omega \Omega_{(2)} ) +
    \mathrm{Tr}(P^\perp  X)\rho_0
\end{equation}
for all elements $X\in\mathfrak{L}(V\otimes V)$.
We first check that it satisfies $\mathfrak{S}(\rho(x,2))=\rho(x,1)$ for all positive
non-zero $x\in A$.
Notice that the second summand in the right-hand side of \cref{eq:rfpS} simply vanishes, i.e. $P^\perp  \rho(x,2) = 0$,
since $\rho(x,2)$ is supported on the orthogonal subspace $P\cdot \mathfrak{L}(V^{\otimes 2})$.
Thus,
\begin{align*}
    \mathfrak{S}(\rho(x,2))
    &=
    \tfrac{1}{\langle\omega,x\rangle}\mathrm{Tr}(\Phi^{\otimes 2}(c_\omega^{\otimes 2}\Delta(\xi T(\Omega_{(1)})x)))\Phi(c_\omega \Omega_{(2)} )
    &&\text{by~\cref{def:Cwha}}
    \\ &=
    \tfrac{1}{\langle\omega,x\rangle}\langle\omega\otimes \omega,\Delta(\xi T(\Omega_{(1)}) x)\rangle \Phi(c_\omega \Omega_{(2)} )
    &&\text{by~\cref{remark:phicwext}}
    \\ &=
    \tfrac{1}{\langle\omega,x\rangle} \langle\omega,\xi T(\Omega_{(1)}) x\rangle \Phi(c_\omega \Omega_{(2)} )
    &&\text{by~\cref{lemma:omega}}
    \\ &=
    \tfrac{1}{\langle\omega,x\rangle} \langle\omega,\xi T(\Omega_{(1)})\rangle \Phi(c_\omega x\Omega_{(2)} )
    &&\text{by~\cref{eq:pt}}
    \\ &=
    \tfrac{1}{\langle\omega,x\rangle}\Phi(c_\omega x  ) = \rho(x,1)
    &&\text{by~\cref{lemma:xi}}
\end{align*}
for all positive non-zero elements $x\in A$,
as we wanted to prove. 
Secondly, let us check that it is trace-preserving:
\begin{align*}
    \mathrm{Tr}(\mathfrak{S}(X))
    &=
    \mathrm{Tr}(\Phi^{\otimes 2}(\Delta(\xi T(\Omega_{(1)}))) X)\mathrm{Tr}(\Phi(c_\omega\Omega_{(2)})) + \mathrm{Tr}(P^\perp  X)
    \\ &=
    \mathrm{Tr}(\Phi^{\otimes 2}(\Delta(\xi T(\Omega_{(1)}))) X) \langle\omega,\Omega_{(2)}\rangle + \mathrm{Tr}(P^\perp  X)
        &&\text{by~\cref{remark:phicwext}}
    \\ &=
    \mathrm{Tr}(\Phi^{\otimes 2}(\Delta(\xi T(\xi^{-1}))) X)+ \mathrm{Tr}(P^\perp X)
        &&\text{by~\cref{lemma:xi}}
    \\ &=
    \mathrm{Tr}(\Phi^{\otimes 2}(\Delta(\xi\xi^{-1})) X)  + \mathrm{Tr}(P^\perp  X)
        &&\text{by~\cref{lemma:xi}}
    \\ &=
    \mathrm{Tr}(P  X) + \mathrm{Tr}(P^\perp  X)
    =
    \mathrm{Tr}((P+P^\perp)  X) = \mathrm{Tr}(X)
        &&\text{by~\cref{eq:SPPperp}}
\end{align*}
for all $X\in\mathfrak{L}(V\otimes V)$.
That $\mathfrak{S}$ is completely positive can be proved analogously and we do
not include it here: simply notice that the second summand in
Equationby~\cref{eq:rfpS} is clearly a completely positive map, and a similar
argument to that for $\mathfrak{T}$ applies to the first summand.
\end{proof}

\section{Proof of \texorpdfstring{\cref{lemma:glueHA}}{Lemma~5.2}}
\label{sec:proofs-classif-hopf}

In this appendix we derive a proof of \cref{lemma:glueHA}. We first provide
the following auxiliary result, related to the trace-preserving condition of the gluing map.


\begin{lemma}\label{lemma:omegaMidHA}
Let $A$ be a C*-HA. Then,
\begin{equation*}
    x_{(1)}\otimes \langle \omega,x_{(2)}\rangle x_{(3)}
    =
    \langle\omega,x\rangle 1\otimes 1
\end{equation*}
for all elements $x\in A$.
\end{lemma}


\begin{proof}
Fix any $x\in A$. Since $\Omega\in A$ is non-degenerate, there exists $f\in A^*$ such that
\begin{equation}\label{eq:muxD}
x = \Omega_{(1)}\langle f,\Omega_{(2)}\rangle.
\end{equation}
As an immediate consequence, 
\begin{equation}
    \langle \omega,x\rangle
    =
    \langle \omega,\Omega_{(1)}\rangle\langle f,\Omega_{(2)}\rangle
    =
    \mathfrak{D}^{-2}\langle f,1\rangle,
    \label{eq:muxD2}
\end{equation}
where the last equality follows from \cref{lemma:xi}.
Then, it is easy to conclude that
\begin{align*}
    x_{(1)}\langle \omega,x_{(2)}\rangle \otimes x_{(3)}
    &=
    \Omega_{(1)}\otimes \langle\omega,\Omega_{(2)}\rangle
    \Omega_{(3)}\langle f,\Omega_{(4)}\rangle
    &&\text{by~\cref{eq:muxD}}
    \\ &=
    \Omega_{(4)}\otimes \langle\omega,\Omega_{(1)}\rangle\Omega_{(2)}
    \langle f,\Omega_{(3)}\rangle
    &&\text{by~\cref{thm:Omega}}
    \\ &=
    \mathfrak{D}^{-2} 1_{(3)} \otimes 1_{(1)} \langle f,1_{(2)}\rangle
    &&\text{by~\cref{lemma:xi}}
    \\ &=
    \mathfrak{D}^{-2} \langle f,1\rangle  1 \otimes 1
    &&\text{by~\cref{def:CHA}}
    \\ &=
    \langle \omega,x\rangle 1\otimes 1,
    &&\text{by~\cref{eq:muxD2}}
\end{align*}
as we wanted to prove.
\end{proof}


\restatableglueHA*


\begin{proof}
Fix any positive non-zero element $x\in A$.
We recall first the definition of the gluing map previously given
in \cref{sec:classif}. For simplicity, let
$\mathfrak{G}_x:= \mathfrak{T}\circ \mathfrak{G}$ for the linear map
$\mathfrak{G}\in\mathfrak{L}(\mathfrak{L}(V\otimes V),\mathfrak{L}(V))$ defined
by the expression
\begin{equation}
    \mathfrak{G}(X\otimes Y) := \tfrac{1}{\langle \omega,x\rangle}\mathrm{Tr}(\Phi(S(x_{(1)})) X) \Phi(c_\omega x_{(2)})
     \mathrm{Tr}(\Phi(S(x_{(3)})) Y)
\end{equation}
for all $X,Y\in\mathfrak{L}(V)$.
Then, it is enough to check that
$\mathfrak{G}(\rho(\Omega,2)\otimes\rho(\Omega,2)) = \rho(\Omega,3)$.
To this end, let us recall that, in the case of C*-HAs,
\begin{equation}\label{eq:xi1CHA}
    \langle\omega,\Omega_{(1)}\rangle \Omega_{(2)}
    =
    \tfrac{1}{\mathfrak{D}^{2}} 1
    =
    \langle\omega,\Omega\rangle 1,
\end{equation}
where the first equality is stated in \cref{lemma:xi}
and the second equality follows by applying the counit in the first one,
since $\varepsilon(1) = 1$.
Then,
\begin{align*}
    (\mathrm{Id}{ }& \otimes { } { }\mathfrak{G} { }\otimes{ }
    \mathrm{Id})(\rho(\Omega,2)\otimes\rho(\Omega,2))
    =
    \\ = { } & { } 
    \tfrac{1}{\langle \omega,x\rangle}\tfrac{1}{\langle\omega,\Omega\rangle^2}
    \Phi(c_\omega \Omega_{(1)} ) \otimes
    \langle\omega,S(x_{(1)}\rangle \Omega_{(2)})
    \Phi(c_\omega x_{(2)})
    \langle\omega,S(x_{(3)})\Omega_{(1')}\rangle \otimes
    \Phi(c_\omega\Omega_{(2')})
    \\ = { } & { } 
    \tfrac{1}{\langle \omega,x\rangle}\tfrac{1}{\langle\omega,\Omega\rangle^2}
    \Phi(c_\omega x_{(1)} \Omega_{(1)} ) \otimes
    \langle\omega, \Omega_{(2)}\rangle
    \Phi(c_\omega x_{(2)})
    \langle\omega,S(x_{(3)})\Omega_{(1')}\rangle \otimes
    \Phi(c_\omega  \Omega_{(2')})
    \\ = { } & { }  
    \tfrac{1}{\langle \omega,x\rangle}\tfrac{1}{\langle\omega,\Omega\rangle^2}
    \Phi(c_\omega x_{(1)} \Omega_{(1)} ) \otimes
    \langle\omega, \Omega_{(2)}\rangle
    \Phi(c_\omega x_{(2)})
    \langle\omega,\Omega_{(1')}\rangle \otimes
    \Phi(c_\omega x_{(3)}\Omega_{(2')})
    \\ = { } & { } 
    \tfrac{1}{\langle \omega,x\rangle}
    \Phi(c_\omega x_{(1)} 1 )\otimes\Phi(c_\omega x_{(2)})
    \otimes\Phi(c_\omega x_{(3)}1 )
    =  
    \rho(x,3).
\end{align*}
This calculation can be explained as follows. In the first place, we have
replaced the trace with the canonical regular element $\omega\in A^*$ since by
\cref{remark:phicwext} the weight $c_\omega\in A$, which is central,
defines a linear extension of $\omega$ to the representation space. In the
second and third steps we have applied the pulling-through identity; see
\cref{eq:pt}. Finally, we apply twice \cref{eq:xi1CHA} to get
rid of $\Omega$ and the coefficients $\langle\omega,\Omega\rangle^{-1}$. As an aside, note
that $\langle\omega,\Omega_{(1)}\rangle\Omega_{(2)} = \Omega_{(1)}\langle\omega,\Omega_{(2)}\rangle$ since
$\Omega$ is cocentral; see~\cref{eq:Omegacoce}. Since $\mathfrak{T}$ is
a quantum channel it only remains to prove that $\mathfrak{G}$ is also a quantum
channel. On the one hand, that $\mathfrak{G}$ is trace-preserving is a
straightforward consequence of \cref{lemma:omegaMidHA}:
\begin{align*}
    \mathrm{Tr}(\mathfrak{G}(X\otimes Y))
    &=
    \tfrac{1}{\langle \omega,x\rangle}\mathrm{Tr}(\Phi(S(x_{(1)})) X) \langle\omega, x_{(2)}\rangle\mathrm{Tr}(\Phi(S(x_{(3)})) Y)
    &&\text{by~\cref{remark:phicwext}}
    \\ &=
    \tfrac{1}{\langle \omega,x\rangle}\langle \omega,x\rangle\mathrm{Tr}(\Phi(S(1)) X) \mathrm{Tr}(\Phi(S(1)) Y)
    &&\text{by \cref{lemma:omegaMidHA}}
    \\ &=
    \mathrm{Tr}(X) \mathrm{Tr}(Y) = \mathrm{Tr}(X\otimes Y)
    \vphantom{\tfrac{1}{\langle \omega,x\rangle}}
    &&\text{by~\cref{remark:propsS}}
\end{align*}
for all $X,Y\in \mathfrak{L}(V)$. On the other hand, in order to prove that $\mathfrak{G}$
is completely positive, let
$x = y y^*$ for some element $y\in A$. Then, we can rewrite it as follows
\begin{align*}
    \mathfrak{G}( & X\otimes Y)=
    \\ &=
    \tfrac{1}{\langle \omega,x\rangle}
    \mathrm{Tr}(\Phi(S((yy^*)_{(1)})) X)
    \Phi(c_\omega (y y^*)_{(2)})
    \mathrm{Tr}(\Phi(S(y y^*)_{(3)}) Y)
    \\ & =
    \tfrac{1}{\langle \omega,x\rangle}
    \mathrm{Tr}(\Phi(S(y_{(1)} y_{(1')}^*)) X)
    \Phi(c_\omega y_{(2)} y_{(2')}^*)
    \mathrm{Tr}(\Phi(S( y_{(3)} y_{(3'}^*)) Y)
    \\ &=
    \tfrac{1}{\langle \omega,x\rangle}
    \mathrm{Tr}(\Phi(S(y_{(1')}^*)S(y_{(1)})) X)
    \Phi(c_\omega y_{(2)} y_{(2')}^*)
    \mathrm{Tr}(\Phi(S(y_{(3')}^*) S( y_{(3)} )) Y)
    \\ &=
    \tfrac{1}{\langle \omega,x\rangle}
    \mathrm{Tr}(\Phi(S(y_{(1')})^*S(y_{(1)})) X)
    \Phi(c_\omega y_{(2)} y_{(2')}^*)
    \mathrm{Tr}(\Phi(S(y_{(3')})^* S( y_{(3)} )) Y)
    \\ &= 
    \tfrac{1}{\langle \omega,x\rangle}
    \mathrm{Tr}(\Phi(S(y_{(1)})) X\Phi(S(y_{(1')})^*))
    \Phi(c_\omega y_{(2)} y_{(2')}^*)
    \mathrm{Tr}(\Phi( S( y_{(3)} )) Y\Phi(S(y_{(3')})^*))
    \\ &=
    (\mathrm{Tr}\otimes\mathrm{Id}\otimes\mathrm{Tr})
    (Q (X\otimes\mathbf{1}\otimes Y) Q^\dagger)
\end{align*}
where we have defined
\begin{equation}
    Q := \tfrac{1}{\langle \omega,x\rangle^{1/2}}
    \Phi^{\otimes 4}(S(y_{(1)}) \otimes c_\omega^{\frac{1}{2}} y_{(2)} \otimes \otimes S(y_{(3)})).
\end{equation}
Therefore, $\mathfrak{G}$ is completely positive.
Indeed, in the first step we have applied that the comultiplication is multiplicative
and the $*$-operation is a coalgebra homomorphism; see \cref{def:Cwha}.
In the second and third steps we have used that $S\in\mathfrak{L}(A)$ is an algebra anti-homomorphism
and the relation between the antipode and
the $*$-operation; see \cref{remark:propsS}.
Note that, for C*-HAs,
$S = S^{-1}$; see \cref{prop:someOnHopf}.
The fourth step is a simple consequence of the fact that $\Phi$ is a $*$-representation
and the cyclic property of the trace. Finally, the middle term can be
rewritten in the form
$\Phi(c_\omega y_{(2)} y_{(2')}^*) = \Phi(c_\omega^{1/2} y_{(2)})\Phi(c_\omega^{1/2} y_{(2')})^\dagger$
since $c_\omega\in A$ is positive central element and $\Phi$
is a $*$-representation. 
\end{proof}

\section{Proof of \texorpdfstring{\cref{lemma:glueVacuum}}{Lemma~5.6}}
\label{sec:proofs-classif-weakhopf}

In this appendix we prove \cref{lemma:glueVacuum}. In order to perform
an analogous construction of this gluing map to the one given in the
C*-HA case, we first derive an appropiate version of the usual
pulling-through identity in \cref{eq:pt} to the trivial sector.


\begin{lemma}
\label{lemma:ptvac}
Let $A$ be a biconnected C*-WHA. Then,
\begin{equation*}
    x_L S(1_{(1)}) \otimes 1_{(2)} \otimes S(1_{(3)}) y_R
    =
    S(1_{(1)}) \otimes y_R 1_{(2)} x_L \otimes S(1_{(3)})
\end{equation*}
for all elements $x_L\in A_L$ and $y_R\in A_R$.
\end{lemma}


\begin{proof}
First, recall Equations~2.31a and 2.31b from \cite{bohm_1999_weak}:
\begin{align*}
    x_L S(1_{(1)}) \otimes 1_{(2)}
    &=
    S(1_{(1)}) \otimes 1_{(2)} x_L,
    \\
    y_R 1_{(1)} \otimes  S(1_{(2)})
    &=
    1_{(1)} \otimes S(1_{(2)}) y_R.
\end{align*}
for all $x_L\in A_L$ and $y_R\in A_R$. This, together with \cref{defs:ALARAmin},
leads by taking coproducts accordingly to the following
identities:
\begin{align*}
    &
    x_L S(1_{(1)}) \otimes 1_{(2)} \otimes 1_{(3)}
    =
    S(1_{(1)}) \otimes 1_{(2)} x_L\otimes 1_{(3)},
    \\
    &
    1_{(1)}\otimes y_R 1_{(2)}\otimes S(1_{(3)})
    =
    1_{(1)} \otimes 1_{(2)} \otimes S(1_{(3)}) y_R,
\end{align*}
respectively, for all elements $x_L\in A_L$ and $y_R\in A_R$. Finally,
since $A_L$ and $A_R$ commute, we conclude the result by combining both
identities.
\end{proof}


In addition, we adapt slightly \cref{lemma:xi} to the trivial sector,
which is a key property concerning complete positivity of the gluing map in
\cref{lemma:glueVacuum}. The following result solves this problem.


\begin{lemma}
\label{lemma:CPminGkey}
Let $A$ be a biconnected C*-WHA. Then,
\begin{equation*}
    \xi_R S( x_L^* )
    =
    S( x_L )^* \xi_R
    \;\,\text{ and }\;\,
    S( y_R ) \xi_L
    =
    \xi_L S( y_R^* )^*
\end{equation*}
for all elements $x_L\in A_L$ and $y_R\in A_R$.
\end{lemma}


\begin{proof}
In the first place, note that $T\in\mathfrak{L}(A)$ coincides with $S$ and $S^{-1}$ restricted to
$A_L$ and $A_R$, respectively. Indeed, by virtue of \cref{prop:expDefT},
\cref{prop:arrowsghat} and \cref{remark:propsS},
\begin{subequations}
\begin{align}
    T(x_L)
    & =
    S(x_L 1_{(1)}) \langle\hat{g},1_{(2)}\rangle
    =
    S(x_L),
    \label{eq:TxL01}
    \\
    T(y_R)
    & =
    \langle\hat{g},1_{(1)}\rangle
    S^{-1}(1_{(2)}y_R)
    =
    S^{-1}(y_R)
    =
    S(y_R^*)^*,
    \label{eq:TyR01}
\end{align}
\end{subequations}
for all $x_L\in A_L$ and $y_R\in A_R$.
Then, recall \cref{lemma:xi} to conclude that
\begin{equation*}
    S(x_L^*)
    = 
    T(x_L^*)
    =
    \xi_L^{-1}\xi_R^{-1} T(x_L)^* \xi_L\xi_R
    =
    \xi_L^{-1}\xi_R^{-1} S(x_L)^* \xi_L \xi_R
    =
    \xi_{R}^{-1}S(x_L)^* \xi_R,
\end{equation*}
where in the last step we have used that $S(x_L)\in A_R$ and $A_L$ and $A_R$ commute.
The remaining identity is proved similarly.
\end{proof}


The following auxiliary results arise naturally in the course of the derivation of the
properties of the gluing map.


\begin{lemma}
\label{lemma:hOmega1}
Let $A$ be a bicoconnected C*-WHA. Then,
\begin{equation*}
\langle\hat{h},\Omega_{(1)}\rangle \Omega_{(2)} = \mathfrak{D}^{-2}\varepsilon(1)^{-1} 1.
\end{equation*}
\end{lemma}


\begin{proof}
It is easy to check that
\begin{equation*}
    \mathfrak{D}^2\varepsilon(1)\langle \hat{h},\Omega_{(1)}\rangle\Omega_{(2)}
    =
    \langle \hat{h},t_{(1)}\rangle t_{(2)}\langle \hat{g},t_{(3)}\rangle
    =
    1_{(1)}\langle \hat{g},1_{(2)}\rangle
    =
    1,
\end{equation*}
where the first step is a consequence of the characterization of $\Omega\in A$
in \cref{prop:OmegaIsTau}, the second follows from the definition of dual left integral
in \cref{eq:dualLeftIntHaar} and the third equality is due to \cref{prop:arrowsghat}.
\end{proof}


\begin{lemma}
\label{lemma:trvac1}
Let $A$ be a biconnected C*-WHA. Then,
\begin{equation}
    1_{(1)}\langle\hat{h},1_{(2)}\rangle\otimes 1_{(3)}
    =
    \varepsilon(1)^{-1} 1\otimes 1.
\end{equation}
\end{lemma}


\begin{proof}
Equivalently, we will check that
\begin{equation*}
    \langle\phi \hat{h} \psi,1\rangle
    =
    \varepsilon(1)^{-1} \langle \phi,1 \rangle \langle \psi,1 \rangle
\end{equation*}
for all $\phi,\psi\in A^*$. Recall that ${\hat{h}}\in A^*$
is a one-dimensional projector supported on the trivial sector \cite[Lemma~4.8]{bohm_1999_weak}. Hence,
\begin{equation}
    \label{eq:hath1dproj}
    \langle\phi\hat{h}\psi\hat{h},\mathrm{Tr}^{1}\rangle
    =
    \langle\phi\hat{h},\mathrm{Tr}^1\rangle
    \langle\psi\hat{h},\mathrm{Tr}^1\rangle
    \;\,\text{ and }\;\,
    \langle\phi\hat{h},\mathrm{Tr}^a\rangle
    =
    \delta_{a 1}
\end{equation}
for all $\phi,\psi\in A^*$ and all sectors $a\in\mathrm{Irr}(A^*)$.
In particular
\begin{equation}
    \label{eq:fhath1}
    (f\hat{h})(\mathrm{Tr}^1)
    =
    (f\hat{h})(\sum_{a}\hat{d}_a\mathrm{Tr}^a )
    =
    \mathfrak{D}^2 (f\hat{h})(\Omega)
    =
    \varepsilon(1)^{-1}\langle f,1\rangle
\end{equation}
for all $f\in A^*$.
Thus, we conclude that:
\begin{equation*}
    \varepsilon(1)^{-1} \langle\phi\hat{h}\psi,1\rangle 
    =
    \langle\phi \hat{h}  \psi \hat{h},\mathrm{Tr}^1\rangle
    =
    \langle \phi \hat h , \mathrm{Tr}^1 \rangle 
    \langle \psi \hat h , \mathrm{Tr}^1 \rangle
    =
    \varepsilon(1)^{-2}
    \langle \phi , 1 \rangle \langle \psi , 1 \rangle,
\end{equation*}
where the first equality follows from \cref{eq:fhath1} using
$f := \phi\hat{h}\psi$, the second equality is simply \cref{eq:hath1dproj} and
the third equality follows from \cref{eq:fhath1} considering $f := \phi,\psi$.
\end{proof}


\begin{lemma}\label{lemma:trvac2}
Let $A$ be a biconnected C*-WHA. Then,
\begin{equation*}
    1_{(1)} \otimes \langle\omega,1_{(2)}\rangle 1_{(3)}
    =
    \mathfrak{D}^{2} \xi_R^{-1}\otimes \xi_L^{-1}.
\end{equation*}
\end{lemma}


\begin{proof}
Note by the definition of $A_L$ and $A_R$ in \cref{defs:ALARAmin}
and the decomposition
$\xi^{-1} = \xi_L^{-1}\xi_R^{-1}$ in \cref{lemma:xi}, that
\begin{equation}
\label{eq:coprodxi}
    (\xi^{-1})_{(1)} \otimes (\xi^{-1})_{(2)} \otimes (\xi^{-1})_{(3)}
    = 
    \xi_L^{-1} 1_{(1)} \otimes 1_{(2)}\otimes\xi_R^{-1} 1_{(3)}
\end{equation}
Then, the statement follows from the following calculation:
\begin{align*}
    1_{(1)}  \otimes \langle\omega,1_{(2)}\rangle 1_{(3)}
    & =
    \mathfrak{D}^{2} \varepsilon(1) \langle{\hat{h}},\Omega_{(1)}\rangle\Omega_{(2)}
    \otimes \langle\omega,\Omega_{(3)}\rangle \Omega_{(4)}
    \vphantom{\xi^{-1}}
    &&\text{by~\cref{lemma:hOmega1}}
    \\ & =
    \mathfrak{D}^{2} \varepsilon(1) \langle\hat{h},\Omega_{(3)}\rangle\Omega_{(4)}
    \otimes \langle\omega,\Omega_{(1)}\rangle \Omega_{(2)}
    \vphantom{\xi^{-1}}
    &&\text{by~\cref{eq:Omegacoce}}
    \\ & =
    \mathfrak{D}^{2} \varepsilon(1)\langle\hat{h},(\xi^{-1})_{(2)}\rangle (\xi^{-1})_{(3)}
    \otimes (\xi^{-1})_{(1)}
    &&\text{by~\cref{lemma:xi}}
    \\ & =
    \mathfrak{D}^{2} \varepsilon(1) \langle\hat{h},1_{(2)}\rangle \xi_R^{-1} 1_{(3)}
    \otimes \xi_L^{-1}  1_{(1)}
    &&\text{by~\cref{eq:coprodxi}}
    \\ & =
    \mathfrak{D}^{2} \xi_{R}^{-1}\otimes\xi_L^{-1}
    &&\text{by~\cref{lemma:trvac1}}
\end{align*}
as we wanted to prove.
\end{proof}


\begin{lemma}
Let $A$ be a biconnected C*-WHA. Then,
\begin{equation}
    \label{eq:omega1omega}
    \langle\omega,1_{(1)}\rangle 1_{(2)} \langle\omega,1_{(3)}\rangle
    =
    \mathfrak{D}^2 \langle\omega,1\rangle \xi^{-1}.
\end{equation}
\end{lemma}
\begin{proof}
First, it will be useful to compute the constant $\langle\omega,1\rangle$ in a more operative way.
The following calculation is a direct consequence of \cref{prop:OmegaIsTau}
and \cref{eq:defExpXiLXiR}:
\begin{equation}
    \langle\omega,1\rangle
    =
    \tfrac{1}{\mathfrak{D}^2\varepsilon(1)} \langle\hat{h},g_L^{-1} g_R^{-1}\rangle
    =
    \tfrac{\mathfrak{D}^4 \varepsilon(1)^2}{\mathfrak{D}^2 \varepsilon(1)} 
     \langle\hat{h},\xi_L^{-1}\xi_R^{-1}\rangle
    =
    \mathfrak{D}^2 \varepsilon(1) \langle\hat{h},\xi_L^{-1}\xi_R^{-1}\rangle.
    \label{eq:calcw1}
\end{equation}
Now, by an analogous reasoning as in the previous proof:
\begin{align*}
    \langle\omega,1_{(1)}\rangle 1_{(2)} \langle\omega , 1_{(3)}\rangle
    & =
    \mathfrak{D}^2 \varepsilon(1) \langle\hat{h},\Omega_{(1)}\rangle\langle\omega,\Omega_{(2)}\rangle\Omega_{(3)}\langle\omega,\Omega_{(4)}\rangle
    \vphantom{\xi_L^{-1}}
    &&\text{by~\cref{lemma:hOmega1}}
    \\ & =
    \mathfrak{D}^2 \varepsilon(1)\langle\hat{h},\Omega_{(4)}\rangle \langle\omega,\Omega_{(1)}\rangle\Omega_{(2)}\langle\omega,\Omega_{(3)}\rangle
    \vphantom{\xi_L^{-1}}
    &&\text{by~\cref{eq:Omegacoce}}
    \\ & =
    \mathfrak{D}^2\varepsilon(1) \langle\hat{h},\xi_R^{-1}1_{(3)}\rangle \xi_L^{-1} 1_{(1)} \langle\omega,1_{(2)}\rangle
    &&\text{by~\cref{lemma:xi}}
    \\ & =
    \mathfrak{D}^4 \varepsilon(1) \langle\hat{h},\xi_R^{-1}\xi_L^{-1}\rangle \xi_L^{-1}\xi_{R}^{-1}
    &&\text{by~\cref{lemma:trvac2}}
    \\ & =
    \mathfrak{D}^2 \langle\omega,1\rangle \xi_L^{-1}\xi_{R}^{-1} = \mathfrak{D}^2\langle\omega,1\rangle\xi^{-1}
    \vphantom{\xi_L^{-1}}
    &&\text{by~\cref{eq:calcw1}}
\end{align*}
as we wanted to prove.
\end{proof}


\restatablenogogluing*


\begin{proof}
Suppose by contradiction that there exists a trace-preserving linear map
$\mathfrak{G}\in\mathfrak{L}(\mathfrak{L}(V\otimes V))$ that is a
``gluing map''. In particular,
\begin{equation*}
    (\mathrm{Id} \otimes \mathfrak{G} \otimes \mathrm{Id})
    (\rho(\Omega,2) \otimes \rho(\Omega,2))
        =
    \rho(\Omega,4).
\end{equation*}
On the one hand, after performing a partial trace on the second and third
subsystems, the left-hand side would be trivially given by the product state
\begin{align*}
    (\mathrm{Id}\otimes\mathrm{Tr}\otimes\mathrm{Tr}\otimes\mathrm{Id})(\rho(\Omega,2)\otimes \rho(\Omega,2))
    &=
    \langle\omega,\Omega\rangle^{-2} \Phi(c_\omega\Omega_{(1)})\langle\omega,\Omega_{(2)}\rangle
    \otimes \langle\omega,\Omega_{(1')}\rangle\Phi(c_\omega\Omega_{(2')})
    \\ &=
    \langle\omega,\Omega\rangle^{-2} \Phi(c_\omega \xi^{-1})\otimes\Phi(c_\omega \xi^{-1})
\end{align*}
by virtue of \cref{remark:phicwext} and  \cref{lemma:xi}.
However, the right-hand side would take the following form:
\begin{align*}
    (\mathrm{Id}\otimes\mathrm{Tr}\otimes\mathrm{Tr}\otimes\mathrm{Id})(\rho(\Omega,4))
    & =
    \langle\omega,\Omega\rangle^{-2}\Phi(c_\omega \Omega_{(1)})  \otimes
    \langle\omega,\Omega_{(2)}\rangle\langle\omega,\Omega_{(3)}\rangle\Phi(c_\omega \Omega_{(4)})
    &&\text{by~\cref{remark:phicwext}}
    \\ & =
    \langle\omega,\Omega\rangle^{-2}\Phi(c_\omega \Omega_{(1)})  \otimes
    \langle\omega,\Omega_{(2)}\rangle\Phi(c_\omega \Omega_{(3)})
    &&\text{by~\cref{lemma:omega}}
    \\ & =
    \langle\omega,\Omega\rangle^{-2}\Phi(c_\omega \Omega_{(2)})  \otimes
    \langle\omega,\Omega_{(3)}\rangle \Phi(c_\omega \Omega_{(1)})
    &&\text{by~\cref{eq:Omegacoce}}
    \\ & =
    \langle\omega,\Omega\rangle^{-2}\Phi(c_\omega \xi_R^{-1}1_{(2)})\otimes
    \Phi(c_\omega \xi_L^{-1}1_{(1)})
    &&\text{by~\cref{lemma:xi}}
\end{align*}
which is not a product state. This contradicts the previous equation.
\end{proof}


\restatableGlueWHA*


\begin{proof}
For simplicity, let $\mathfrak{G}_1:= \mathfrak{T}\circ\mathfrak{G}$, where $\mathfrak{T}:\mathfrak{L}(V)\to\mathfrak{L}(V\otimes V)$ stands for the local coarse-graining quantum channel from \cref{sec:rfp} and
$\mathfrak{G}:\mathfrak{L}(V\otimes V)\to\mathfrak{L}(V)$ is given by
\begin{equation*}
\mathfrak{G}(X\otimes Y) :=
\tfrac{1}{\mathfrak{D}^2} \mathrm{Tr}(\Phi(S(1_{(1)})\xi_L)  X)
\Phi(c_\omega 1_{(2)} )
 \mathrm{Tr}(\Phi(\xi_R S(1_{(3)}))  Y)
\end{equation*}
for all $X,Y\in\mathfrak{L}(V)$.
First, assume that $m = n = 2$ without loss of generality and let us check
that it fulfills $\mathfrak{G}(\rho(1,2)\otimes\rho(1,2)) = \rho(1,3)$.
To this end, it turns out to be enough to prove:
\begin{equation}\label{eq:suffCondGWHA}
    \mathfrak{G}(\Phi(c_\omega x_L )\otimes \Phi(c_\omega x_R)) = \langle\omega,1\rangle \Phi(c_\omega x_L x_R)
\end{equation}
for all $x_L\in A_L$ and $x_R\in A_R$.
Indeed, in that case,
\begin{align*}
    (\mathrm{Id}\otimes \mathfrak{G}\otimes\mathrm{Id})(\rho(1,2)^{\otimes 2})
    & =
    \tfrac{1}{\langle\omega,1\rangle^2} \Phi(c_\omega 1_{(1)}) \otimes
    \mathfrak{G}(\Phi(c_\omega 1_{(2)})\otimes \Phi(c_\omega 1_{(1')})) \otimes\Phi(c_\omega 1_{(2')})
    \\ & =
    \tfrac{1}{\langle\omega,1\rangle} \Phi(c_\omega 1_{(1)}) \otimes
    \Phi(c_\omega 1_{(2)} 1_{(1')}) \otimes\Phi(c_\omega 1_{(2')})
    \\ & =
    \tfrac{1}{\langle\omega,1\rangle} \Phi(c_\omega 1_{(1)}) \otimes
    \Phi(c_\omega 1_{(2)}) \otimes\Phi(c_\omega 1_{(3)})
    = \rho(1,3).
\end{align*}
by the weak comultiplicativity of the counit and the fact that
$1_{(1)}\otimes 1_{(2)} \in A_R\otimes A_L$;
see \cref{def:Cwha} and \cite{bohm_1999_weak}.
Thus, let us move to the proof of \cref{eq:suffCondGWHA}:
\begin{align*}
    \mathfrak{G}(\Phi(c_\omega x_L)\otimes \Phi(c_\omega x_R))
    & =
    \tfrac{1}{\mathfrak{D}^2} \langle\omega,S(1_{(1)})\xi_L x_L\rangle
    \Phi(c_\omega 1_{(2)} )
    \langle\omega,\xi_R S(1_{(3)})x_R\rangle
    \vphantom{\xi_L^{-1}\tfrac{1}{\mathfrak{D}^2}}
    &&\text{by~\cref{remark:phicwext}}
    \\ &=
    \tfrac{1}{\mathfrak{D}^2} \langle\omega,\xi_L x_L S(1_{(1)})\rangle
    \Phi(c_\omega 1_{(2)} )
    \langle\omega, S(1_{(3)})x_R\xi_R\rangle
    \vphantom{\xi_L^{-1}\tfrac{1}{\mathfrak{D}^2}}
    \\ &=
    \tfrac{1}{\mathfrak{D}^2} \langle\omega,S(1_{(1)})\rangle
    \Phi(c_\omega x_R \xi_R 1_{(2)}\xi_L x_L )
    \langle\omega, S(1_{(3)}) \rangle
    \vphantom{\xi_L^{-1}\tfrac{1}{\mathfrak{D}^2}}
    &&\text{by~\cref{lemma:ptvac}}
    \\ &=
    \tfrac{1}{\mathfrak{D}^2} \langle\omega , 1_{(1)}\rangle
    \Phi(c_\omega  x_R\xi_R1_{(2)}\xi_L x_L )
    \langle\omega, 1_{(3)} \rangle
    \vphantom{\xi_L^{-1}\tfrac{1}{\mathfrak{D}^2}}
    &&\text{by~\cref{remark:wSwwT}}
    \\ &=
    \langle\omega,1\rangle\Phi(c_\omega x_R\xi_R \xi_R^{-1}\xi_L^{-1} \xi_L x_L )
    &&\text{by~\cref{eq:omega1omega}}
    \\ &=
    \langle\omega,1\rangle\Phi(c_\omega x_R x_L )
    \vphantom{\xi_L^{-1}\tfrac{1}{\mathfrak{D}^2}}
\end{align*}
as we wanted to prove. Additionally, $\mathfrak{G}$ is trace-preserving
as an immediate consequence of \cref{lemma:trvac2}:
\begin{align*}
    \mathrm{Tr}(\mathfrak{G}(X\otimes Y))
    &=
    \tfrac{1}{\mathfrak{D}^2} \mathrm{Tr}(\Phi(S(1_{(1)})\xi_L)  X)
    \langle\omega, 1_{(2)} \rangle
    \mathrm{Tr}(\Phi(\xi_RS(1_{(3)}))  Y)
    &&\text{by~\cref{remark:phicwext}}
    \\ &=
    \mathrm{Tr}(\Phi(S(  \xi_R^{-1}  )\xi_L)  X)
    \mathrm{Tr}(\Phi(\xi_R S(  \xi_L^{-1}  ))  Y)
    &&\text{by~\cref{lemma:trvac2}}
    \\ &=
    \mathrm{Tr}(\Phi(\xi_L^{-1}\xi_L)  X)
    \mathrm{Tr}(\Phi(\xi_R \xi_R^{-1})  Y)
    &&\text{by~\cref{eq:defExpXiLXiR}}
    \\ &=
    \mathrm{Tr}(X)\mathrm{Tr}(Y) = \mathrm{Tr}(X\otimes Y).
\end{align*}
Finally, in order to check that $\mathfrak{G}$ is a completely positive linear
map, let us first consider the following two calculations:
\begin{align*}
    \mathrm{Tr}(\Phi(S(x_R y_R^*)\xi_L)  X)
    \vphantom{\xi_L^{\frac{1}{2}}}
     & =
    \mathrm{Tr}(\Phi(S(y_R^*) S(x_R)\xi_L)  X)
    &&\text{by~\cref{remark:propsS}}
    \\ &=
    \mathrm{Tr}(\Phi(S(y_R^*) \xi_L S(x_R^*)^*)  X)
    \vphantom{\xi_L^{\frac{1}{2}}}
    &&\text{by~\cref{lemma:CPminGkey}}
    \\ &= 
    \mathrm{Tr}(\Phi(S(y_R^*) \xi_L^{\frac{1}{2}} \xi_L^{\frac{1}{2}} S(x_R^*)^*)  X)
    &&\text{by~\cref{eq:defExpXiLXiR}}
    \\ & =
    \mathrm{Tr}(\Phi( S(y_R^*) \xi_L^{\frac{1}{2}}) \Phi(\xi_L^{\frac{1}{2}} S(x_R^*)^*)  X)
    \\ & =
    \mathrm{Tr}( \Phi(\xi_L^{\frac{1}{2}} S(x_R^*)^*)  X  \Phi( S(y_R^*) \xi_L^{\frac{1}{2}}))
    \\ & =
    \mathrm{Tr}(\Phi(\xi_L^{\frac{1}{2}} S(x_R^*)^*)  X  \Phi( (\xi_L^{\frac{1}{2}}S(y_R^*))^* ))
    &&\text{by~\cref{eq:defExpXiLXiR}}
    \\ & =
    \mathrm{Tr}(\Phi(\xi_L^{\frac{1}{2}} S(x_R^*)^*)  X  \Phi( \xi_L^{\frac{1}{2}} S(y_R^*)^* )^\dagger )
    \\
\intertext{for all $x_R,y_R\in A_R$ and, analogously,}
    \mathrm{Tr}(\Phi(\xi_R S(x_L y_L^*))  Y)
    \vphantom{\xi_R^{\frac{1}{2}}}
     & =
    \mathrm{Tr}(\Phi(\xi_R S(y_L^*)S(x_L))  Y)
    &&\text{by~\cref{remark:propsS}}
    \\ & =
    \mathrm{Tr}(\Phi( S(y_L)^* \xi_R S(x_L))  Y)
    \vphantom{\xi_R^{\frac{1}{2}}}
    &&\text{by~\cref{lemma:CPminGkey}}
    \\ & =
    \mathrm{Tr}(\Phi( S(y_L)^* \xi_R^{\frac{1}{2}} \xi_R^{\frac{1}{2}} S(x_L) )  Y)
    &&\text{by~\cref{eq:defExpXiLXiR}}
    \\ & =
    \mathrm{Tr}(\Phi( S(y_L)^* \xi_R^{\frac{1}{2}}) \Phi(\xi_R^{\frac{1}{2}} S(x_L))  Y)
    \\ & =
    \mathrm{Tr}(\Phi(\xi_R^{\frac{1}{2}} S(x_L))  Y  \Phi( S(y_L)^* \xi_R^{\frac{1}{2}}))
    \\ & =
    \mathrm{Tr}(\Phi(\xi_R^{\frac{1}{2}} S(x_L))  Y  \Phi((\xi_R^{\frac{1}{2}} S(y_L))^*))
    &&\text{by~\cref{eq:defExpXiLXiR}}
    \\ &=
    \mathrm{Tr}(\Phi(\xi_R^{\frac{1}{2}} S(x_L))  Y  \Phi( \xi_R^{\frac{1}{2}} S(y_L))^\dagger)
\end{align*}
for all $x_L, y_L\in A_L$.
Now, recall that $1_{(1)}\otimes 1_{(2)}\otimes 1_{(3)}\in A_R\otimes A\otimes A_L$;
see~\cite{bohm_1999_weak}. This allows
us to rewrite $\mathfrak{G}$ in the following form:
\begin{align*}
    \mathfrak{G}(X\otimes Y)
    &=
    \tfrac{1}{\mathfrak{D}^2}\mathrm{Tr}(\Phi(S((1{\cdot}1^*)_{(1)})\xi_L)   X)
    \Phi(c_\omega (1 {\cdot} 1^*)_{(2)} )\mathrm{Tr}(\Phi(\xi_R S((1 {\cdot} 1^*)_{(3)}))   Y)
    \\ &=
    \tfrac{1}{\mathfrak{D}^2}\mathrm{Tr}(\Phi(S(1_{(1)}(1^*)_{(1')})\xi_L)   X)
    \Phi(c_\omega 1_{(2)}(1^*)_{(2')})\mathrm{Tr}(\Phi(\xi_R S(1_{(3)}1_{(3')}))  Y)
    \\ &=
    \tfrac{1}{\mathfrak{D}^2}\mathrm{Tr}(\Phi(S(1_{(1)}1_{(1')}^*)\xi_L)   X)
    \Phi(c_\omega 1_{(2)}1_{(2')}^*)\mathrm{Tr}(\Phi(\xi_R S(1_{(3)}1_{(3')}^*))  Y)
    \\ &=
    (\mathrm{Tr}\otimes\mathrm{Id}\otimes\mathrm{Id})(Q  (X\otimes\mathbf{1}\otimes Y)  Q^\dagger)
\end{align*}
where the last step follows from the previous calculations, and we have defined
\begin{equation}\label{eq:defQgCWHA}
    Q:= \tfrac{1}{\mathfrak{D}} \Phi^{\otimes 3}(\xi_L^{\frac{1}{2}}S(1_{(1)}^*)
     \otimes c_\omega^{\frac{1}{2}} 1_{(2)}\otimes \xi_R^{\frac{1}{2}}S(1_{(3)})).
\end{equation}
This concludes the proof.
\end{proof}




\begin{thebibliography}{99}

\bibitem{anshu_2021_area}
Anshu, A., Arad, I., Gosset, D.:
An area law for 2D frustration-free spin systems. arXiv:2103.02492 [quant-ph] (2021).
\href{https://doi.org/10.48550/arXiv.2103.02492}{doi.org/10.48550/arXiv.2103.02492}

\bibitem{bachmann_2012_automorphic}
Bachmann, S., Michalakis, S., Nachtergaele, B., Sims, R.:
Automorphic Equivalence within Gapped Phases of Quantum Lattice Systems.
Commun. Math. Phys. 309, 835-871 (2012).
\href{https://doi.org/10.1007/s00220-011-1380-0}{doi:10.1007/s00220-011-1380-0}

\bibitem{bardyn_2012_modes}
Bardyn, C.-E., Baranov, M.A., Rico, E., {\.I}mamo{\u{g}}lu, A., Zoller, P., Diehl, S.:
Majorana Modes in Driven-Dissipative Atomic Superfluids with a Zero Chern Number.
Phys. Rev. Lett. 109, 130402 (2012).
\href{https://doi.org/10.1103/PhysRevLett.109.130402}{doi:10.1103/PhysRevLett.109.130402}

\bibitem{bardyn_2013_topology}
Bardyn, C.-E., Baranov, M.A., Kraus, C.V., Rico, E., {\.I}mamo{\u{g}}lu, A., Zoller, P., Diehl, S.:
Topology by dissipation.
New J. Phys. 15, 085001 (2013).
\href{https://doi.org/10.1088/1367-2630/15/8/085001}{doi:10.1088/1367-2630/15/8/085001}

\bibitem{bohm_1996_coassociative}
B{\"o}hm, G., Szlach{\'a}nyi, K.:
A coassociative C*-quantum group with nonintegral dimensions.
Lett Math Phys. 38, 437-456 (1996).
\href{https://doi.org/10.1007/BF01815526}{doi:10.1007/BF01815526}

\bibitem{bohm_1999_weak}
B{\"o}hm, G., Nill, F., Szlach{\'a}nyi, K.:
Weak Hopf Algebras: I. Integral Theory and C*-Structure.
Journal of Algebra. 221, 385-438 (1999).
\href{https://doi.org/10.1006/jabr.1999.7984}{doi:10.1006/jabr.1999.7984}

\bibitem{bohm_2000_weak2}
B{\"o}hm, G., Szlach{\'a}nyi, K.:
Weak Hopf Algebras: II. Representation Theory, Dimensions, and the Markov Trace.
Journal of Algebra. 233, 156-212 (2000).
\href{https://doi.org/10.1006/jabr.2000.8379}{doi:10.1006/jabr.2000.8379}

\bibitem{bravyi_2010_topological}
Bravyi, S., Hastings, M.B., Michalakis, S.:
Topological quantum order: Stability under local perturbations.
J. Math. Phys. 51, 093512 (2010).
\href{https://doi.org/10.1063/1.3490195}{doi:10.1063/1.3490195}

\bibitem{brandao_2015_area}
Brand{\~a}o, F.G.S.L., Cubitt, T.S., Lucia, A., Michalakis, S., P{\'e}rez-Garc{\'i}a, D.:
Area law for fixed points of rapidly mixing dissipative quantum systems.
J. Math. Phys. 56, 102202 (2015).
\href{https://doi.org/10.1063/1.4932612}{doi:10.1063/1.4932612}

\bibitem{bultinck_2017_anyons}
Bultinck, N., Mari{\"e}n, M., Williamson, D.J., Şahino{\v{g}}lu, M.B., Haegeman, J., Verstraete, F.:
Anyons and matrix product operator algebras.
Annals of Physics. 378, 183-233 (2017).
\href{https://doi.org/10.1016/j.aop.2017.01.004}{doi:10.1016/j.aop.2017.01.004}

\bibitem{chen_2010_local}
Chen, X., Gu, Z.-C., Wen, X.-G.:
Local unitary transformation, long-range quantum entanglement, wave function renormalization, and topological order.
Phys. Rev. B. 82, 155138 (2010).
\href{https://doi.org/10.1103/PhysRevB.82.155138}{doi:10.1103/PhysRevB.82.155138}

\bibitem{chen_2011_classification}
Chen, X., Gu, Z.-C., Wen, X.-G.:
Classification of gapped symmetric phases in one-dimensional spin systems.
Phys. Rev. B. 83, 035107 (2011).
\href{https://doi.org/10.1103/PhysRevB.83.035107}{doi:10.1103/PhysRevB.83.035107}

\bibitem{cirac_2011_entanglement}
Cirac, J.I., Poilblanc, D., Schuch, N., Verstraete, F.:
Entanglement spectrum and boundary theories with projected entangled-pair states.
Phys. Rev. B. 83, 245134 (2011).
\href{https://doi.org/10.1103/PhysRevB.83.245134}{doi:10.1103/PhysRevB.83.245134}

\bibitem{cirac_2017_mpdo}
Cirac, J.I., P{\'e}rez-Garc{\'i}a, D., Schuch, N., Verstraete, F.:
Matrix product density operators: Renormalization fixed points and boundary theories.
Annals of Physics. 378, 100-149 (2017).
\href{https://doi.org/10.1016/j.aop.2016.12.030}{doi:10.1016/j.aop.2016.12.030}

\bibitem{cirac_2021_matrix}
Cirac, J.I., P{\'e}rez-Garc{\'i}a, D., Schuch, N., Verstraete, F.:
Matrix product states and projected entangled pair states: Concepts, symmetries, theorems.
Rev. Mod. Phys. 93, 045003 (2021).
\href{https://doi.org/10.1103/RevModPhys.93.045003}{doi:10.1103/RevModPhys.93.045003}

\bibitem{coser_2019_classification}
Coser, A., P{\'e}rez-Garc{\'i}a, D.:
Classification of phases for mixed states via fast dissipative evolution.
Quantum. 3, 174 (2019).
\href{https://doi.org/10.22331/q-2019-08-12-174}{doi:10.22331/q-2019-08-12-174}

\bibitem{diehl_2011_topology}
Diehl, S., Rico, E., Baranov, M.A., Zoller, P.:
Topology by dissipation in atomic quantum wires.
Nature Phys. 7, 971-977 (2011).
\href{https://doi.org/10.1038/nphys2106}{doi:10.1038/nphys2106}

\bibitem{etingof_2005_on_fusion}
Etingof, P., Nikshych, D., Ostrik, V.:
On fusion categories.
Ann. Math. 162, 581-642 (2005).
\href{https://doi.org/10.4007/annals.2005.162.581}{doi:10.4007/annals.2005.162.581}

\bibitem{etingof_2012_descent}
Etingof, P., Gelaki, S.:
Descent and Forms of Tensor Categories.
International Mathematics Research Notices. 2012, 3040-3063 (2012).
\href{https://doi.org/10.1093/imrn/rnr119}{doi:10.1093/imrn/rnr119}

\bibitem{etingof_2015_tensor}
Etingof, P., Gelaki, S., Nikshych, D., Ostrik, V.:
Tensor Categories.
American Mathematical Society, Providence, Rhode Island (2015).

\bibitem{freedman_2002_topological}
Freedman, M., Kitaev, A., Larsen, M., Wang, Z.:
Topological quantum computation.
Bull. Amer. Math. Soc. 40, 31-38 (2002).
\href{https://doi.org/10.1090/S0273-0979-02-00964-3}{doi:10.1090/S0273-0979-02-00964-3}

\bibitem{grusdt_2017_topological}
Grusdt, F.:
Topological order of mixed states in correlated quantum many-body systems.
Phys. Rev. B. 95, 075106 (2017).
\href{https://doi.org/10.1103/PhysRevB.95.075106}{doi:10.1103/PhysRevB.95.075106}

\bibitem{hastings_2005_quasiadiabatic}
Hastings, M.B., Wen, X.-G.:
Quasiadiabatic continuation of quantum states: The stability of topological ground-state degeneracy and emergent gauge invariance.
Phys. Rev. B. 72, 045141 (2005).
\href{https://doi.org/10.1103/PhysRevB.72.045141}{doi:10.1103/PhysRevB.72.045141}

\bibitem{hastings_2007_area_law}
Hastings, M.B.:
An area law for one-dimensional quantum systems.
J. Stat. Mech. 2007, P08024-P08024 (2007).
\href{https://doi.org/10.1088/1742-5468/2007/08/P08024}{doi:10.1088/1742-5468/2007/08/P08024}

\bibitem{kac_1966_finite}
Kac, G.I., Paljutkin, V.G.: Finite ring groups. Trans. Moscow Math Soc., 251-294 (1966).

\bibitem{kast_2019_local} Kastoryano, M.J., Lucia, A., P{\'e}rez-Garc{\'i}a, D.:
Locality at the Boundary Implies Gap in the Bulk for 2D PEPS.
Commun. Math. Phys. 366, 895-926 (2019).
\href{https://doi.org/10.1007/s00220-019-03404-9}{doi:10.1007/s00220-019-03404-9}

\bibitem{kitaev_2003_anyons}
Kitaev, A.Yu.:
Fault-tolerant quantum computation by anyons.
Annals of Physics. 303, 2-30 (2003).
\href{https://doi.org/10.1016/S0003-4916(02)00018-0}{doi:10.1016/S0003-4916(02)00018-0}

\bibitem{konig_2014_generating}
K{\"{o}}nig, R., Pastawski, F.:
Generating topological order: No speedup by dissipation.
Phys. Rev. B. 90, 045101 (2014).
\href{https://doi.org/10.1103/PhysRevB.90.045101}{doi:10.1103/PhysRevB.90.045101}

\bibitem{levin_2005_stringnet}
Levin, M.A., Wen, X.-G.:
String-net condensation: A physical mechanism for topological phases.
Phys. Rev. B. 71, 045110 (2005).
\href{https://doi.org/10.1103/PhysRevB.71.045110}{doi:10.1103/PhysRevB.71.045110}

\bibitem{li_2008_entanglement}
Li, H., Haldane, F.D.M.:
Entanglement Spectrum as a Generalization of Entanglement Entropy: Identification of Topological Order in Non-Abelian Fractional Quantum Hall Effect States.
Phys. Rev. Lett. 101, 010504 (2008).
\href{https://doi.org/10.1103/PhysRevLett.101.010504}{doi:10.1103/PhysRevLett.101.010504}

\bibitem{lieb_1972_finite}
Lieb, E.H., Robinson, D.W.: The finite group velocity of quantum spin systems.
Commun.Math. Phys. 28, 251-257 (1972).
\href{https://doi.org/10.1007/BF01645779}{doi:10.1007/BF01645779}

\bibitem{longo_2001_math}
Longo, R. ed:
Mathematical Physics in Mathematics and Physics.
American Mathematical Society, Providence, Rhode Island (2001)

\bibitem{molnar_2022_mpo}
Moln{\'a}r, A., Ruiz de Alarc{\'o}n, A., Garre-Rubio, J., Schuch, N., Cirac, J.I., P{\'e}rez-Garc{\'i}a, D.:
Matrix product operator algebras I: representations of weak Hopf algebras and projected entangled pair states.
arXiv:2204.05940 (2022).

\bibitem{montgomery_2001_rep}
Montgomery, S.:
Representation Theory of Semisimple Hopf Algebras.
En: Roggenkamp, I.K.W. y {\c{S}}tef{\u{a}}nescu, M. (eds.) Algebra - Representation Theory. pp. 189-218. Springer Netherlands, Dordrecht (2001)

\bibitem{nikshych_2004_semisimple}
Nikshych, D.:
Semisimple weak Hopf algebras.
Journal of Algebra. 275, 639-667 (2004).
\href{https://doi.org/10.1016/j.jalgebra.2003.09.025}{doi:10.1016/j.jalgebra.2003.09.025}

\bibitem{nill_1998_axioms}
Nill, F.: Axioms for Weak Bialgebras. arXiv:math/9805104 [math.QA]. (1998).
\href{https://doi.org/10.48550/arXiv.math/9805104}{doi:10.48550/arXiv.math/9805104}

\bibitem{ogata_2021_classif}
Ogata, Y.:
A classification of pure states on quantum spin chains satisfying the split property with on-site finite group symmetries.
Trans. Amer. Math. Soc. Ser. B. 8, 39-65 (2021).
\href{https://doi.org/10.1090/btran/51}{doi:10.1090/btran/51}

\bibitem{perez_2020_estimates}
P{\'e}rez-Garc{\'i}a, D., P{\'e}rez-Hern{\'a}ndez, A.:
Locality estimates for complex time evolution in 1D.
arXiv:2004.10516 [math-ph] (2020).
\href{https://doi.org/10.48550/arXiv.2004.10516}{doi:10.48550/arXiv.2004.10516}

\bibitem{larson_1988_semisimple}
Larson, R.G., Radford, D.E.:
Semisimple Cosemisimple Hopf Algebras.
American Journal of Mathematics. 110, 187 (1988).
\href{https://doi.org/10.2307/2374545}{doi:10.2307/2374545}

\bibitem{larson_1988_finite}
Larson, R.G., Radford, D.E.:
Finite dimensional cosemisimple Hopf algebras in characteristic 0 are semisimple.
Journal of Algebra. 117, 267-289 (1988).
\href{https://doi.org/10.1016/0021-8693(88)90107-X}{doi:10.1016/0021-8693(88)90107-X}

\bibitem{sahinoglu_2021_char}
{\c{S}}ahino{\v{g}}lu, M.B., Williamson, D., Bultinck, N., Mari{\"e}n, M., Haegeman, J., Schuch, N., Verstraete, F.:
Characterizing Topological Order with Matrix Product Operators.
Ann. Henri Poincaré. 22, 563-592 (2021).
\href{https://doi.org/10.1007/s00023-020-00992-4}{doi:10.1007/s00023-020-00992-4}

\bibitem{schuch_2011_classif}
Schuch, N., P{\'e}rez-Garc{\'i}a, D., Cirac, J.I.:
Classifying quantum phases using matrix product states and projected entangled pair states.
Phys. Rev. B. 84, 165139 (2011).
\href{https://doi.org/10.1103/PhysRevB.84.165139}{doi:10.1103/PhysRevB.84.165139}

\bibitem{schuch_2013_topological}
Schuch, N., Poilblanc, D., Cirac, J.I., P{\'e}rez-Garc{\'i}a, D.:
Topological Order in the Projected Entangled-Pair States Formalism: Transfer Operator and Boundary Hamiltonians.
Phys. Rev. Lett. 111, 090501 (2013).
\href{https://doi.org/10.1103/PhysRevLett.111.090501}{doi:10.1103/PhysRevLett.111.090501}

\bibitem{wolf_2008_area}
Wolf, M.M., Verstraete, F., Hastings, M.B., Cirac, J.I.:
Area Laws in Quantum Systems: Mutual Information and Correlations.
Phys. Rev. Lett. 100, 070502 (2008).
\href{https://doi.org/10.1103/PhysRevLett.100.070502}{doi:10.1103/PhysRevLett.100.070502}

\end{thebibliography}
\end{document}